\PassOptionsToPackage{table,dvipsnames}{xcolor}

\documentclass[sigconf,authorversion,nonacm]{acmart}

\AtBeginDocument{%
  \providecommand\BibTeX{{%
    \normalfont B\kern-0.5em{\scshape i\kern-0.25em b}\kern-0.8em\TeX}}}


\providecommand{\eat}[1]{}

\usepackage{subcaption}
\usepackage{cleveref}
\usepackage{multicol}
\usepackage{multirow}
\usepackage[table,dvipsnames]{xcolor}
\usepackage{array}

\usepackage[normalem]{ulem} 
\usepackage[linesnumbered,algoruled, lined, noend]{algorithm2e}
\usepackage{enumitem}
\usepackage{bbm} 

\SetCommentSty{mycommfont}

\newcommand{\introparagraph}[1]{\medskip \noindent {\bf  #1.}}  

\newcommand{\TODO}[1]{\textcolor{red}{#1}}
\newcommand{\Ling}[1]{\textcolor{black}{#1}}
\newcommand{\jmp}[1]{\textcolor{blue}{#1}}
\newcommand{\sd}[1]{\textcolor{brown}{SD: #1}}
\newcommand{\note}[1]{\textcolor{purple}{#1}}

\newcommand{\literals}{\texttt{literals}}

\newcommand{\blare}{\texttt{BLARE}} 
\newcommand{\sys}{\texttt{REI}} 

\newcommand{\wlsqlfull}{\texttt{Database X}}
\newcommand{\wltraffic}{\texttt{US-Accident}}
\newcommand{\wlkusto}{\texttt{System Y}}

\newcommand{\baseline}{\texttt{BLARE-RE2}}

\Crefname{algocf}{Algorithm}{Algorithms}
\newtheorem{observation}{Observation}[section]
\newcommand{\eidx}{\textsf{EnglishIndex}}
\newcommand{\qidx}{\textsf{QueryIndex}}

\newcommand{\biVsTri}{in~\Cref{app:bigram_trigram}}
\newcommand{\correctProof}{in~\Cref{app:correct}}
\newcommand{\tuner}{in~\Cref{app:rei_tuner}}

\begin{document}

\title{Regular Expression Indexing for Log Analysis. Extended Version}

\author{Ling Zhang}
\orcid{0009-0003-2228-3945}
\affiliation{%
    \institution{University of Wisconsin-Madison} 
    \city{Madison}
    \state{WI}
    \country{United States}
}
\email{ling-zhang@cs.wisc.edu}

\author{Shaleen Deep}
\orcid{0000-0003-2342-4060}
\affiliation{%
    \institution{Microsoft Jim Gray Systems Lab}
    \country{United States}
}
\email{shaleen.deep@microsoft.com}

\author{Jignesh M. Patel}
\orcid{0000-0003-3653-2538}
\affiliation{%
    \institution{Carnegie Mellon University}
    \city{Pittsburgh}
    \state{PA}
    \country{United States}
}
\email{jignesh@cmu.edu}

\author{Karthikeyan Sankaralingam}
\orcid{0000-0002-8315-2389}
\affiliation{%
    \institution{University of Wisconsin-Madison}
    \city{Madison}
    \state{WI}
    \country{United States}
}
\email{karu@cs.wisc.edu}

\renewcommand{\shortauthors}{Zhang et al.}

\begin{abstract}
    In this paper, we present the design and architecture of \sys, a novel system for indexing {log data for} regular expression {querie}s. Our main contribution is an $n$-gram-based indexing strategy and an efficient storage mechanism that results in a speedup of up to 14$\times$ compared to state-of-the-art regex processing engines that do not use indexing, using only 2.1\% of extra space. We perform a detailed study that analyzes the space usage of the index and the improvement in workload execution time, uncovering interesting insights. Specifically, we show that even an optimized implementation of strategies such as inverted indexing, which are widely used in text processing libraries, may lead to suboptimal performance {for regex indexing on log analysis tasks}. Overall, the \sys{} approach presented in this paper provides a significant boost when evaluating regular expression queries {on log data}. \sys{} is also modular and can work with existing regular expression packages, making it easy to deploy in a variety of settings.
    {The code of \sys{} is available at \url{https://github.com/mush-zhang/REI-Regular-Expression-Indexing}.}
\end{abstract}

\maketitle


\section{Introduction}

Text processing for data analytics is a fundamental requirement for several applications such as system diagnostics, performance evaluations, security audits, business intelligence. Within text processing, log data processing is a common task that extracts information from semi-structured log entries from heterogeneous sources.
A routine, but computationally expensive,
log processing operation is regular expression (regex, for short) processing. In many scenarios (such as debugging production system errors, root cause analysis, etc.), it is important to get responses from the underlying regex engine in a timely manner. This challenge is further exacerbated by the ever-increasing size of the data that needs to be processed, where brute-force regex searching is unsatisfactory. Recent work~\cite{blare} has introduced new insights on how to improve regex query evaluation for log analysis. \blare{}, focusing on regexes in log analysis tasks, made the conscious choice of not using any indexes to ensure the wide applicability of their techniques. Indeed,~\cite{blare} remarked that since the dataset may be used
in an ad-hoc manner without any need for querying repeatedly, the cost of building indexes is not justified. However, if one hopes to achieve faster log processing, building indexes over the log data becomes necessary.

\begin{figure}[!tp]
\begin{subfigure}{0.75\columnwidth}
  \centering
  \includegraphics[width=0.75\columnwidth]{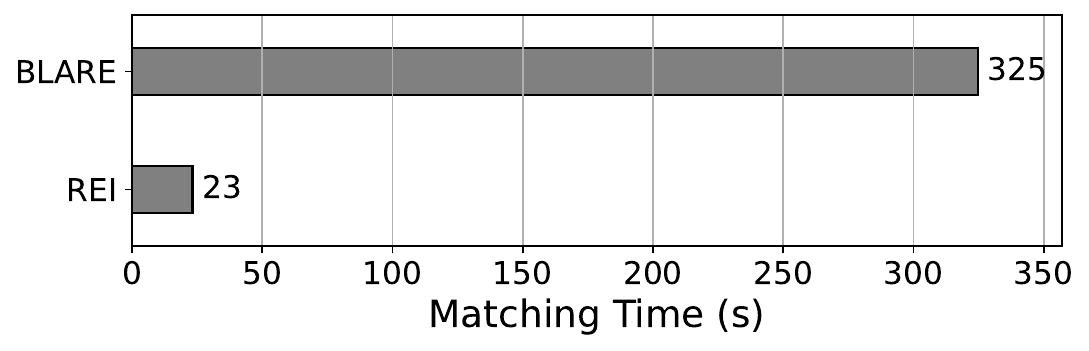}  
  \caption{Workload query matching time comparison.}
  \label{fig:best_time}
\end{subfigure}\\
\vspace{0.5em}
\begin{subfigure}{0.75\columnwidth}
  \centering
  \includegraphics[width=0.75\columnwidth]{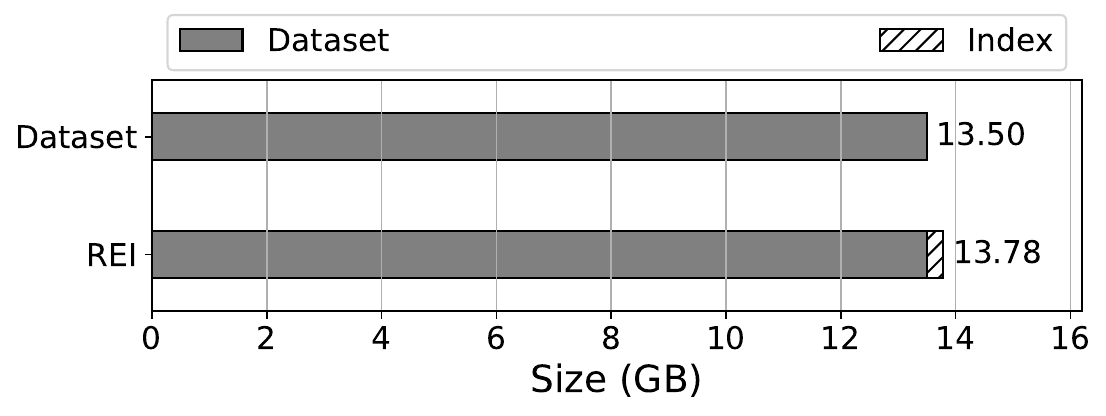}  
  \caption{Size comparison of the dataset and index achieving best performance.}
  \label{fig:best_time_size}
\end{subfigure}
\caption{Compared to the state-of-the-art regex matching framework \blare{}, \sys{} improves the performance by 14$\mathbf{\times}$ on a production workload, with 2.1\% extra space for the index.}
\label{fig:best}
\end{figure}
\eat{Add log workload characteristic here?
might want to condense the detailed description of the existing ones in this section and move it to background?}

\introparagraph{Log processing workload characteristics} Log processing workloads have several unique properties compared to other text processing tasks such as code search and document indexing. First, the smallest granular unit for logs is a log line whose length exhibits a skewed distribution. Most of the log lines are very small in length (up to $100$ characters in our workloads) with very few lines containing many characters (as large as $2000$). Consequently, the number of log lines are of the same order as the size of the logs. In contrast, for document indexing related tasks, each document is {the} smallest granular unit that needs to be retrieved for processing. The document size can range from a few bytes to several hundred MBs in size. Second, the sheer volume of log data even for simple applications can be huge. As an example, GitHub indexes~\cite{github} $15.5$B documents across all its code repositories with total data size of $115$TB. In contrast, for our log processing datasets, $15$B log lines is generated by production applications in only about $5$ hours. 
{Third, the log dataset is usually generated by log-generating formatted code from different components of the system.}
Finally, unlike for search applications where retrieving top-$k$ matches is the common setting (which allows for optimizations such as early termination of processing), the query workloads for log data require fetching all matches correctly.

\eat{\note{Log processing workloads has several characteristic, (1) its regexes are has simple regex components, and long literal components compared to regexes in general (more likely to be filtered by index than normal regex? It sits in between conventional regex query and text search) (2) log length distribution is extremely right-skewed; in one of our dataset, 50\% of the logs has around 100 or under characters, but the longest log line has 2,000 characters (cannot use signature files as it requires similar sized blocks). How do we incorporate these into log analysis regex index requirement? By comparing with existing method, I had a rough placeholder characteristic guideline (in discussion part of section 3.2): then the size of each data block is small, alphabet is large, and queries has the \textit{needle-in-a-haystack} characteristic, bit-vector index with well-designed encoding method can out-perform inverted indexes in log analysis task.}}
There are three key parameters that need to be considered when constructing an index for regex log processing: the index construction time, the space used by the index, and the performance gain obtained when running the workload at hand.
In particular, the choice of index data structure and entries for indexing has a substantial impact on the performance. Prior works mainly use inverted index suffix trees, and signature files as index data structures and $n$-grams of characters {or words} as index entries~\cite{Williams_2002, Navarro_2018, Kim_2010, SHIN_2018, Faloutsos_1984, Baeza_Yates_1996, Cao_2005}. Novel variants of suffix trees are used for processing biological data, utilizing the small alphabet size (such as the limited character set for expressing DNA). Yet suffix tree indexes can easily occupy space more than 10$\times$ compared to the original dataset itself~\cite{Ruano_2021}. Signature files~\cite{Faloutsos_1984} use list or hierarchical lists of special signatures that encodes all words in each document. This method is less popular as inverted index presents a better running time (but at the cost of using more space)\footnote{\cite{Faloutsos_1984} noted that signature files method requires space overhead of about $5-30\%$ of initial file sizes but this overhead is still too large for our use-cases. In fact, $5\%$ is the threshold on the space usage based on our discussions with product teams at X.}. Postgres uses inverted index of all trigrams for speeding up both exact and approximate text search~\cite{postgres_trg}. \eat{These indexes are used in the processing of string operations such as the \textsf{LIKE} operator in SQL.} Inverted indexes are also a popular choice that is used in frameworks such as Lucene~\cite{lucene} for text search. However, for both Postgres and Lucene, there is no option to control the space usage. As we will demonstrate later, inverted indexes for log processing can lead to significant space overhead (\cite{Faloutsos_1984} noted a $50-300\%$ space overhead). Index size can be significantly reduced by reducing the number of entries indexed. Cho and Rajagopalan~\cite{FREE} proposed an algorithm called \textsf{FREE} that builds an inverted index over carefully chosen multigrams. This work was subsequently improved upon by Hore et al.~\cite{BEST}, who also utilized multigrams but developed a more effective multigram selection strategy that formed the core of \textsf{BEST}. 

\introparagraph{Why existing solutions do not work for log indexing} \looseness=-1 Each of the above works have limitations. Both \textsf{FREE} and \textsf{BEST} require both the query workload and the data as a part of the input for selecting the multigrams. While the improved strategy of \textsf{BEST} helped the query performance, finding the near-optimal subset of multigram to index is significantly more expensive.  Unlike document search based applications where building inverted and tree-based indexes requires less than $1\%$ of the total input size\footnote{Using GitHub as an example, a complete binary tree with $\sim15B$ leaf nodes and each node storing $10$ bytes of information is $0.26\%$ of the $115$TB input size.} since the number of documents is much smaller than the total size of all the documents, for logs, the sheer number of log lines make it impractical to build tree-based indexes. \eat{Lastly, the reader can see that each of these works cannot be directly compared with each other. Therefore, there is an unmet need to systematically analyze the trade-off between the cost of index construction, the benefit of the index for query evaluation, and the impact of the index storage mechanism. }

\eat{
\jmp{Would be great to get a ``showcase'' picture in the intro at the top of column 2 on page 1. It could just be BLARE vs \sys{} for one dataset and note that compared to the state-of-the-art method BLARE, \sys{} improves the performance by XXX on a sample dataset YYY.}
}

\smallskip
\noindent {\it \textbf{Our Contribution.}} In this paper, we revisit the problem of regular expression indexing and systematically analyze how classical approaches such as $n$-grams usage and inverted indexing behave when used with state-of-the-art regex evaluation engines on modern hardware. Our contributions are as follows.

\introparagraph{1. \emph{A new framework for regex indexing}} We present a lightweight indexing framework, \sys{}, that is configurable to balance the cost of building indexes with query performance improvement. In contrast with existing approaches, our technique differs in two aspects. First, we use a filtering-based bit-vector index for a carefully chosen subset of strings of length $n$ (called $n$-grams). Second, rather than storing the index in a separate data structure, we store the bit-vector index along with each log line of our input. {The framework leverages the \emph{negative index property}, where the absence of required $n$-grams in the index can quickly eliminate log lines from consideration without the need for expensive regex evaluation.} This approach has several benefits, which we outline in~\Cref{subsec:index_construction}.

\introparagraph{2. \emph{Thorough analysis of gram selection strategies}} We compare and contrast our proposal with prior works in the area of regex processing using indexes. We empirically demonstrate that existing approaches fall short and lead to suboptimal performance compared to our approach. Regex queries in log analysis workloads tend to have relatively low selectivity as they mostly contain queries that are finding \emph{needle-in-a-haystack}~\cite{Whitaker_2004, Weigert_2014, Ellis_2015, Yu_2019}. We demonstrate that for such workloads, it is sufficient to use the query workload itself to choose the right set of $n$-grams, rather than finding the optimal set of $n$-grams.

\introparagraph{3. \emph{Adaption to Unknown Queries}}
{We propose a strategy for constructing an index when the query set is unknown. Existing key selection methods struggle with unknown queries in log processing workloads. Our system adapts by using similar index key selection guidelines, leveraging log processing characteristics. This strategy significantly reduces runtime and is comparable to an index built with a known query set.}

\introparagraph{4. \emph{Experimental Evaluation}} \looseness=-1 To demonstrate \sys{}'s efficiency, scalability, configurability, and robustness, we tested it on three real-world workloads. Our analysis shows \sys{} consistently outperforms popular schemes like inverted indexes in both construction overhead and performance. It improves the matching time of \blare{} by up to 14$\times$ with low space overhead (\Cref{fig:best}). \eat{We examined how parameters like the number of bigrams indexed and index entry granularity affect performance and overhead, allowing systematic optimization.} Additionally, \sys{} shows significant gains even without prior knowledge of regex queries by utilizing the log processing workload characteristics. 
\eat{
\begin{itemize}
    \item We present a lightweight indexing framework, \sys{}, that is configurable to balance the cost of building indexes with query performance improvement. In contrast with existing approaches, our technique differs in two aspects. First, we use a filtering-based bit-vector index for a carefully chosen subset of strings of length $n$ (called $n$-grams). Second, rather than storing the index in a separate data structure, we store the bit-vector index along with each log line of our input. This approach has several benefits (outlined in~\Cref{subsec:index_construction}) that lead to simplicity of use and ease of extension.
    \item We compare and contrast our proposal with prior works in the area of regex processing using indexes. We empirically demonstrate that existing approaches fall short and lead to suboptimal performance compared to our approach. Regex queries in log analysis workloads tend to have relatively low selectivity as they mostly contain queries that are finding \emph{needle-in-a-haystack}~\cite{Whitaker_2004, Weigert_2014, Ellis_2015, Yu_2019}. We demonstrate that for such workloads, it is sufficient to use the query workload itself to choose the right set of $n$-grams, rather than finding the optimal set of $n$-grams.
    \item To demonstrate that \sys{} is highly efficient, scalable, configurable, and robust, we implemented and tested it on three real-world workloads and systematically analyzed the impact of different parameters on the performance and overhead of our indexing framework. The results of our experiments show that \sys{} is effective in that it consistently outperforms popular schemes such as inverted indexes in terms of both construction overhead and performance improvement. It also outperforms the state-of-the-art framework \blare{} by a factor as large as $14\times$ with low space overhead (\Cref{fig:best}). We examined at how configurable parameters (the number of bigrams indexed, the granularity of index entries) affect the overhead and performance, and we demonstrate the ability to systematically tweak those knobs to get the best outcome. Moreover, our framework shows significant gains in performance even in the absence of prior knowledge about the content distribution of the regex workload. 
\end{itemize}
}

Overall, our experiments and analysis demonstrate that \sys{} is a simple yet effective regex indexing framework, and showcases its potential to significantly enhance regex query performance in real-world log analysis tasks.

\eat{The remainder of this paper is structured as follows: Section 2 provides an overview of the necessary background, covering regular expression, $n$-grams, and previous work in $n$-gram-based indexing. In Section 3, we present the design of the \sys{} indexing framework and provide insight into the techniques used and their justifications. A comprehensive empirical evaluation of both production and academic real-world workloads is presented in Section 4. In Section 5, we discuss the related works, and we conclude in Section 6.}

\section{Background} \label{sec:background}
\eat{In this section, we explain the foundational principles of regular expressions, focusing on essential concepts. }
\eat{In this section, we explain the foundational principles of regular expressions queries, $n$-grams that are commonly used in 
indexing for regex evaluation, and existing data structures that can be used for indexing $n$-grams.}

\begin{table}[t]
  \footnotesize
  \centering
  \captionof{table}{{Summary of symbols and notation used in this paper.}}
    \begin{tabular}{c|l |c|l}
      \hline
      \textbf{Symbol} & \textbf{Meaning}  &  \textbf{Symbol} & \textbf{Meaning} \\
      \hline
      $\Sigma$  & Finite alphabet & $W$         & Workload consisting of $Q, L$ \\
      $R$       & Regular expression & $q$         & Individual query \\
      $L$       & Log dataset & $n$         & Length of $n$-grams \\
      $Q$       & Set of regex queries & $k$         & Number of $n$-grams indexed \\
      $I$       & Index built & $g$         & Individual $n$-gram \\
      $\ell$    & Individual log line & $\literals$ & Literal components\\
      \hline
    \end{tabular}
  \label{tab:symbols}
\end{table}

\introparagraph{Regular Expression} 
We define regular expressions over a finite alphabet $\Sigma$ using syntax such as the empty language ($\emptyset$), empty string ($\varepsilon$), literals (character $c \in \Sigma$), union of two languages ($\mid$), concatenation of two languages ($\cdot$), and the Kleene Star (zero or more concatenations of the language, denoted by $*$). The syntax allows us to represent a wide variety of string patterns, with the language denoted by a regular expression defined through a function $L : R \rightarrow 2^{\Sigma*}$. The syntax of regex is as follows:
$$ R : \emptyset \mid \varepsilon \mid c \mid (R \mid R) \mid (R \cdot R) \mid (R*) $$
\eat{Almost all state-of-the-art regex libraries compile regexes to automata, either non-deterministic automata (NFAs) or deterministic automata (DFAs), and performs matching on the automata.}

\introparagraph{Literal Component}
It is also important to consider literal components, which is a major source for indexing, in regular expressions. Literal components are sequences of characters from the alphabet that match the input string exactly as they appear, without any interpretation as operators or special characters. Formally, we denote them as $\literals \leftarrow \Sigma* $. They are distinct from the regular expression components, which include operators like $*$, $+$, $?$, $\mid$, and other defined patterns. 
\eat{For instance, \texttt{abc} is a literal component and \texttt{\textbackslash{}d\{3\}} is a regular expression component in the regular expression \texttt{abc\textbackslash{}d\{3\}}". While the regular expression component \texttt{\textbackslash{}d\{3\}}" matches any sequence of three digits, the literal component \texttt{abc} will only match the string "abc". 
}
{Literal components in regexes are especially relevant in log processing workload, as the regexes in such workload often contains literals that are more selective than their regex components.}
The concept of \literals\ has been used in several prior works~\cite{blare, wang2019hyperscan}, in which the authors utilize the fact that exact character matching is much cheaper than state transitions in an automaton. Hyperscan~\cite{wang2019hyperscan} evaluates components in a sequential fashion together with other optimizations. {For each regex,} \blare{}~\cite{blare} {selects a {specific} proportion of the input data and} uses machine learning techniques {on the picked subset} to determine the best regex-literal splitting strategy for each query{, and then evaluate the rest of the dataset with the chosen strategy.} It usually evaluates all literal components before regex components.

\eat{
In this section, we present the notation and the class of regex queries
that are considered in the paper. Given a finite alphabet $\Sigma$, the syntax
of regex can be defined by:

$$ R : \emptyset \mid \varepsilon \mid c \mid (R \mid R) \mid (R \cdot R) \mid (R*) $$

where $\emptyset$ denotes the empty language, $\varepsilon$ is the language containing the empty string $\{ \epsilon\}$, $c \in \Sigma$ denotes the language consisting of the single character $\{c\}$, $R \mid R$ is the union of two languages, $R \cdot R$ is the concatenation of two languages, and $R*$ denotes zero or more concatenations of the language $R$. We denote by $\mathcal{R}$ the set of all possible regular expressions. The {language} of regular expression is defined by $\mathcal{L} : \mathcal{R} \rightarrow 2^{\Sigma*}$:

\begin{align*}
    \mathcal{L}(\emptyset) &= \emptyset  \\   \mathcal{L}(\varepsilon) &= \epsilon \\   
\mathcal{L}(R_1 \mid R_2) &= \mathcal{L}(R_1) \cup \mathcal{L}(R_2)\\
\mathcal{L}(R_1 \cdot R_2) &= \mathcal{L}(R_1) \cdot  \mathcal{L}(R_2) \\
\mathcal{L}(c) &= \{c\} \\ \mathcal{L}(R*) &= \mathcal{L}(R)*
\end{align*}

The concatenation operator (denoted by $\cdot$) over languages is defined as $L_1 \cdot L_2 = \{ \omega_1 \omega_2 \mid \omega_1 \in L_1, \omega_2 \in L_2\}$. The closure operator (denoted by $*$) is defined as $L* = \bigcup_{i \geq 0} L^i$ where $L^0 = \{\epsilon\}$ and $L^i = L \cdot L^{i-1}$ for any $i > 0$. We define $\literals \leftarrow \Sigma* $, i.e., any sequence of zero or more that can be formed using the letters of the alphabet. The concept of \literals\ has been used in several prior works~\cite{blare, wang2019hyperscan}.
}

\introparagraph{$N$-Grams}
An $n$-gram (also referred to as a  $q$-gram in the literature) is a subsequence of length $n$. An $n$-gram generation from a string $s$ of length $c$ refers to constructing the set of subsequences of length $n$ for each starting position $i \in \{0, \dots, c - n\}$. Given a string of length $c$, there are $c - n + 1$ $n$-grams that can be generated from $s$. For example, given a string \emph{vmName} and $n=3$, there are four $n$-grams that can be constructed: \emph{vmN}, \emph{mNa}, \emph{Nam}, and \emph{ame}. If $c \leq n$, then the $n$-gram generation of $s$ is $s$ itself. The most commonly used $n$-grams in practice are bigrams ($n=2$) and trigrams ($n=3$). Mixed $n$-grams have also been used in prior works. \eat{Existing indexes for regex usually use character-based $n$-grams; in area of text-search, word-based $n$-grams are also popular.}

\introparagraph{$N$-Gram Selection} 
There have been several optimizations proposed in the literature for improving the space efficiency of an inverted index. For example, \textsf{FREE}~\cite{FREE} shows that there is an inverse trade-off between the increasing value of $n$ for generating the $n$-grams and the size of the inverted index. 
{\textsf{BEST}~\cite{BEST} reduces the gram selection problem to graph covering problem, and select the multigram set with approximately (with provable guarantees) highest \textit{benefit} {(which will be discussed in \Cref{subsec:gram_select})}. {\textsf{LPMS}~\cite{FAST} combines strategies of \textsf{FREE} and \textsf{BEST} and uses linear programming to approximate the optimal multigram set.} These methods {require} full knowledge of the dataset and the query before selection and the calculation is very time consuming.

GitHub Code Search~\cite{github} reduces the size by using multigrams starting from bigrams for its generality. It limits the number of multigrams with $n>2$ by assigning weights to every bigram and creates indices for multigrams formed by adjoining bigrams. For these multigrams, the weights of the internal bigrams are strictly lower than those of the bigrams at both ends.

\introparagraph{Index Data Structure}
An $n$-gram index is a data structure used to determine if a subsequence \eat{identical to a given string }$s$ exists in an input $\ell$. Google Code Search~\cite{cox2012regular} introduced a common method for constructing an $n$-gram index for document search, creating an inverted index with trigrams as keys and lists of document identifiers as values. This allows for finding documents containing all trigrams generated from $s$ and supports disjunctions, such as finding documents containing either \emph{abc} or \emph{def}. Baeza-Yates and Gonnet~\cite{Baeza_Yates_1996} introduced suffix tree indexes for regex evaluation, which store all possible suffix strings in a dataset. The large space requirement makes them only practical in protein datasets with limited alphabets. Bit-vector or bitmap indexing is popular for its fast bitwise operations and performance benefits for highly selective queries~\cite{Chan_1998, Chan_1999}. This type of index associates one bit-vector with each record, suitable for general selection queries across various data types. For text data, methods like superimposed signature files~\cite{Faloutsos_1984, Sacks_Davis_1987, Wong_1985, Goodwin_2017} and bloom filters~\cite{Grabowski_2016} encode documents into bit-vectors, providing effective prefiltering but with significant space overhead.

\section{Framework and Techniques}

In this section, we present the design and techniques used for our indexing framework. The framework consists of two steps. The first step is the extraction of $k$ bigrams from the query set and ordering them based on their frequencies. The second step involves the construction of an $n$-gram index. Using the bigrams extracted, we construct a bit-vector filter consisting of $k$ bits for each line in the input log. The $i^\text{th}$ bit in the filter for a log line $\ell$ indicates whether the $i^\text{th}$ bigram is present in $\ell$ or not. In the querying phase, given a regular expression, we find the string literals in the regular expression, generate $n$-grams for each string literal, and use the $n$-gram index to check if all $n$-grams for one literal exist in the index. We will discuss in details in the following subsections.

\subsection{N-Gram Selection}\label{subsec:gram_select}
\looseness=-1 The choice of $n$-grams for the index is crucial to the effectiveness of the indexing framework when matching regular expressions over large log datasets. It is influenced by four components: queries $Q$, dataset $L$, the choice(s) of $n$, and the number of chosen $n$-grams ($k$). Prior works have made mixed choices of $n$ or have chosen a fixed sized $n$. We made the choice of only using bigrams (i.e. $n=2$). 
{Bigrams strike an optimal balance between being general enough to appear across diverse queries while maintaining sufficient filtering power. As $n$ increases, $n$-grams become more specific to particular queries, reducing their ability to filter across a broad range of regex patterns. We demonstrate it experimentally in \biVsTri{} with a detailed comparison between bigrams and trigrams. Our key findings are that as $n$ goes beyond two, the fraction of bigrams that are a subset of $n$-grams also increases, and thus, with high probability, the bigrams generated from the larger $n$-gram are as selective as the bigger $n$-gram itself.}

\Ling{Recent work compared across the existing $n$-gram selection methods (\textsf{FREE}, \textsf{BEST}, and \textsf{LPMS}) both theoretically and experimentally~\cite{ngram_select_comp}. It runs the three methods on different types of workloads and summarize the insights into a general guideline to select among the three selection methods according to the workload characteristics. For our workloads where unseen queries may be expected, query set is large with long query literals, \textsf{FREE} performs the best. On the other hand, \textsf{BEST} can generate the theoretically optimal $n$-gram set for indexing. }

\Ling{Therefore}, we compare against the strategy of choosing $n$-gram as introduced by the \textsf{BEST} and \textsf{FREE} algorithms. 
\Ling{\textsf{FREE} rus with space overhead $O(max_n \cdot |L|)$ and the compute overhead as $O(max_n \cdot  |L|)$ where $max_n$ is the configurable maximum length of $n$-grams.
\textsf{BEST} finds $k$ $n$-grams that theoretically maximizes the effectiveness in time $O(k \cdot |Q| \cdot |L|)$ and uses $O(|Q| \cdot |L|)$ space.}

{The \textsf{BEST}~\cite{BEST} method introduces the concept of \emph{benefit} and \emph{cover} to measure the effectiveness of $n$-gram selection. The \emph{cover} of an $n$-gram $g$ on workload $W = (L, Q)$, $cover(g) = \{(q,l)\in Q\times L\mid g\in q\land g\notin l\}$, is the set of log lines that can be filtered out by $g$ among all queries.
The \emph{benefit} can be expressed as  $\mathnormal{bene}_{W}(g) = |cover(g)|$. 
\textsf{FREE} selects grams by choosing the $n$-grams with the highest benefit, in order of increasing $n$ (i.e., smaller $n$ is preferred).
The \emph{benefit} of index $I$ on a workload denoted as $\mathnormal{bene}(I) = |\bigcup_{g \in I} cover(g)|$, is the number of log lines that can be filtered out by the index $I$. Therefore, the incremental benefit given an index is
$\mathnormal{bene}(g|I)= |cover(I\cup \{g\}) - cover(I)|$.

This metric captures both the filtering power of the $n$-gram and its relevance to the query workload, considering other $n$-grams in the index. \textsf{BEST} reduces the gram selection to a graph covering problem and selects the multigram set with provable guarantees on the benefit. }

We propose an $n$-gram selection strategy that runs in time $O(|Q|)$ \Ling{with negligible space overhead }{by selecting based on $n$-gram frequency in the query workloads}. \Ling{This is significantly smaller than both \textsf{FREE} and \textsf{BEST}} {in that it focuses on query-relevant $n$-grams rather than dataset-based selectivity measures. The intuition is that frequently occurring $n$-grams in queries are more likely to be useful for filtering across diverse log analysis tasks.}
{Our $n$-gram selection strategy works as follows: For each query $q$ in the workload $Q$, we first preprocess the query to extract literal components, then scan these literals from left to right to identify all bigrams present in the query. For each bigram discovered, we update a global counter dictionary that tracks the number of queries containing that specific bigram. After processing all queries, we rank bigrams by their frequency (number of queries containing them) and select the top-$k$ most frequent bigrams for indexing.
}

In order to show the cost and effectiveness of $n$-gram selection strategies, we run a micro-benchmark that builds and uses indexes selected with 3 different strategies {with low index overhead}. \Ling{We take the simplified concept of \textit{benefit} from \textsf{BEST}, defined in~\cite{ngram_select_comp}. Let $g$ be an individual bigram. The three are: }
(1) the most frequently occurring bigrams (Freq), (2) bigrams with the highest $benefit(g,\emptyset)$ (Bene), and (3) selected bigrams $g_i$\Ling{based on $benefit(g_i|\{g_{i-1}\})$, where $g_j$ ranks according to $benefit(g_j,\emptyset)$} (Incr\_Bene). We run the three methods on a subset of a real-world workload containing one million log lines and 132 queries.

\begin{figure}[!tp]
\begin{subfigure}{0.95\columnwidth}
  \centering
  \includegraphics[width=0.6\columnwidth]{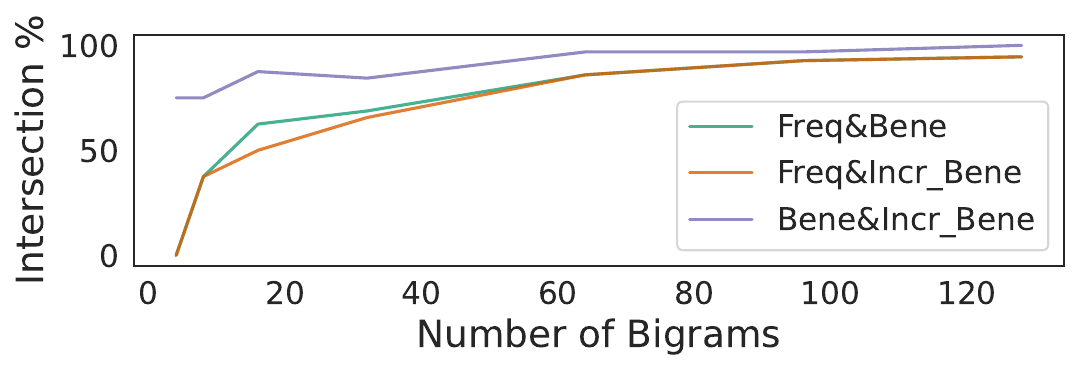}  
  \caption{Bigram intersection  as a percentage from strategy pairs.}
  \label{fig:microbench_intersect}
\end{subfigure}
\\
\hspace{1.1em}
\begin{subfigure}{0.95\columnwidth}
  \centering
  \includegraphics[width=0.6\columnwidth]{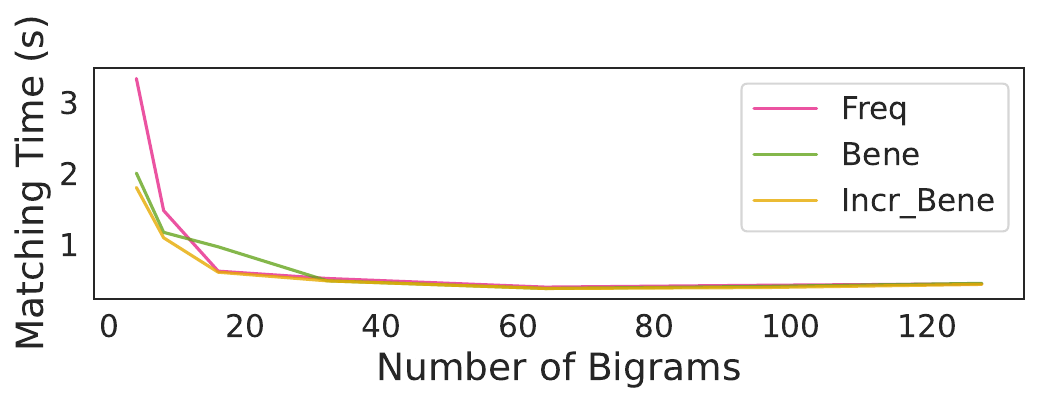}  
  \caption{Regex Matching Time.}
  \label{fig:microbench_time}
\end{subfigure}
\caption{Varying the number of bigrams, compare the set of bigrams selected by the three methods and the matching time applying their resulting indices.}
  \vspace{-1.5em}
\end{figure}
{To develop some intuition about how different the bigrams chosen by each of the strategies are,} we look at their overlap in their selected bigrams when selecting the top 4, 8, 16, 32, 64, 96, and 128 bigrams and present the \emph{intersection percentage} in~\Cref{fig:microbench_intersect}. 
{The \emph{intersection percentage} refers to the percentage of common bigrams selected by different strategies. We use it to measure the overlap between the output sets of different selection methods. With the selection by \textit{Incr\_Bene} closer to optimal selection, a higher intersection percentage with \textit{Incr\_Bene} shows a higher performance of the selection method.}
There is a notable similarity in the bigrams selected by Bene and Incr\_Bene. The intersection percentage is high at 75\% with only 4 and 8 bigrams and steadily increases to 96.9\% for 96 bigrams and 100\% for 128 bigrams. {We observed that conditional on the top-100 selected bigrams, no other remaining bigrams can generate a positive incremental benefit.} As the number of chosen bigrams increases from 4 to 64, the intersection percentage of Freq with the other two methods gradually increases to 85.9\%. Their rates of intersection plateaus as the number of bigrams selected continues to increase. 

In addition, we evaluate the matching performance using the indices derived from the three methods and present the results in~\Cref{fig:microbench_time}. 
The matching time is quite different when the number of bigrams selected is small, where using the index built with bigrams selected by Freq gives the slowest matching time. The times gradually converge for the indices built with the three methods as the number of bigrams increases. The running time reaches a minimum when indexing with 64 bigrams, where they are 0.395, 0.376, and 0.376 for Freq, Bene, and Incr\_Bene respectively.

We also evaluate the time needed to choose bigrams for each approach.  
The ranking calculations for frequency and benefit are finished in 0.0014 seconds and 21.5 seconds, respectively, for selecting 64 bigrams. In sharp contrast, the Incr\_Bene technique requires hours (31,638.2 seconds) when only taking into account the bigrams with top-200 individual benefit. The computational time increases even more when the Incr\_Bene approach is applied to the entire dataset, taking well over $24$ hours to complete\footnote{We terminate the experiments at $24$ hours.}. When choosing bigram selection strategies for indexing, there are trade-offs between filtering power of selection result and computational efficiency.

Based on this microbenchmark, we find that selecting bigrams by frequency is the most effective strategy for log analysis workloads, offering significantly lower overhead than method \textsf{FREE}, while yielding results comparable to the most precise method \textsf{BEST}.

\subsubsection{Considerations for $N$-Gram Selection} \label{subsubsec:num_ngram}

In this subsection, we outline some of the considerations for $n$-gram selection.

\introparagraph{Frequency Threshold for $N$-Grams Selection}
Prior work has proposed heuristics such as removing $n$-grams that appear in more than 10\% of queries, as these are typically pruned by algorithms when identifying the optimal covering set of bigrams. We chose not to adopt this strategy due to differences in index use cases. Our selected $n$-grams are intended not only to enhance matching performance for the current set of queries and dataset but also to improve performance for potential future queries on new inputs that may be added to the log.

\introparagraph{Number of $N$-Grams Indexed and Choosing $k$} 
The choice of $k$ is non-trivial. Assuming we know the ideal $k$, our microbenchmark shows that selecting the $k$ most frequent $n$-grams from the workload is sufficient. This works because regex queries typically exhibit \emph{needle-in-a-haystack} behavior with very low selectivity, largely driven by string literals~\cite{Whitaker_2004, Weigert_2014, Ellis_2015, Yu_2019}. Thus, focusing on literals when picking $n$-grams is effective. Increasing $k$ improves coverage but offers diminishing returns beyond a point.

The size of a bit-vector index is determined by the number of bigrams $k$ we choose to index. This fixed size ensures predictability in storage requirements, making memory allocation and management more straightforward. We can set the value of $k$ {to multiples of word size} such that no space is wasted due to alignment or padding. {When the size of the query subset in the workload is small and expected to not be a perfect representation of all future queries, consider {increasing} $k$. This will allow indexing some bigrams that {appear} less frequently in the query subset at hand, but may appear in future queries.}

{
In practice, users may want to specify size constraints for the index. We propose to automatically suggest optimal values for $k$ and index granularity under given size constraints by using heuristics involving bigram selectivity analysis. 
We provide a prototype \sys-Tuner with heuristic estimation of filtering benefit per bit and present initial experimental results \tuner{}. 
}

\subsection{Index Construction}\label{subsec:index_construction}
\begin{algorithm}[!t]
        \small
	\SetCommentSty{textsf}
	\DontPrintSemicolon 
	\SetKwInOut{Input}{Input}
	\SetKwInOut{Output}{Output}
	\SetKwRepeat{Do}{do}{while}
	\SetKwFunction{uniqueGram}{\textsc{Get-Unique-Gram}}
	\SetKwFunction{selectBigram}{\textsc{Select-Bigram}}
	\SetKwFunction{getLiteral}{\textsc{Get-Literal-In-Regex}}
	\SetKwFunction{setBits}{\textsc{Set-BitVector}}
	\SetKwFunction{append}{\textsf{append}}
	\SetKwFunction{contains}{\textsf{contains}}
	\SetKwFunction{size}{\textsf{size}}
	\SetKwFunction{len}{\textsf{len}}	
        \SetKwFunction{add}{\textsf{add}}	
        \SetKwFunction{put}{\textsf{put}}	
        \SetKwFunction{get}{\textsf{get}}	
        \SetKwData{sj}{$\mathbf{S}$}
	\SetKwProg{myalg}{procedure}{}{}
	\SetKwProg{myproc}{\textsc{procedure}}{}{}
	\SetKwData{return}{\textbf{return}}
	\Input{Workload $W$ with logs $L$ and regexes $R$, $k$}
        $I\gets []$\;
        $G\gets \selectBigram{R, k}$\;
	\ForEach{$\ell \in L$ \tcc*{{processing each log line}}}{
            \sj $\leftarrow$ \uniqueGram{$\ell$, $2$}\;
            $b \gets$ \setBits{\sj, G}\;
            $I.$\append{$b$}
	}
        \return $I$\;
        \myalg{\uniqueGram{s, n}}{
            $S \gets $ hash map \label{line:S} \;
            \For{$i \in \{ 0, \dots, \size{s} - n \}$}{
                $g \leftarrow s[i:i+n- 1 ]$\;
                \If{$\neg S.\contains(g)$ and $g$ is a string literal}{
                    $S.\put(g, 0)$ \;
                }
                $S.\put(g, S.\get(g) + 1)$
            }
            \return $S$
        }
        \myalg{\selectBigram{R, k}}{
            $T \leftarrow $ hash map \label{line:T} \;
            \ForEach{$r \in R$}{
                $C \leftarrow$ \getLiteral{$r$} \;
                $T_i \leftarrow $ hash set\;
                \ForEach{$c \in C$}{
                    $T_i.$\add(\uniqueGram{$c$, $2$}.keys())\;
                }
                \ForEach{$g\in T_i$}{
                    \If{$\neg T.\contains(g)$}{
                        $T.\put(g, 0)$ \;
                    }
                    $T.\get(g) \leftarrow T.\get(g) + 1$\;
                }
            }
            Sort $T.$keys() by value in descending order into vector $V$\;
            $G \leftarrow $ hash map \;
            \For{$i \in \{0, \dots, k-1\}$}{
                $G.\put(V[i], i)$\;
            }
            \return $G$
        }
        \myalg{\setBits{S, $G$}}{
            $b \leftarrow $bit vector of length $\size{G}$\;
            \tcc{Compare sizes of the two hash-maps and decide the probing map}
            \If{\size{$S$} $<$ \size{$G$}}{
                \ForEach{key $\in S$}{
                    \If{$G.$\contains{key}}{
                        $b[G[key]] \leftarrow 1$
                    }
                }
            }
            \Else{
                \ForEach{key, val $\in G$}{
                    \If{$S.$\contains{key}}{
                        $b[val] \leftarrow 1$
                    }
                }
            }
            \return $b$
        }
	\caption{{\sc BitVector-Index-Construction}}
	\label{algo:index_construct}
\end{algorithm}

Once the bigrams have been selected and ranked according to their frequency, we proceed to scan the dataset. For each log line in the dataset, we get all its bigrams and identify the presence of each of the selected $k$ bigrams. This information is encoded as a $k$-bit vector where the $i^\text{th}$ bit indicates if the $i^\text{th}$ bigram is present in the log line or not. 

\Cref{algo:index_construct} shows the execution steps for building the index. Let $\mathbf{b} = (b_1, b_2, \cdots, b_k)$ be the ordered sequence of the selected bigrams. We use $\mathbf{G}$ as one-to-one mapping from each bigram $b_i$ to its offset $i$ in the bit-vector. First, for each log line $\ell$, we use procedure \textsc{Get-Unique-Bigrams} to do the bigram generations for $\ell$. Then, given the bigrams, we create the bit-vector in procedure \textsc{Set-BitVector}. The resulting bit-vector for the processed log line is added to a list $I$ that stores the bit-vectors, which is the final index. Observe that both the time requirement and the space requirement of the index are $O(k \cdot |L|)$.

\introparagraph{Discussion} Beyond bit-vector based index, signature files based index, and inverted index are alternate methods techniques for building index. In fact,~\cite{Goodwin_2017} showed that signature files {inspired bit-vector based} methods can outperform inverted indexes contrary to conventional belief. In \Cref{subsec:other_idx}, we experimentally demonstrate that for log processing, {carefully crafted} bit-vector based methods outperform signature files and inverted indexes.

\subsubsection{Index Granularity} \label{sec: index_granularity}
{Index granularity defines how many log lines each index entry covers. Our framework creates one bit-vector for every $m$ consecutive lines, where each bit marks the presence of a specific bigram in that group (note that \Cref{algo:index_construct} considers $m=1$, i.e., each log line has a bitvector). For example, with $m=4$, a single bit-vector represents four lines; if any line contains a bigram, its bit is set to 1. During queries, if all required bigrams appear in the bit-vector, the group is scanned line by line using regex.

The choice of granularity involves important trade-offs: \begin{itemize}[wide] \item \textbf{Finer granularity} (smaller $m$): Higher filtering precision but larger index size and higher construction cost. \item \textbf{Coarser granularity} (larger $m$): Lower storage and computation cost but more false positives. \end{itemize}

The optimal setting depends on log characteristics, query patterns, and resource constraints. \sys{} supports configurable granularity to balance precision and efficiency.
}

\subsection{Regular Expression Querying} \label{subsec:framework:regex_query}
\begin{algorithm}[!t]
        \small
	\SetCommentSty{textsf}
	\DontPrintSemicolon 
	\SetKwInOut{Input}{Input}
	\SetKwInOut{Output}{Output}
	\SetKwRepeat{Do}{do}{while}
	\SetKwFunction{uniqueGram}{\textsc{Get-Unique-Gram}}
	\SetKwFunction{splitRegex}{\textsc{Extract-Literals-In-Regex}}
	\SetKwFunction{setBits}{\textsc{Set-BitVector}}
	\SetKwFunction{bitMask}{\textsc{Get-Bitmask}}
	\SetKwFunction{add}{\textsf{add}}
	\SetKwFunction{append}{\textsf{append}}
	\SetKwFunction{contains}{\textsf{contains}}
	\SetKwFunction{size}{\textsf{size}}
	\SetKwFunction{len}{\textsf{len}}
	\SetKwProg{myalg}{procedure}{}{}
	\SetKwProg{myproc}{\textsc{procedure}}{}{}
	\SetKwData{return}{\textbf{return}}
        \SetKwData{allone}{\textsf{allOne}}
	\Input{Log $L$, Index $I$, Regex $r$, $n$, hash-map $G$ from $n$-gram to offset}
        $m \gets$ \bitMask{$r$, $n$, $G$} \; 
        $\allone \leftarrow$ $k$ bit array with all bits set to $1$ \;
	\For{$j\gets 0$ to \size{$L$}}{
            \If{$I[j] \lor m =  \allone$ \label{line:filter}}{
                do Regex-Match($L_j$, $r$) \label{line:regex}
            }
        }
        \myalg{\bitMask{r, n, G}}{
            $T \gets$ \splitRegex{r} \tcc*{$T$ is the list of string literals in the regex}  
            $S \leftarrow \emptyset$ \tcc*{empty set}
            \ForEach{t $\in T$}{
                $S_t \leftarrow \uniqueGram{r, n}$\;
                $S.$\add{$S_t$}
            }
            $m \gets$ \setBits{$S$, $G$} \label{line:setbit} \;
            $m \leftarrow \neg m$  \label{line:flip} \tcc*{flipping all the bits}
            \return $m$ 
        }
	\caption{{\sc BitVector-Index-Query}}
	\label{algo:index_query}
\end{algorithm}
Leveraging the created index, when querying over the dataset with a regular expression $r$, we build a bit-mask for the literal components of $r$ and filter out part of the logs that do not match.
We present the steps for querying a regex on the dataset with a bit-vector index in~\Cref{algo:index_query}.
When a regular expression query is received, it is first parsed to generate the subset of bigrams that are used in the index. Then, we construct a bit-mask of size $k$, {placing a bit value of \texttt{1} in the $i^{th}$ index to denote that a bigram is \emph{not} present in the regex.} In particular, if a bigram $g$ that is used in indexing but is absent in the regex, we set $G[g]^{\text{th}}$ bit of the mask to $1$, and $0$ otherwise. 

\begin{figure}[!tp]
 \centering
\includegraphics[width=0.85\columnwidth]{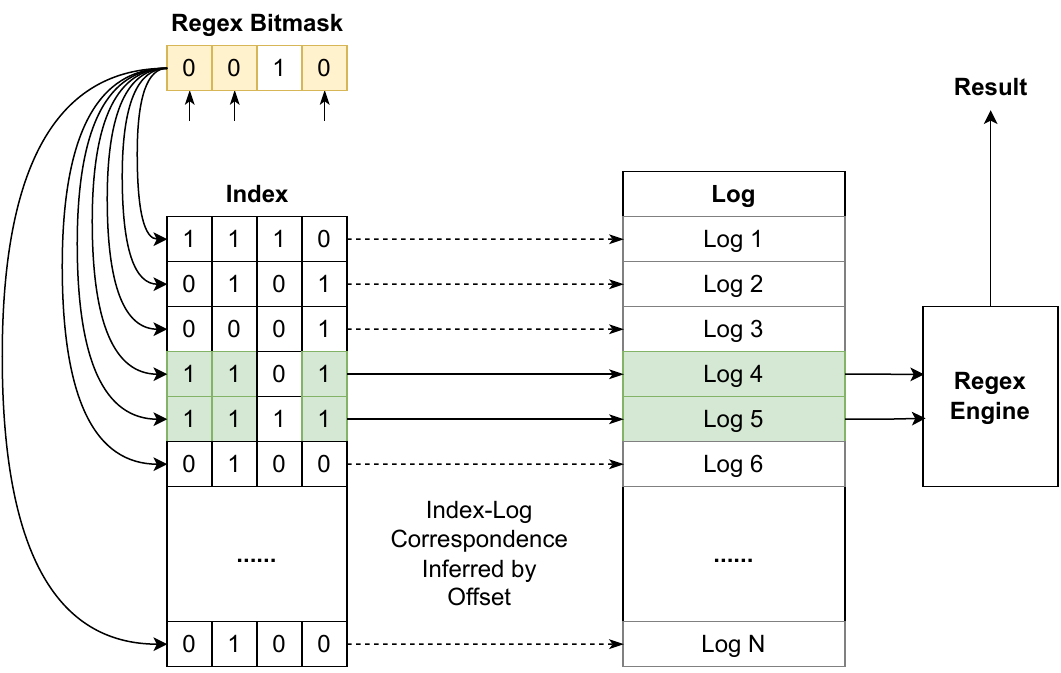}  
\caption{Query overview using a bit-vector index with $k=4$.}
\label{fig:framework}
\end{figure}

We scan through the array of bit-vectors and do a bitwise OR operation for each entry with the bitmask of the regex. The log line passes the index's initial filter if the outcome of the OR operation is a $k$-bit vector of all \texttt{1}'s. {Here, we leverage the negative index property: a set bit in the bitmask indicates absence of a bigram and thus, its presence or absence in the log line is inconsequential. The OR operation effectively verifies this condition. Observe that for bigrams that are not used in indexing, we always need to read the log line.}  {The index does not produce false-negatives, and we include the proof in \correctProof{} of the full paper.}
{Our implementation use the C++ \texttt{all()} function to check if all bits are set. This is done to implicitly capture the negative index property. }

An example of the querying process is illustrated in~\Cref{fig:framework}. We use four bigrams in the index ($g_0$ to $g_3$). Suppose the query $q$ has the bigrams $g_0$, $g_1$, and $g_3$. Then, the bitmask generated will be $0010$. Only the bit-vector of log lines four and five lead to a vector of all 1’s when a bitwise OR operator is performed with $0100$. 
In~\Cref{fig:framework}, the log lines corresponding to index entries with all relevant bits set to 1 (colored green) pass the filter.

Only the log lines whose bit-vector representations match the derived bigrams are considered for further processing.
This significantly reduces the number of log lines that need to be processed, ensuring efficiency.
For the log lines that pass this initial filtering, a full regex matching process is initiated (line~\ref{line:regex}). We employ the Google-RE2~\cite{re2} regex library for this purpose, given its robustness and efficiency in handling complex regular expressions.

\subsubsection{Unknown Workload And Future Queries} \label{subsec:unknown_wl}
In the previous section, we discussed the effectiveness of \sys{} given a fixed known workload. However, \sys{} can also help with a new set of queries on the same set of log data or an unknown workload. 
{This is particularly relevant for ad-hoc querying scenarios where analysts may need to search logs with previously unseen regex patterns.}

\sys{} addresses this challenge by leveraging a key characteristic of log analysis workloads: both logs and queries contain substantial amounts of human-readable text. Log analysis queries typically aim to extract variable values (e.g., using \texttt{\textbackslash{}d+} to extract numeric vmID) while filtering based on literal English text components (such as \texttt{vmID=}). This suggests that bigrams in log analysis workloads have distributions similar to English text, and that bigrams selected from English literature have a high probability of appearing in future log analysis queries.

Our strategy for unknown workloads involves using the most frequent bigrams from English literature as index keys. This approach significantly reduces the computational overhead of index construction while maintaining effectiveness comparable to indexes built with known query sets.
{The limitation of this strategy arises when the distribution of literal components in the query set significantly diverges from that of natural English language. This is especially true when the unknown query set exhibits domain-specific vocabulary.
Therefore, \sys{}’s performance on unknown workloads is most reliable when they retain sufficient overlap with general human-readable English patterns.
We demonstrate the effectiveness and limitations of this strategy in~\Cref{subsec:qgram_english}.}


\section{Evaluation} \label{sec:expr}

We implemented \sys{} using C++ and used the state-of-the-art regular expression library, Google's RE2, for any exact regex matching after passing the index checking. We also implemented an inverted index \Ling{and signature files }under the same setting. 
{For comparison purposes, we implement posting lists and inverted indexes in-memory using hashmaps, where each $n$-gram serves as a key and maps to a list (or set) of identifiers indicating which log lines or groups contain that $n$-gram. During query processing, the relevant posting lists are retrieved and intersected to find candidate log lines that contain all required $n$-grams. We use these terms interchangeably in the section.
}
We want to answer the following questions:

\begin{enumerate}[label=\textbf{Q.\arabic*}, ref=\textbf{Q.\arabic*}]
    \item How does the chosen value of $n$ for $n$-grams of the index affect the overhead of constructing the index and the improvement of performance (cf. Section~\ref{subsec:qgram_type})? \label{q1}
    \item How does the chosen value of $k$ (the number of $n$-grams used in the index) affect the overhead of constructing the index and the improvement of performance (cf. Section~\ref{subsec:qgram_num})?\label{q2}
    \item {How does the index impact query performance when we have an unknown log analysis workload (cf. Section~\ref{subsec:qgram_english})?\label{q3}}
    \eat{\item How does the existence of prior knowledge of regex workload impact the performance improvement of the constructed index (cf. Section~\ref{subsec:qgram_english})?\label{q3}}
    \item How does the granularity of the index affect the overhead of constructing the index and the improvement of performance (cf. Section~\ref{subsec:idx_granu})?\label{q4}
    \item How does \sys{} compare with other commonly used indexing schema on the overhead of constructing the index and the improvement of performance (cf. Section~\ref{subsec:other_idx})?\label{q5}
\end{enumerate}

\subsection{Experiment Setup} \label{subsec:expr:setup}
In our experiments, we construct indices using bigrams, trigrams, and 4-grams. The index is physically represented by a $k$-bit vector using a C++ bitset for each log line. The bit-vector indicates the presence of the top $k$ $n$-grams.  For regex matching, the regular expressions are first parsed to derive all $n$-grams intersecting with the $n$-grams used in the index. The $n$-grams in the intersection are encoded in a length $k$ bit mask, facilitating the initial filtering of log lines using the index. Only the log lines that pass the filtering are used as input for the full regex matching process using RE2~\cite{re2} which is widely used in the industry. We use \blare{}~\cite{blare} with RE2 as the backend supporting regex engine, as the performance baseline {, which is faster than RE2 without indexing}. 
For simplicity, all experiments are single-threaded and in the main memory setting.

\subsubsection{Workloads.} 

\begin{table}
\caption{Statistical information of bigrams in the literal component of regular expression queries in the workload.}
\label{table:bigram_stats}
\vspace{-1em}
\scalebox{0.99}{
\begin{tabular}{ c| c| r r r r r } 
\toprule
\multirow{2}{*}{Workload} & \multirow{2}{*}{\# Total Bigrams} & \multicolumn{5}{c}{\% Occurrence in Regexes}\\ \cline{3-7}
& & min & Q1 & mid & Q3 & max \\ \hline
DB-X & 762 & 0.8& 0.8& 3.0& 11.4& 89.4 \\\hline
Sys-Y & 138 & 5.9& 5.9& 11.8& 17.6& 41.2 \\\hline
US-Acc. & 18 & 25.0& 25.0& 25.0& 50.0& 75.0 \\
 \bottomrule
\end{tabular}}
\vspace{-1em}
\end{table}

Our experiments employ three real-world datasets and their respective query workloads. Two datasets include logs produced by production database systems, whereas the third is text-oriented. For each workload, detailed statistical information about the literal elements in regexes is shown in~\Cref{table:bigram_stats}.

\introparagraph{\wlsqlfull{} Workload} A 13.5 GB dataset with 8,941 regexes and 101,876,733 log lines from DB-X was obtained from a well-known cloud provider. These data were utilized in log analysis tasks by various data scientists and engineers. {When examining all the regex queries in the original DB-X workload, the majority of the queries had no matches on the dataset, which represents logs from a specific time duration. Thus, for a realistic representation of the workload, we consulted an expert data scientist working in this domain and sampled regexes favoring those with matches on the dataset. For completeness, we report the experimental results for the entire set of 8,941 regexes in Section~\ref{subsec:expr:db_full}.
} We picked 132 regexes from their workloads and evaluated the regexes on the dataset to construct \eat{logs in our workload,} \wlsqlfull{} (\textbf{DB-X}). In this workload, a log line has an average of 138.5 characters. The regexes used in this task are characterized by an average of 105 literal characters per regex and 39.9 characters per literal component. From the literal components of the regexes, 762 unique bigrams were extracted, with the majority appearing in just one regex.

\eat{
\introparagraph{\wlsqlfull{} Workload} We obtained a dataset from a well-known cloud provider that contains 8,941 regexes and 101,876,733 DB-X log lines totaling 13.5 GB that were used in various data scientists' and engineers' log analysis operations. We picked and tested 132 regexes at random from the logs to create the workload, which we called DB-X.  In this workload, a log line typically has 138.5 characters. 
The regexes in this task are distinguished by having an average of 105 literal characters per regex and 39.9 characters per literal component. There are 762 unique bigrams extracted from literal components in the regexes, and majority of them occur in only one regex.
}

\introparagraph{\wlkusto{} Workload} This dataset is 100 GB in size and consists of 890,623,051 log lines generated by System-Y, a production data exploration tool. The associated workload includes 17 regexes used by data scientists for analysis tasks. The log lines have an average length of 116.2 characters. We also use the shorthand \textbf{Sys-Y} to denote this workload. Compared to \textbf{DB-X}, the regexes in this workload often have fewer literal components per regex and shorter literal components overall. We get 138 unique bigrams from the literal components of the regex queries. Similarly, most of the bigrams are unique and appear only in one regex.

\introparagraph{\wltraffic{} Workload} The dataset, sourced publicly and collected by Moosavi et al.~\cite{Moosavi19}, has 1.08 GB of data and includes 2,845,343 records of traffic accidents that occurred in the United States between February 2016 and March 2019. For our objectives, we used the 4 regexes described in their paper along with the \textit{accident description} strings from the dataset. There are 66.6 characters on average per line. We extracted 18 unique bigrams from the query set, and among them, slightly more than half occurs in only on regex, and one of the bigrams occurs in 3 regexes.

\subsubsection{Methodology} We compare a number of metrics, such as index building time, index size, and regex matching execution time, across various index configurations. The time required to extract and select the $n$-grams and then build the index is used to compute the index construction time. We measure the bit-vectors' overall size, including any padding, as the size of the index. The cumulative time required for each regex in the workload to match against all log lines in the dataset is measured for determining the workload running time, using indices when applicable. The runtimes are calculated using a trimmed mean of 10 runs, which leaves out the greatest and lowest runtimes. 

Indices were created using the bigram, trigram, or 4-gram that appeared the most frequently in the literal components of the workload queries. Depending on the workload and experiment design $n$-grams with $n$ varying from $4$ to $512$ may be used to build the index. Additionally, indices were created with index granularity of 1, 8, 64, 192, 256, or 512 lines. Detailed methodologies for individual experiments will be described in the respective subsections.\eat{ prior to presenting the results.}

\subsubsection{Hardware.} We run all experiments on an Azure Standard\_E32-16ds\_v5 machine with Intel(R) Xeon(R) Platinum 8370C CPU @ 2.80GHz with 16 vCPUs, 256 GB memory, and 1 TB hard disk. Our experiment code is written in C++17 and compiled with the -O3 flag. We used RE2 (release version 2022-06-01) and the version of BLARE (state-of-the-art system for regex processing over log data) retrieved on 09-01-2023. 

\eat{
\subsection{Experiment Results}

\begin{itemize}
    \item Compare across: query performance, index construction, and index size with different type of $n$-grams
    \item Compare across: query performance, index construction, and index size with different number of $n$-grams 
    \item Compare the performance by indexing $k$ most popular trigrams in the regex workload vs. most popular character trigrams in English language.
    \item Performance of coordinately search and update the index; compare performance difference of search only and search + update 1 ngram; hope to get small update overhead - no need
\end{itemize}
}

\begin{figure}[!tp]
\begin{subfigure}{0.44\columnwidth}
  \centering
  \includegraphics[width=0.75\columnwidth]{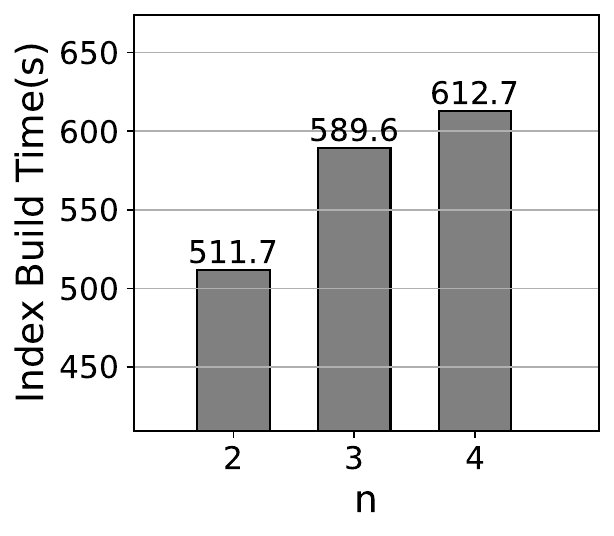} 
    \vspace{-1em}
  \caption{Index Construction}
  \label{fig:type_build}
\end{subfigure}
\hfill
\begin{subfigure}{0.44\columnwidth}
  \centering
  \includegraphics[width=0.75\columnwidth]{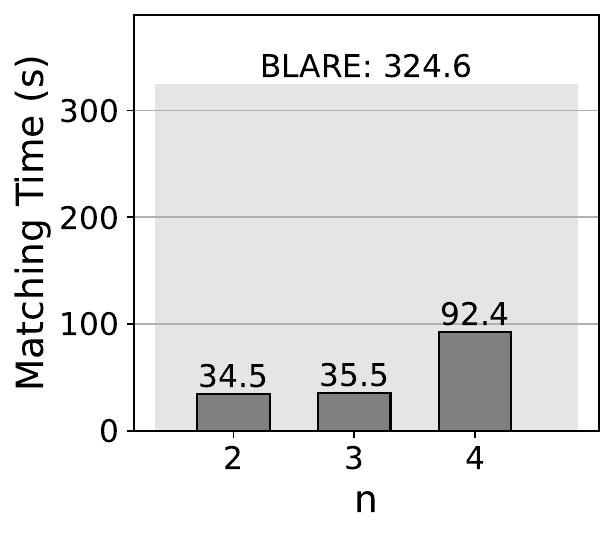} 
    \vspace{-1em}
  \caption{Regex Matching}
  \label{fig:type_match}
\end{subfigure}
\vspace{-.6em}
\caption{Comparing the impact of different types of $n$-grams on index construction time and matching time on indexes with the top 64 most frequent $n$-grams in workload queries.}
\label{fig:gram_type}
\end{figure}

\subsection{Type of $N$-Gram}\label{subsec:qgram_type}
In this section, we address \ref{q1} by evaluating our strategy through a comparison of different types of $n$-gram used for index building. Specifically, we compare the query performance improvement and index construction overhead in terms of space and time. The top 64 most frequent $n$-grams, ranging from bigrams to 4-grams, were extracted from the regular expression queries of \wlsqlfull{} workload. For each log line, a bit-vector of size 64 was constructed to denote the presence (or absence) of these $n$-grams. 
\eat{The construction time was measured from the onset of the indexing process up to the complete representation of the bit-vector for all log lines. }

\subsubsection{Index Construction Overhead}
The bit-vector index size remains constant across different types of $n$-grams, as {each $n$-gram is represented in an index entry with one bit regardless of the value of $n$.}
\eat{a 64-size bit-vector is used for each log line.} Specifically, the size depends solely on the number of $n$-grams, $k$, ($k=64$ in this case) and the total number of log lines. Consequently, {for a fixed $k$,} the space overhead of bit-vector index is directly proportional to the number of {log lines.}
\eat{in log lines rather than the type or length of $n$-grams used.}

In~\Cref{fig:type_build}, we present a summary of the index construction time for the \wlsqlfull{} workload, comparing bigrams, trigrams, and 4-grams.
For bigrams, the index took 511.7 seconds to build. Bigrams are short and simple, and their extraction and subsequent indexing are fairly quick. Building the trigram index took slightly longer, at 589.6 seconds. 
\eat{The result of 4-grams, having largest $n$ among bigram, trigram, and 4-gram, in~\Cref{fig:type_build},
shows that the length of $n$-grams has a subtle impact on the extraction time and, consequently, the overall building time with an index construction time of 612.7 seconds. } 
The 4-grams, being the longest among the three types of $n$-grams, took the longest time to construct the index at 612.7 seconds. This shows that the length of the $n$-grams subtly impacts
\eat{the extraction time and, consequently, }
the overall index construction time.
The subtle difference in time for different $n$ can be traced back to the number of unique $n$-character sequences extracted from the log and the hashing strategy of character tuples.
As the length of the $n$-grams increases, the number of unique $n$-grams also increases. This translates to 
{an increase in the number of hashmap probing and higher probability of hash collision {(recall that our index contruction algorithm uses hashmaps on~\autoref{line:S} and~\autoref{line:T} of~\Cref{algo:index_construct}}), and consequently, }
an increase in the time required for index construction. 

\eat{\begin{table}[t]
\captionof{table}{Comparison of mean percentage of log lines processed by regex matching engine after filtering with bit-vector index for different types of $n$-grams used in the index.}
\centering
\begin{tabular}{ c| c| | r | r | r  }\hline
\multicolumn{2}{c||}{} & \multicolumn{3}{c}{\% Log Passed Index}\\ \hline
\multicolumn{2}{c||}{$n $} & 2 & 3 & 4 \\ \hline
\multirow{2}{*}{$k$} & 64 & 0.63 & 0.58 & 4.99 \\ \cline{2-5}
& 128 & 0.12 & 0.35 & 3.13 \\ \hline
\end{tabular}
\label{tab:log_filter}
\vspace{-1em}
\end{table}}

\begin{table}[t]
\captionof{table}{Comparison of mean percentage of log lines processed by regex matching engine after filtering with bit-vector index for different types of $n$-grams used in the index. \eat{\sd{"log lines matched by regex matching engine" means log lines processed by the engine, correct? Not actually producing a match? If yes, consider using 'processed' to reflect this difference.}}}
\vspace{-.5em}
\centering
\begin{tabular}{ c| | r | r | r  }\hline
$n$ & 2 & 3 & 4 \\ \hline
Log \% filtered through the index & 0.63 & 0.58 & 4.99 \\ \hline
\end{tabular}

\label{tab:log_filter}
\end{table}

\subsubsection{Matching Time with Index}

To compare the query efficiency of our indexing approach {with varying $n$-gram sizes}, we benchmark the performance of regular expression matching on \wlsqlfull{} workload using indices constructed with bigrams, trigrams, and 4-grams. We report the time taken to match all 132 regexes against the log dataset using each $n$-gram index type. From ~\Cref{fig:type_match}, we found notable differences in matching times. 
Using the bigram index, the full set of workload queries runs in 34.5 seconds.
The trigram index results in a slightly slower matching time of 35.5 seconds.
The matching process takes significantly longer with the 4-gram index, at 92.4 seconds, compared to the bigram and trigram indices.
\eat{Using the 4-gram index, the matching process takes 92.4 seconds, which is significantly longer than the runtime for bigram index and trigram index.} 
Despite the increase in time with 4-grams, all indexed methods were substantially faster than the baseline of 324.6 seconds, which is the time taken to match using the state-of-the-art regular expression matching framework without using an index. 

After leveraging the indices to filter potential matches, we summarize the mean percentage of log lines that remain for each regex in~\Cref{tab:log_filter}.
The bigram and trigram indices leave 0.63\% and 0.58\% of log lines respectively, showing similar filtering efficacy.  {The percentages suggest that trigram index has slightly higher actual \textit{benefit}, but both bigrams and trigrams are effective in the task.} This is reflected in their relatively similar matching times of 34.5 and 35.5 seconds.
However, the 4-gram index retains a larger subset of log lines at 4.99\%{, signifying a much lower actual \textit{benefit}}, despite its longer construction time. 
\eat{This bigger pool of data demanding regex matching aligns with its substantially longer matching time of 92.4 seconds, indicating a direct correlation between the volume of data sent to the regex engine and the resulting total matching time. }
This larger data pool for regex matching results in a longer matching time of 92.4 seconds, indicating a direct correlation between the volume of data sent to the regex engine and total matching time.
The choice of $n$-gram length can significantly impacts the construction overhead and the performance of the subsequent regex matching process.

\eat{
From the table, we can see that 
bigram index left about 0.48\% of log lines for further matching using the regex engine.
trigram index retained approximately 0.58\%.
Notably, the 4-gram index left a significantly larger subset, 4.99\% of log lines, for the regex matching process. The coverage difference contributes to the longer matching time for index built with 4-grams compared to bigrams and trigrams. 

\TODO{Due to Coverage of the regex and/or $n$-gram extraction performance difference as shown by the index construction time difference? add more analysis.}
}

\subsection{Number of $N$-Gram}\label{subsec:qgram_num}
\begin{figure}[!tp]
 \centering
\includegraphics[width=0.75\columnwidth]{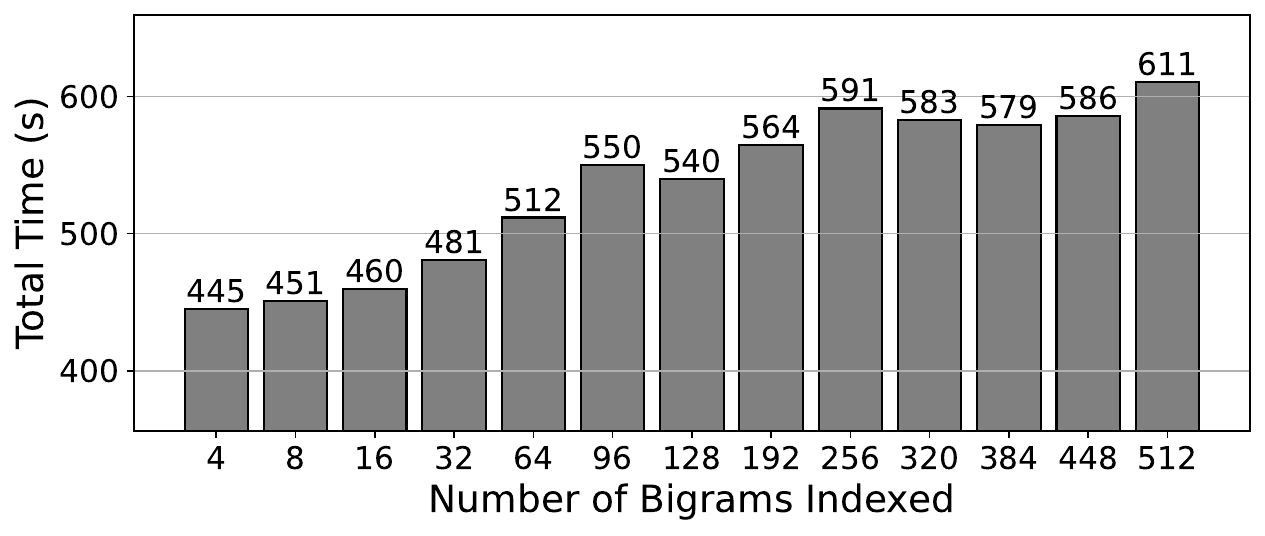}  
\vspace{-1em}
\caption{Comparing the impact of different numbers of $n$-grams on index construction time of the indices. Uses top-$k$ most frequent $n$-grams in workload queries.}
\label{fig:num_ngram_build}
\end{figure}

\eat{
\begin{figure}[!tp]
 \centering
\includegraphics[width=0.85\columnwidth]{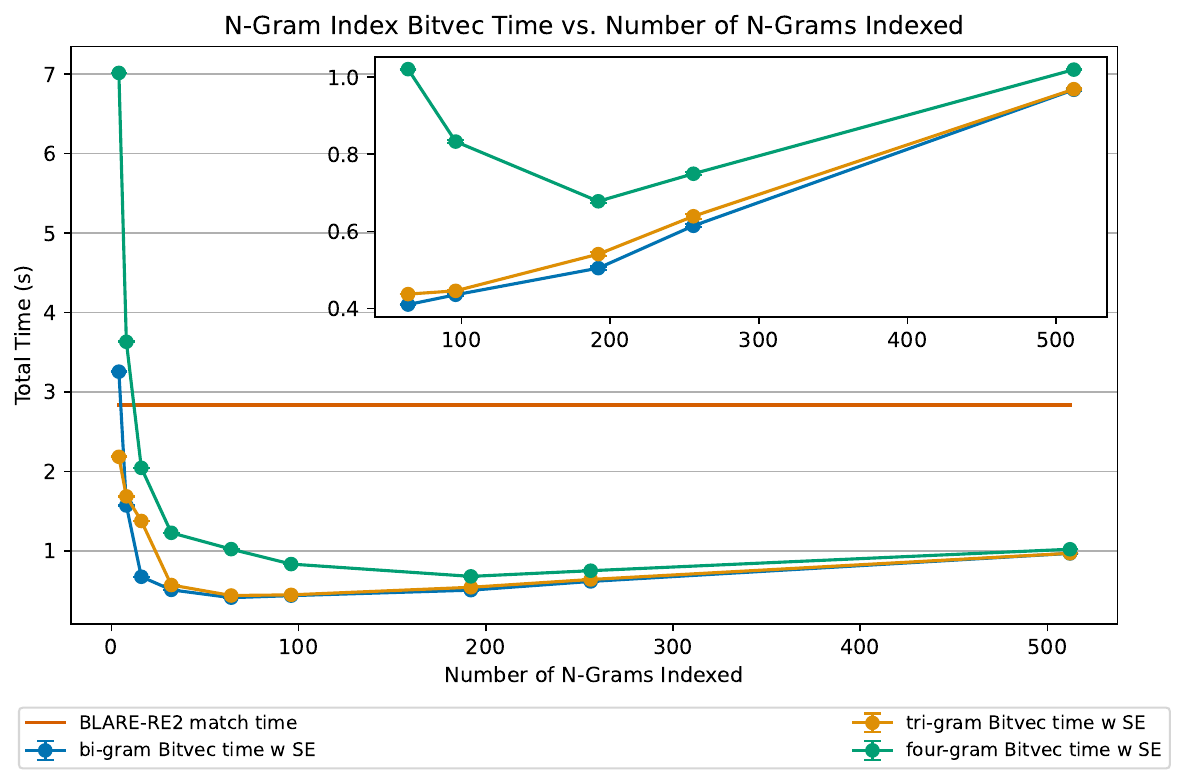}  
\caption{Comparing the impact on matching time for different numbers of $n$-grams indexed. Uses top-$k$ most frequent $n$-grams in workload queries.}
\label{fig:num_ngram_match}
\end{figure}
}

\eat{
\begin{figure}[!tp]
 \centering
\includegraphics[width=0.85\columnwidth]{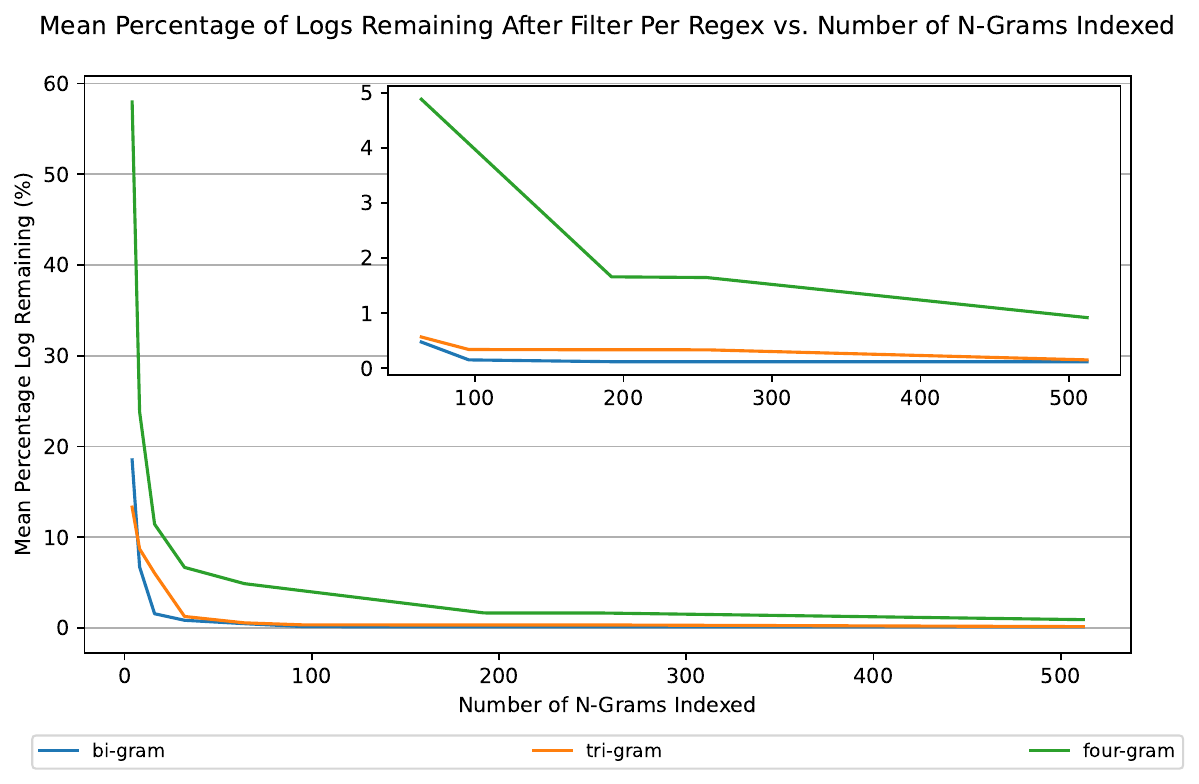}  
\caption{Comparing the mean percentage of log lines processed by regex matching engine after filtering with bit-vector index. Uses top-$k$ most frequent $n$-grams in workload queries.}
\label{fig:num_ngram_perc_filter}
\end{figure}
}

In this section, we answer \ref{q2} by comparing the number of $n$-grams used for index building. Specifically, we compare the performance gain and index construction overhead across using the most frequent $k$ $n$-grams in \wlsqlfull{} workload, where choices of $k$ are 4, 8, 16, 32, 64, 96, 192, 256, and 512. We use bigrams, trigrams, and 4-grams for all $k$. Similar to the previous experimental setting, one bit-vector of size $k$ in the index corresponds to one log line.

\subsubsection{Index Construction Overhead}

\begin{table}[t]
\captionof{table}{Comparison of the size of the bit-vector index for different $k$ values.} 
\vspace{-0.7em}
\centering
\begin{tabular}{ c| |  r  |  r |  r |  r |  r |  r |  r |  r   }\hline
$k$ & 4 - 64 & 96 - 128 & 192 & 256 & 320 & 384 & 448& 512\\ \hline
Size (GB) & 0.8 & 1.5 & 2.3 & 3.0 & 3.8 & 4.6 & 5.3 & 6.1 \\ \hline
\end{tabular}
\vspace{-1em}
\label{tab:index_size}
\end{table}

In this subsection, we evaluate the impact of the number of $n$-grams, $k$, on the index size and the construction time. The number of bigrams selected determines the length of each bit-vector representation for a log line, thereby directly affecting the size and construction time of the index.

\Cref{tab:index_size} presents the comparison result of index sizes for different numbers of bigrams used. Due to word-size and padding in C++ bitset, the index size stays constant when indexing with 4 to 64 bigrams. Packing multiple index entries together can further reduce the index size for storage and loading.
The index size grows from 0.8 GB for $k=64$ to 1.5 GB for $k=512$.
This trend shows that integrating more bigrams into the index encapsulates more information, increasing the filtering power by eliminating more unmatched log lines, at a cost of more space.
We also analyze the index construction time for various numbers of bigrams, as depicted in~\Cref{fig:num_ngram_build}. 
Construction time marginally increases from 445.2 seconds with 4 bigrams to 610.6 seconds with 512 bigrams.
Thus, a $100 \times$ increase in the number of bigrams results in just a $30\%$ increase in index construction time. The slight variation in index construction time, despite the increase in bigrams, highlights the low cost of the bit-setting operation {and \sys{}'s
\eat{. The minor growth in construction time, relative to the index size growth, underscores our approach's }
ability to balance the filtering power with index construction overhead.}

\begin{figure}[!tp]
 \centering
\includegraphics[width=0.75\columnwidth]{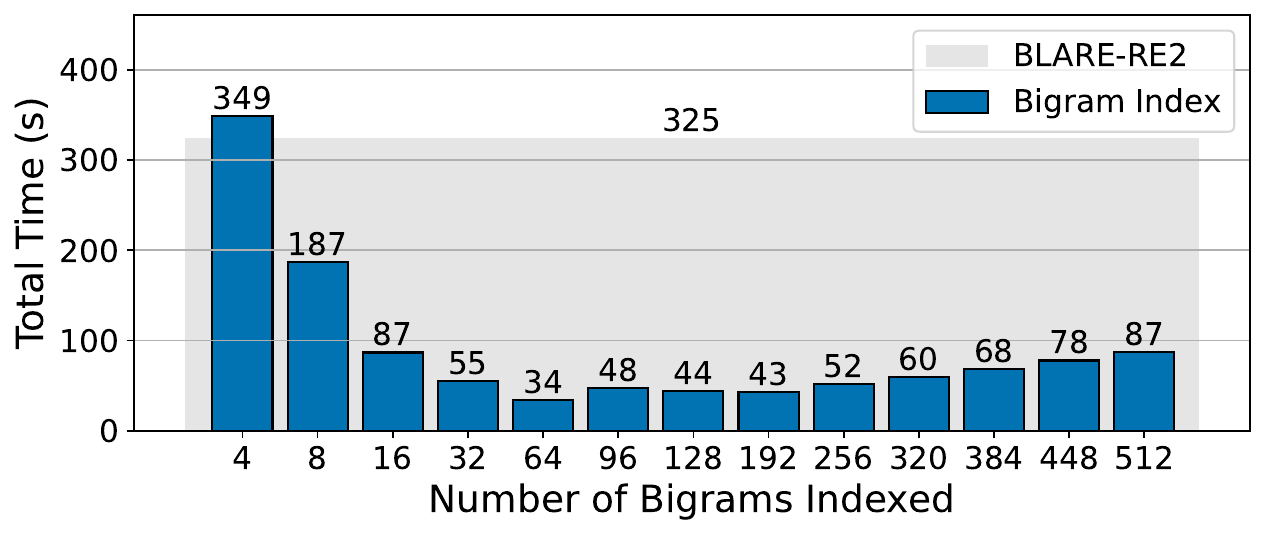}  
\vspace{-1em}
\caption{Comparing the impact on matching time for different numbers of $n$-grams indexed. Uses top-$k$ most frequent $n$-grams in workload queries.}
\label{fig:num_ngram_match}
\end{figure}

\begin{figure}[!tp]
 \centering
\includegraphics[width=0.75\columnwidth]{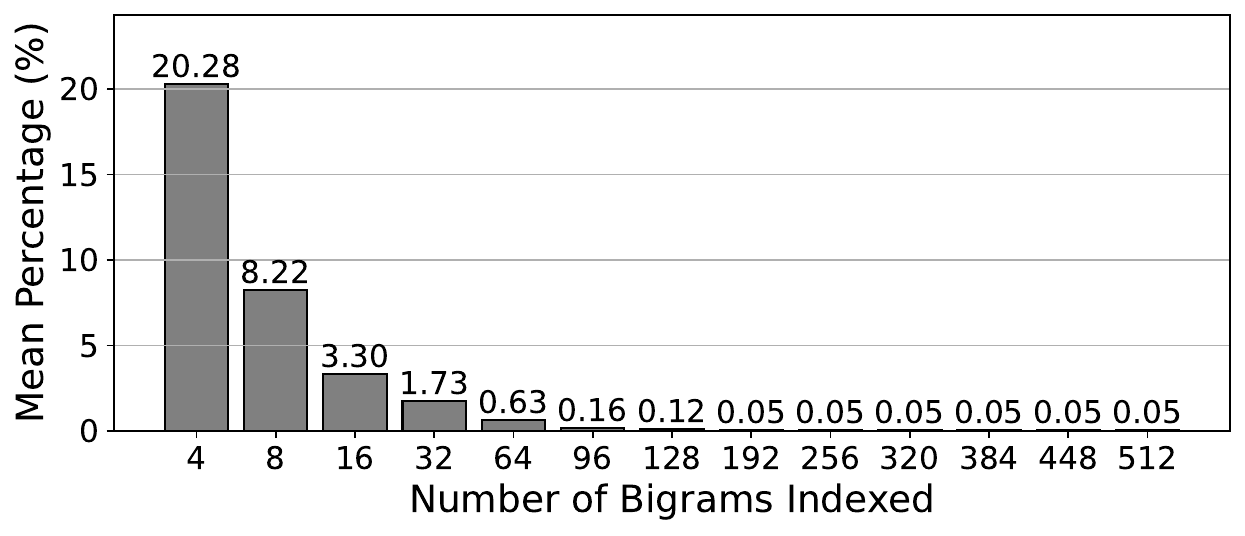}  
\vspace{-1em}
\caption{Comparing the mean percentage of log lines processed by regex matching engine after filtering with bit-vector index. Uses top-$k$ most frequent $n$-grams in queries.}
\label{fig:num_ngram_perc_filter}
\vspace{-1em}
\end{figure}

\subsubsection{Matching Time with Index. }

Now we compare the overall matching time for the \wlsqlfull{} workload across different numbers of bigrams used to construct the bit-vector index. \Cref{fig:num_ngram_match} shows the regex matching time corresponding to each $k$, the number of bigrams used.

The baseline, \baseline{}, which operates without any indexing, has a matching time of 324.6 seconds for the given workload. An index constructed with 4 bigrams results in a matching time of 348.9 seconds.
\eat{A notable reduction is observed as the matching time drops to 187.1 seconds with the use of an index built with 8 bigrams, eventually dropping to a minimum of 34.5 seconds when 64 bigrams are employed. }
The matching time drops to 187.1 seconds with an index built with 8 bigrams and further decreases to a minimum of 34.5 seconds when using 64 bigrams in the index.
Beyond 64 bigrams, the matching time gradually increases, reaching 87.4 seconds for the index constructed with 512 bigrams.
~\Cref{fig:num_ngram_perc_filter} provides a visual representation of the percentage of log lines remaining after index lookup for subsequent regex matching.

From 4 to 64 bigrams, the matching time drops significantly as the proportion of log lines passing the index filter decreases from 20.28\% to 0.63\%.
\eat{The substantial reduction in matching time from the use of 4 to 64 bigrams can be directly attributed to the steep decrease in the proportion of log lines that pass the index filter, which decreases from 20.28\% to 0.63\%. }
Consequently, fewer log lines are subjected to the expensive regex matching. 
A more moderate decline of the percentage is observed as $k$ increases from 64 to 128. It drops to 0.12\% when there is a rise in matching time in~\Cref{fig:num_ngram_match}. 
The gains from reduced regex matching {provided by the decreased number of log lines after index filtering} are offset by the incremental overhead introduced by the increasing complexity and size of the index. 
While the index lookup remains a low-cost operation, the increasing number of bigrams introduces additional overhead for bit operations across multiple bytes, affecting the overall performance of the matching process.
Therefore, as the percentage stabilizes at approximately 0.05\% for $k$ from 192 to 512, the matching time increases instead.

In the \wlkusto{} workload, the blue bars in~\Cref{fig:english_kusto} demonstrate that the index created using bigrams from queries outperforms the baseline after indexing with 16 or more bigrams. The overall query runtime decrease sharply as the $k$ increase from 4 to 16. The performance gradually plateaus as more bigrams are indexed beyond 16. {When $k=64$, it achieves more than $7\times$ query performance improvement compared to the baseline.} \Cref{fig:english_traffic} shows that the {$n$-gram index} provide performance {gain} for \wltraffic{} workload {for $k>=2$, and the gain increases as more bigrams are indexed, achieving a near $5\times$ speedup when $k=16$.}
With only 4 queries containing short literals, each indexed bigram appears in at least 25\% of the queries, having similar filtering power for the log lines.

We observe that although increasing the number of bigrams indexed improves performance by reducing the number of log lines subjected to regex matching, it also requires careful consideration on the runtime overhead introduced by the additional bigrams.
Balancing these factors is crucial for achieving optimal matching performance in regex log analysis.

\eat{Existence of Prior Knowledge}
\subsection{Unknown Workload}\label{subsec:qgram_english}
In this section, we answer \ref{q3} by
{showing if \sys{} can improve the log analysis workload performance without knowing the workload in advance.}
\Ling{Due to the incapability of existing n-gram selection methods in handling unknown log processing workloads, we compare our index performance to the baseline where the queries are assumed to be and usable for index construction.}
\eat{examining the impact of prior knowledge of workload characteristics on the performance of the index.}
Specifically, we conduct a direct comparison of the matching performance achieved with indices built from the most frequent bigrams in the queries (i.e. having prior information about string literal distribution in the queries){, \qidx{},} to those created using the most common bigrams in an English corpus~\cite{english_bigram} (i.e. the query workload is not known){, \eidx{}} , across three real-world workloads: \wlsqlfull{}, \wlkusto{}, and \wltraffic{}.
\eat{\sd{We should use \qidx{} and \eidx{} in the legend for plots in Figure 8.}}
\eat{Specifically, for 3 real-world workloads \wlsqlfull{}, \wlkusto{}, and \wltraffic{}, we directly compare the construction overhead and matching performance achieved with indices built from the most frequent bigrams in the queries  against those created using the most common bigrams in an English corpus~\cite{english_bigram} . }

\eat{\subsubsection{Index Construction Overhead.}}

Given that the indices are constructed using identical data structures and an equivalent number of bigrams, the sizes of the indices remain the same. 
\eat{Moreover, we observe comparable time frames for the construction of indices across both sets of bigrams. }
This similarity in construction time can be attributed to the same index structure and bigram count.
\eat{, underscoring the consistency of the indexing process regardless of the bigram set used.}

\eat{
\begin{figure}[!tp]
 \centering
\includegraphics[width=\columnwidth]{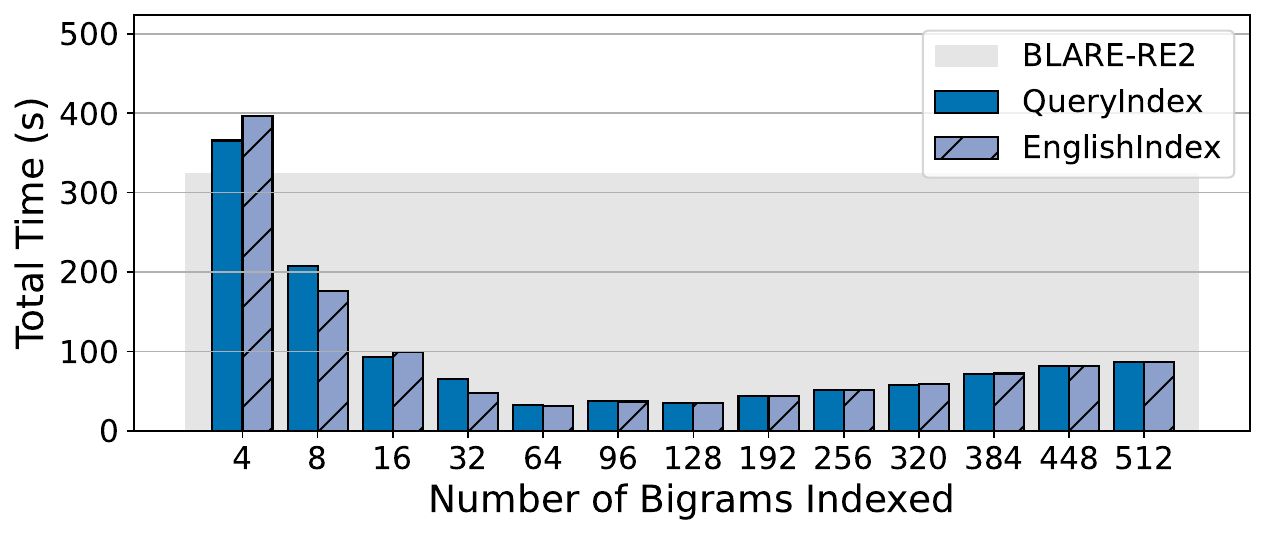}  
\caption{Comparing the impact on matching time for different numbers of bigrams indexed. Uses top-$k$ most frequent bigrams that appeared in the workload query and top-$k$ most frequent bigrams appeared in English literature.}
\label{fig:prior_stats_match}
\end{figure}

\begin{figure}[!tp]
 \centering
\includegraphics[width=\columnwidth]{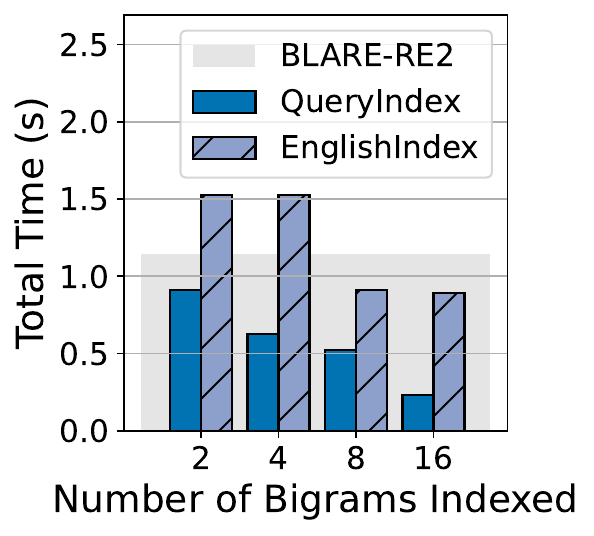}  
\caption{Comparing the impact on matching time for different numbers of bigrams indexed. Uses top-$k$ most frequent bigrams appeared in the workload query and top-$k$ most frequent bigrams appeared in English literature.}
\label{fig:prior_stats_match}
\end{figure}
}

\begin{figure}[!tp]
\begin{subfigure}{.95\columnwidth}
  \centering
  \includegraphics[width=.75\columnwidth]{figs/Bigram_english_SearchTime_errbar_.pdf}  
  \vspace{-0.5em}
  \caption{\wlsqlfull{} Workload}
  \label{fig:english_sqlserver}
\end{subfigure}\\
\vspace{0.8em}
\begin{subfigure}{.55\columnwidth}
  \centering
  \includegraphics[width=.8\columnwidth]{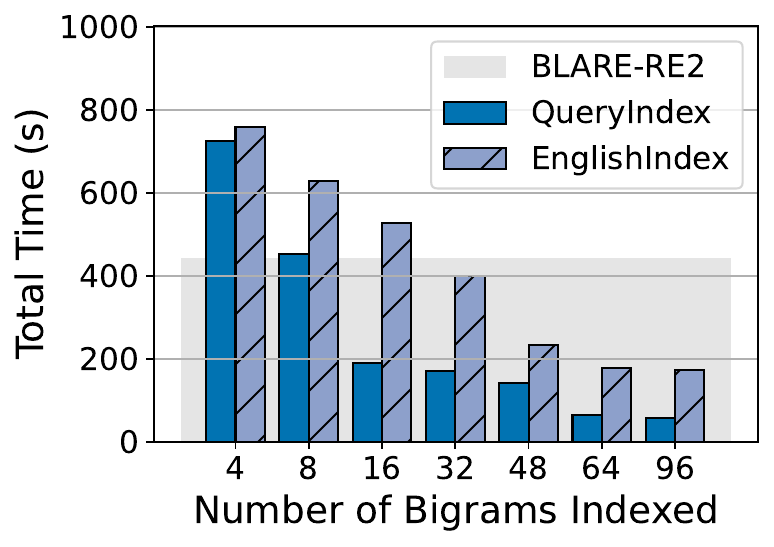}  
  \vspace{-0.5em}
  \caption{\wlkusto{} Workload}
  \label{fig:english_kusto}
\end{subfigure}
\hfill
\begin{subfigure}{.43\columnwidth}
  \centering
  \includegraphics[width=.8\columnwidth]{figs/Bigram_english_Traffic_SearchTime_errbar_.pdf}  
  \vspace{-0.5em}
  \caption{\wltraffic{} Workload}
  \label{fig:english_traffic}
\end{subfigure}
\vspace{-0.5em}
\caption{Comparing the impact on matching time for different numbers of bigrams indexed. Uses top-$k$ most frequent bigrams appeared in the workload query and top-$k$ most frequent bigrams appeared in English literature.}
\vspace{-.5em}
\label{fig:prior_stats_match}
\end{figure}

\subsubsection{Matching Time with Index.}
We explore the performance differences between indices created using the most frequent bigrams from the queries, \qidx{}, and those from an English corpus, \eidx{}. Examining~\Cref{fig:prior_stats_match}, we observe similar performance in workload queries for indices built with the two sets of bigrams. Also,~\Cref{fig:prior_stats_match} illustrates the impact of using different $k$, the numbers of bigrams for indexing, across the three workloads. 
Both indices show a consistent decrease in overall runtimes as the $k$ indexed increases, particularly for workloads \wlkusto{} and \wltraffic{}.
\eat{The decrease in overall running time for both indices as the quantity of bigrams indexed increases is consistent for both indices and for workloads \wlkusto{} and \wltraffic{}. }
For \wlsqlfull{} workload, runtimes using indices from the two sources follow a similar trend that decreases sharply as the $k$ increases, reaches a minimum at 64 bigrams for indexing, and then gradually increases as the number of bigrams indexed increases.

The \wlsqlfull{} workload exhibits similarity with the English language in terms of bigram distribution, as it has the largest number of queries, each of which contain long literal components of human-readable text. From~\Cref{table:bigram_stats}, we can see that 89\% of the log entries contain the most frequently occurring bigram, and 25\% of all 762 bigrams have frequency percentages above 11\%. \Cref{fig:english_sqlserver} presents the performance of the two indices built with a varying number of bigrams. When indexing 8 or more bigrams, both indices outperform the baseline performance {by as much as $10\times$ with 64 bigrams}, indicating that performance improvement can be achieved with a relatively small number of bigrams. We notice a minor difference in performance between the two indices when using bigrams from 4 to 64. This discrepancy in performance becomes negligible as the number of indexed bigrams continues to rise.

The \wlkusto{} workload consists of fewer queries, and there is noticeable dissimilarity across the majority of them. From~\Cref{table:bigram_stats}, we can see that half of the 138 bigrams appear in more than 11.8\% of the queries, and the most frequently occurring bigram is found in 41.2\% of the queries. In \Cref{fig:english_kusto}, \eidx{} needs at least 32 bigrams to surpass the baseline runtime. In contrast, \qidx{} requires only 16 bigrams {and achieves close to $8\times$ speedup when $k$ reaches 96.} The performance results of \wlkusto{} workload exhibits a consistent trend in which \qidx{} outperforms \eidx{}. While the numbers are similar when indexing with 4 bigrams, the performance gap widens as the $k$ increases to 8 and 16. Despite that \eidx{} has lower filtering power than \qidx{}, it still achieves a good speedup of $2.6\times$ when indexing with 96 bigrams.

The \wltraffic{} workload has only 4 queries with short literal components. With only 18 bigrams extracted, each bigram is substantially important as even the most infrequent one appears in a quarter of all the queries as shown in~\Cref{table:bigram_stats}. \qidx{} continuously beats the baseline. Conversely, the top-6 most frequent bigrams in English are absent from the queries, causing \eidx{} to only surpass the baseline after indexing 8 or more bigrams. \qidx{} outperforms \eidx{} across all values of $k$. For \qidx{}, the runtime gradually decreases as the number of bigrams increases. \eidx{} also experiences a reduction in workload runtime, although meaningful speed improvements are only shown when we include one or more new bigrams that also exist in the queries of the workload.

{With the result, we show that \sys{} is resilient to unknown workload, taking advantage of the fact that logs and literal components in log analysis tasks are mostly human-readable English text.} 
This is most significant when the workload has a large query set with a large portion of text literals. The performance gain for using bigrams according to their frequency in the English corpus is less significant when the literal components in the query set have bigram distribution dissimilar to that of the English language. Even so, \sys{} can still achieve satisfactory performance gain by increasing the number of bigrams used for indexing.

\subsection{Index Granularity}\label{subsec:idx_granu}

\begin{figure}[!tp]
\begin{subfigure}{0.45\columnwidth}
  \centering
  \includegraphics[width=.85\columnwidth]{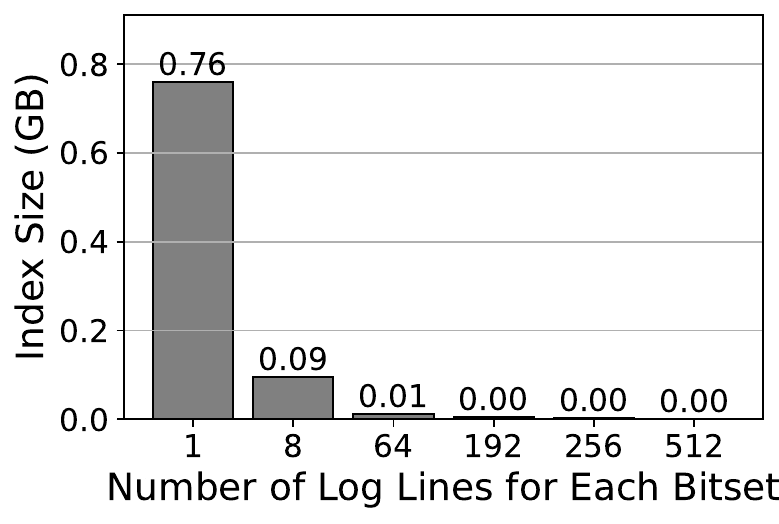}  
  \vspace{-0.5em}
  \caption{Index Size - 64}
  \label{fig:granu_64_index_size}
\end{subfigure}
\hfill
\vspace{0.5em}
\begin{subfigure}{0.45\columnwidth}
  \centering
  \includegraphics[width=.85\columnwidth]{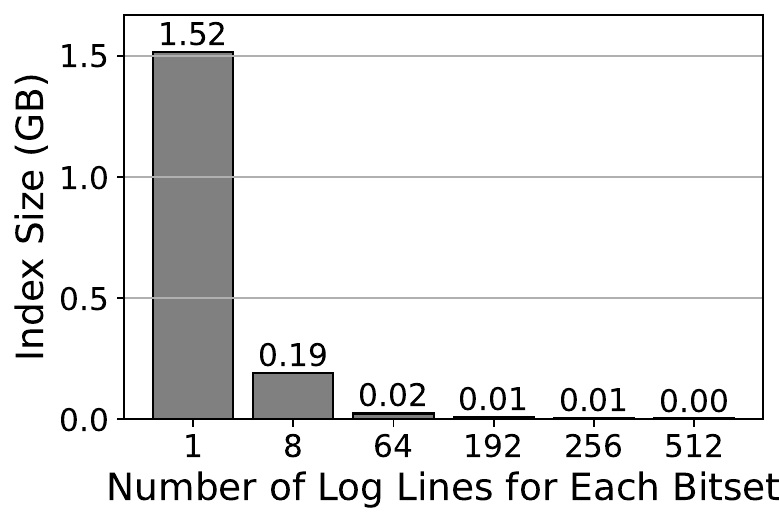}  
  \vspace{-0.5em}
  \caption{Index Size - 128}
  \label{fig:granu_128_index_size}
\end{subfigure}
\vspace{0.5em}
\begin{subfigure}{0.45\columnwidth}
  \centering
  \includegraphics[width=.85\columnwidth]{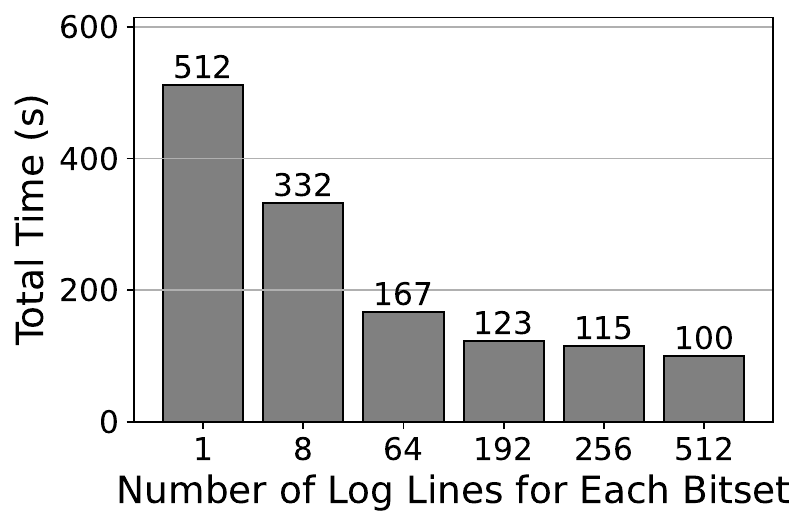}  
  \vspace{-0.5em}
  \caption{Index Build Time - 64}
  \label{fig:granu_64_build_time}
\end{subfigure}
\hfill
\vspace{0.5em}
\begin{subfigure}{0.45\columnwidth}
  \centering
  \includegraphics[width=.85\columnwidth]{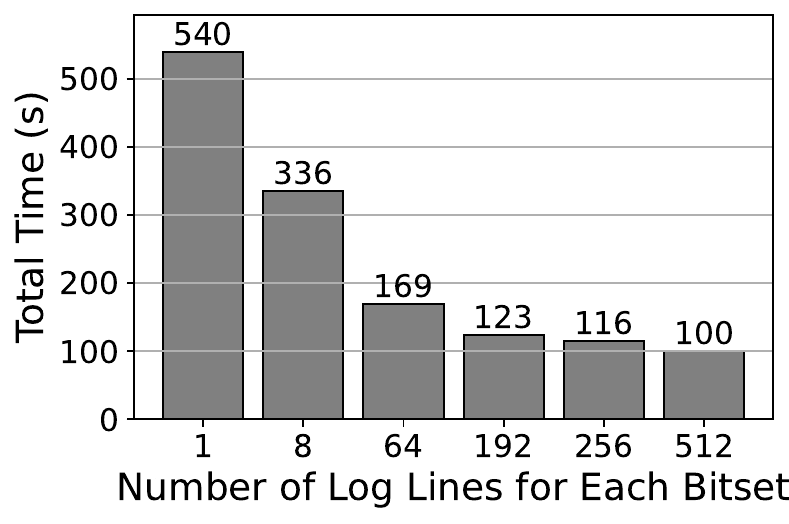}  
  \vspace{-0.5em}
  \caption{Index Build Time - 128}
  \label{fig:granu_128_build_time}
\end{subfigure}
\vspace{0.5em}
\begin{subfigure}{0.45\columnwidth}
  \centering
  \includegraphics[width=.85\columnwidth]{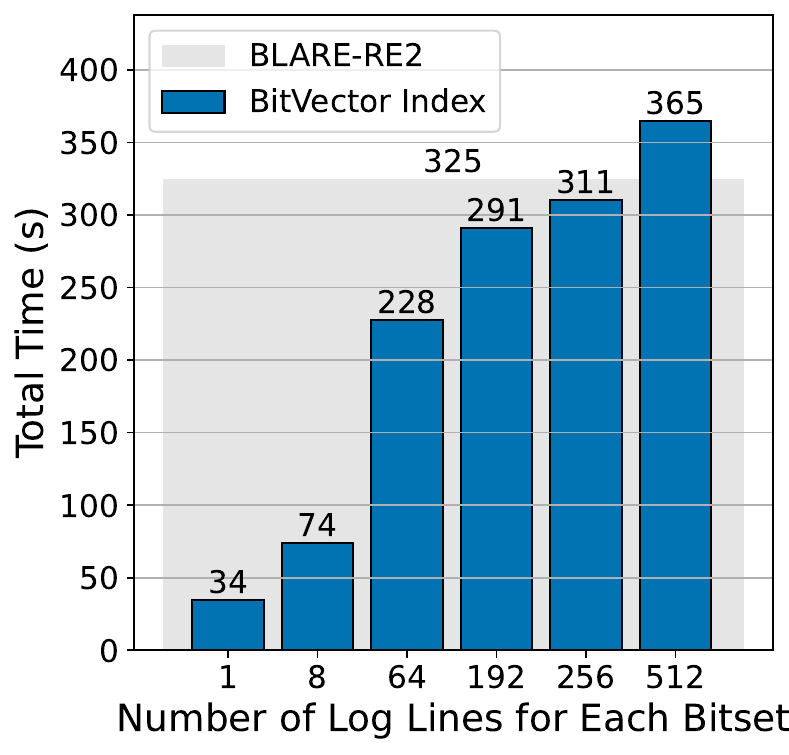}  
  \vspace{-0.5em}
  \caption{Matching Time - 64}
  \label{fig:granu_64_match}
\end{subfigure}
\hfill
\vspace{0.5em}
\begin{subfigure}{0.45\columnwidth}
  \centering
  \includegraphics[width=.85\columnwidth]{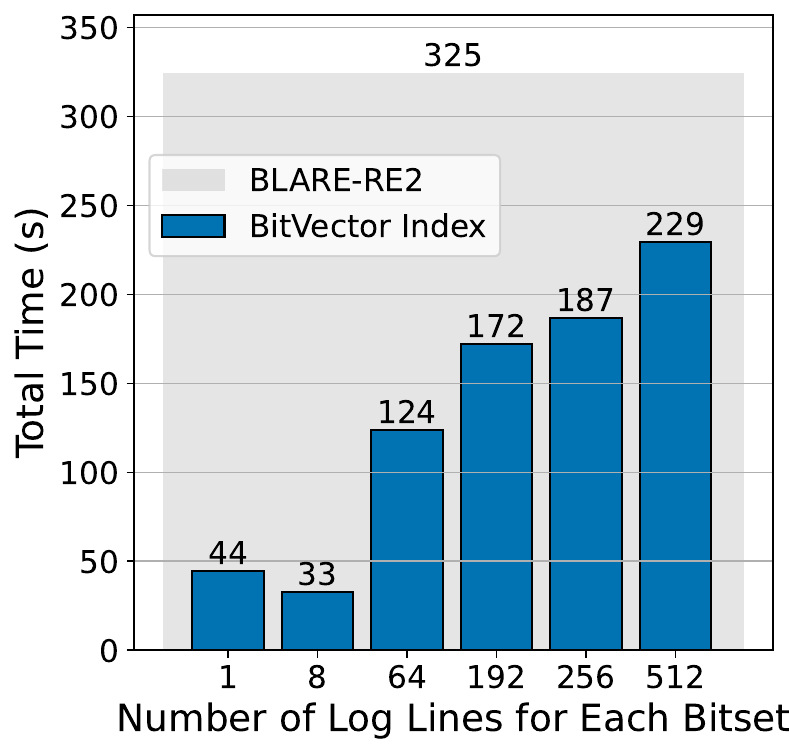}  
  \vspace{-0.5em}
  \caption{Matching Time - 128}
  \label{fig:granu_128_match}
\end{subfigure}
\vspace{-1em}
\caption{Comparing the impact of the bit-vector index under different granularity levels. Uses top-128 and the top-64
most frequent bigrams in workload queries.}
\label{fig:granu}
\end{figure}


Now we answer \ref{q4} by analysing the impact of index granularity on the index construction overhead and performance gain. In the experimental setup, we use index granularity of 8, 64, 192, 256, and 512 as number of log lines in a group. 
We maintain uniformity with the same set of bigrams for result reporting. Specifically, we plot results for the top-128 and top-64 bigrams derived from the queries of the \wlsqlfull{} workload, and run the full \wlsqlfull{} workload to compare the results under different levels of granularity. 

\subsubsection{Index Construction Overhead.}

\Cref{fig:granu_64_index_size} and \Cref{fig:granu_128_index_size} show as granularity levels become coarser, there is a drastic decrease in index size. In~\Cref{fig:granu_64_index_size} and \Cref{fig:granu_128_index_size}, we observe that construction time decreases as the number of log lines within each group increases. When log lines are amalgamated into larger groups, for bit-vector index, it has smaller number of bit-vectors. Additionally, fewer bigram existence checks are made on average for each log line as the number of log lines corresponding to one index entry increases. This correlation between granularity and index overhead suggests opportunity for a coarser granularity to achieve lower space overhead and higher performance improvement.

\subsubsection{Matching Time with Index.}
Looking at~\Cref{fig:granu_64_match}, when we index with the top-64 bigrams, we can see that as the number of log lines represented by one index entry increases, the total runtime also increases. However, for the case with top-128 bigrams in ~\Cref{fig:granu_128_match}, having a granularity of 8 actually achieves better performance than a granularity of 1. Although having a finer granularity of 1 records more information about the dataset compared to a granularity level of 8, the gain in indexing the top 64 to 128 bigrams cannot compensate for the increased time in index lookup. For the \wlsqlfull{} workload, the fastest running time is achieved with a granularity level of 8 and by indexing 128 bigrams.\eat{ \sd{192 or 128? 9f is for 128.}.}  Selecting index granularity involves balancing index overhead and matching time.
\eat{Finer granularity increases overhead, with more details stored and processed, but typically reduces matching time due to more selective filtering.}
Finer granularity, which results in a larger number of index entries due to the more detailed division of the dataset, increases overhead but typically reduces matching time due to more selective filtering.
Conversely, coarser granularity decreases overhead by aggregating information but may result in longer matching time due to a less selective filtering process, resulting in more log lines being considered during matching. The choice of granularity thus depends on optimizing computational efficiency within the constraints and requirements of the specific application.

\subsection{Other Indexing Methods}\label{subsec:other_idx}
\Ling{Other popular solutions to encode the existence of $n$-grams are signature files and inverted index.}
\eat{Another popular solution is to use the inverted index to record the existence of $n$-grams. }
In this section, we answer \ref{q5} with a comparative analysis of our bit-vector index\Ling{, signature files,} and the inverted index, \eat{both}\Ling{all} constructed using an identical set of bigrams. \Ling{We use the hashing schema, MurmurHash~\cite{murmurhash} for bigram signatures in BitFunnel~\cite{Goodwin_2017}. We used line level bit array sizes of 64 and 128 for both signature files and bivector index. For each size $k$, we build the bit-vector index with top $k$ bigrams, and signature files with top-$k$, upper half of the top-$k$, and lower half of the top-$k$ bigrams. We will explain the setting during evaluation.} 

Since an inverted index is often used to index documents of a size larger than the average size of a log line, we compare the impact of the indexing strategies on varying levels of index granularity. 
We group a fixed number of log lines together and assign a unique group id for indexing purposes. For example, at a granularity of 64, consecutive log lines were amalgamated into groups of 64, resulting in indices that are constructed based on whether specific bigrams are present in any of the log lines within each group. 
We build the inverted index using a hash map, with the bigrams being the key, and a set of log group IDs\eat{  \sd{vector or set? For inverted indexes, you would be doing an intersection which is more efficient with IDs being in a set?}} where the corresponding bigram exists being the value. {Each bit in the bit-vector index now represent the existence of the bigrams in the corresponding log group.}
We run the experiments using the same configurations of index granularity and number of bigrams as in~\Cref{subsec:idx_granu}.

\begin{figure}[!tp]
\begin{subfigure}{0.45\columnwidth}
  \centering
  \includegraphics[width=0.95\columnwidth]{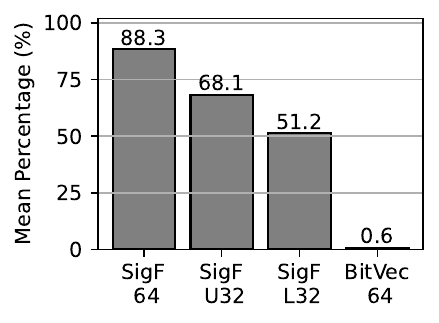}  
  \caption{Percentage After Filter}
  \label{fig:sig_filter}
\end{subfigure}
\hfill
\begin{subfigure}{0.45\columnwidth}
  \centering
  \includegraphics[width=0.95\columnwidth]{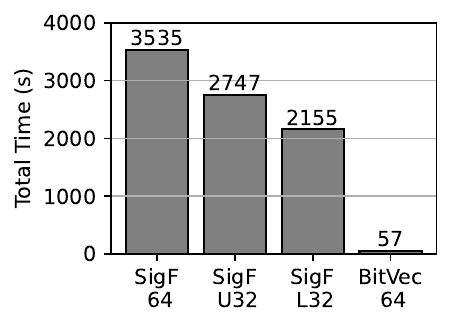}  
  \caption{Matching Time}
  \label{fig:sig_time}
\end{subfigure}
\caption{Comparing the performance of \sys\ index and signature files with the same number of bits per log line in the dataset. Uses top-64 most frequent bigrams in workload queries to construct the \sys\ (BitVec 64). Uses top-64 (SigF 64), top-32 (SigF U32), and 32th-to-64th (SigF L32) most frequent bigrams for signature files.\eat{ \sd{avoid color in the bars if it doesn't mean anything.}}}
\label{fig:sigfiles}
\end{figure}

\begin{figure*}[!htb]
    \centering
    \begin{minipage}{.62\textwidth}
        \begin{subfigure}{0.48\columnwidth}
          \centering
          \includegraphics[width=.95\columnwidth]{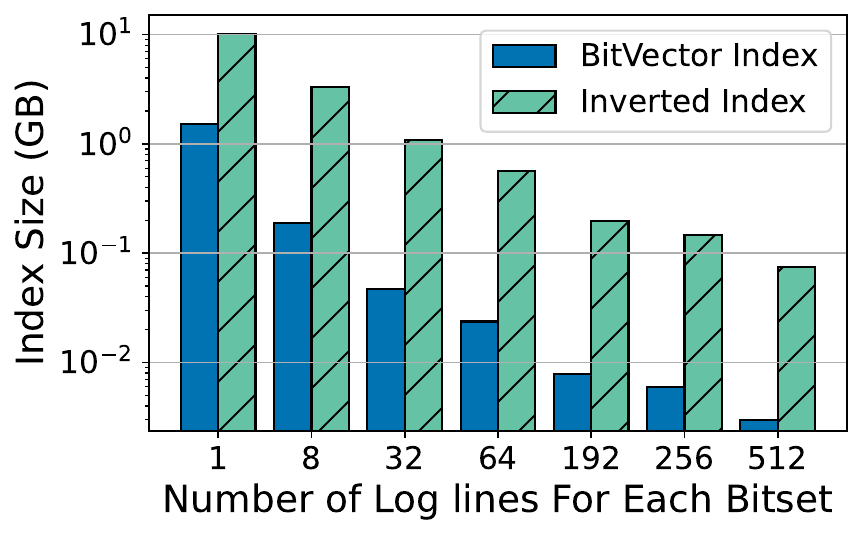}  
          \caption{Index Size}
          \label{fig:inverted_size}
        \end{subfigure}
        \begin{subfigure}{0.48\columnwidth}
          \centering
          \includegraphics[width=.95\columnwidth]{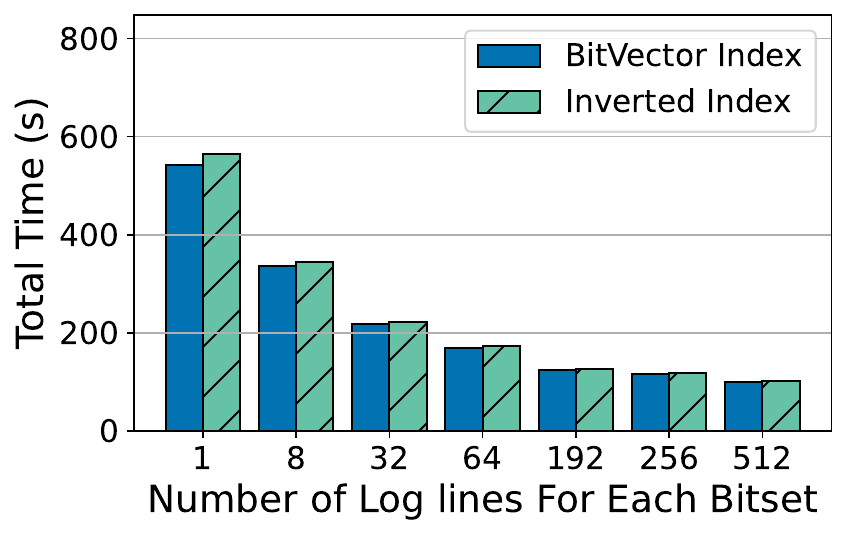}  
          \caption{Index Construction Time}
          \label{fig:inverted_build}
        \end{subfigure}
        \vspace{-.5em}
        \caption{Comparing the time and space overhead of index construction of \sys{} index and inverted index. Uses top-128
        most frequent bigrams in workload queries for both indices.}
        \label{fig:inverted}
    \end{minipage}%
    \hfill
    \begin{minipage}{0.35\textwidth}
         \centering
        \includegraphics[width=0.8\columnwidth]{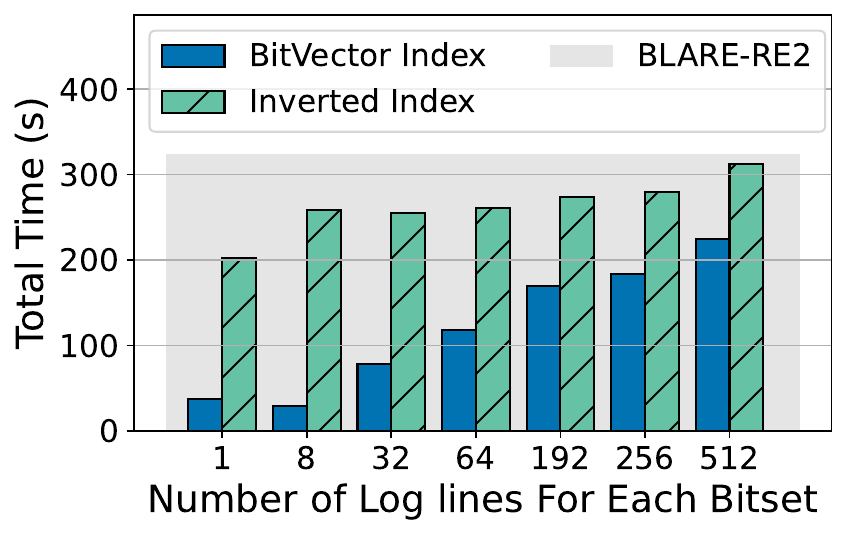}  
        \caption{ Comparing matching time of the \sys{} index and inverted index. Uses top-128
        most frequent $n$-grams in workload queries for both indices.}
        \label{fig:inverted_match}
    \end{minipage}
\end{figure*}

\subsubsection{Index Construction Overhead.}
\eat{
\begin{figure}[!tp]
\begin{subfigure}{0.5\columnwidth}
  \centering
  \includegraphics[width=.95\columnwidth]{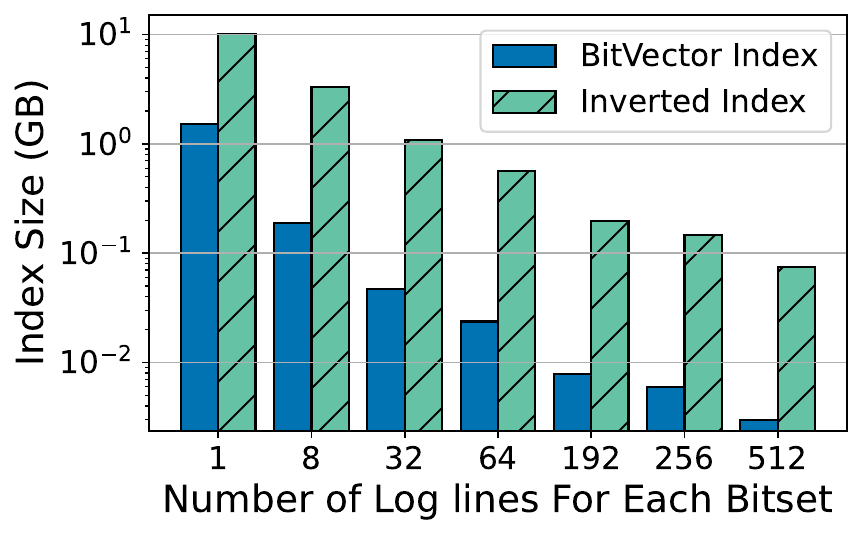}  
  \caption{Index Size}
  \label{fig:inverted_size}
\end{subfigure}
\hspace{-1em}
\begin{subfigure}{0.5\columnwidth}
  \centering
  \includegraphics[width=.95\columnwidth]{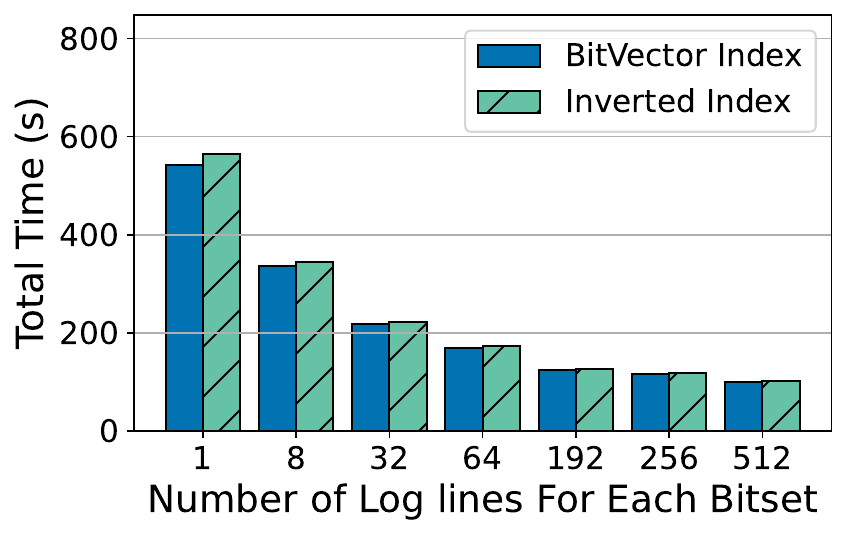}  
  \caption{Index Construction Time}
  \label{fig:inverted_build}
\end{subfigure}
\caption{Comparing the time and space overhead of the bit-vector index and inverted index. Uses top-128
most frequent $n$-grams in workload queries for both indices.}
\label{fig:inverted}
\end{figure}
}
\Ling{Since the index structures and construction steps of signature files and \sys\ index are very similar, where the only difference lies in how to encode the bigrams for each line, their index construction overheads are almost identical. However, this is not the case for inverted index.}
\Cref{fig:inverted_size} shows the detailed index size in log scale of both indices for different levels of granularity.
In our analysis, we consider only the minimum required size of the inverted index, counting only the space used to hold actual index data. This is important to note, as the space reserved for a common hash map data structure is usually much larger.
Upon observing the bar plot, we can see a notable difference in the trends of sizes between bit-vector index and the inverted index. 
For both index types, the index size decreases as granularity becomes coarser. The inverted index demands less space to denote the presence of specific bigram, and the bit-vector index has smaller number of index entries when the size of each entry stays constant.
Under identical conditions, the bit-vector index has a significantly smaller size compared to the inverted index. This difference in size highlights the space efficiency of the bit-vector index in maintaining a compact representation while managing the same set of bigrams.
\eat{In addition, as the granularity level becomes coarser, a proportional decrease in index size is noted for both bit-vector and inverted indexes. }
In~\Cref{fig:inverted_build}, we present the comparison of index construction time between the two types of index under varying granularity levels. Notably, the construction time of the two indices aligns closely under comparable conditions, as both indices are constructed by scanning through the dataset once, extracting all unique bigrams from each log line, looking up the existence of the same number of bigrams in hash-map type data structure, and then performing an update on the index.

\subsubsection{Matching Time with Index.}
\eat{
\begin{figure}[!tp]
 \centering
\includegraphics[width=0.6\columnwidth]{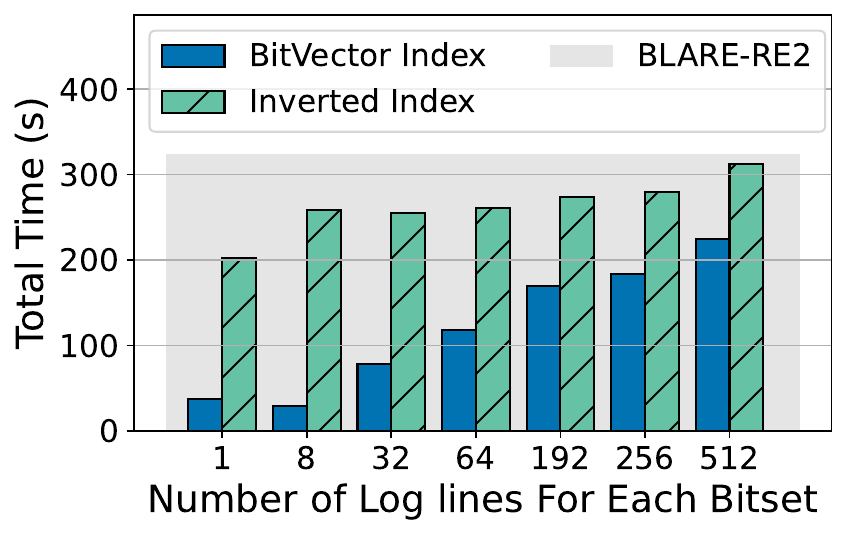}  
\caption{ Comparing matching time of the bit-vector index and inverted index. Uses top-128
most frequent $n$-grams in workload queries for both indices.}
\label{fig:inverted_match}
\end{figure}
}

Now we compared the workload runtime of the \sys, signature files, and inverted indices. 

\Ling{Looking at signature files first, we run an experiment with bit-vector size of 64 to index the top-64 bigrams for both signature files and \sys\ index. Signature for each log line filtered out only 11.7\% of the dataset, left with 88.3\% log lines going through the regex engine, whereas \sys\ index filtered out 99.4\% of the dataset, as shown in the leftmost and rightmost bars in~\Cref{fig:sig_filter}. Looking at the bit-vector signatures for each log line and the bitmask for each regex, we note that they are over saturated; bit-vectors are often all ones for log lines and all zeros for regexes. }

\Ling{With the same bit-vector size 64, we reduce the number of bigrams indexed to 32. By selecting the top-32 frequent bigrams, the percentage of log lines filtered increase to 31.9\%; by selecting the other half of the top-64 frequent bigrams, only about half (51.2\%) of the dataset need to be matched with regex engine after the index filtering. Both percentages of signature files method are still significantly higher\eat{ \sd{lower or higher? BitVec has 0.6 so everything else is higher, right?}} than \sys\ index, as shown in~\Cref{fig:sig_filter}. 
}
\eat{Looking at~\Cref{fig:sig_time}, the filtering difference translates to more significant query overhead of signature files compared to \sys\ index, resulting in a 12.5$\times$, 11.2$\times$, and 10.1$\times$ slower overall matching time for signature files using 64 bigrams, top-32 bigrams, and 32-64th bigrams respectively compared to \sys\ index of the same size.}
\Ling{Looking at~\Cref{fig:sig_time}, the filtering difference translates to more higher querying overhead for signature files compared to \sys\ index, resulting in a 62$\times$, 48$\times$, and 38$\times$ slower overall matching time for signature files using 64 bigrams, top-32 bigrams, and 32-64th bigrams respectively compared to \sys\ index of the same size. In fact, the matching time for signature files with 64 bigrams is very close to the total workload time of using RE2 directly without any index as the signature of each log line is usually too dense.}
\eat{With the knowledge of all bigrams indexed, we can represent its existence with a minimum of one bit as in the \sys\ index. False positives comes only for log lines that has all the required bigrams indexed but does not produce an exact match. Meanwhile for signature, using 64 bits to represent each bigram and OR together multiple 64 bits together to an overall 64-bit signature, we are guaranteed to lose information along the way, and the false positive of signature files indices additionally comes from failing to identify log lines that does not have all the required bigrams by the index.}

\Ling{Shifting our focus to inverted index, }in~\Cref{fig:inverted_match} 
there is a general trend that the \sys\ index consistently outperforms the inverted index in terms of matching time, especially at finer granularities. This performance disparity becomes less pronounced as granularity become coarser.
Also, as we transition to coarser granularities, the inverted index demonstrates a more gradual increase in the runtime compared to the \sys\ index, which experiences a sharper increase. The cost of index look-up to find all candidate log lines is higher for the inverted index than for the bit-vector index. When the granularity gets coarser, the number of lines remaining for full regex matching increases, and the cost of candidate search is amortized. Therefore, we can see from~\Cref{fig:inverted_match} that as the granularity becomes coarser, the performance gap between the bit-vector index and the inverted index becomes smaller. This aligns with conclusions of prior works that inverted index is suitable for indexing documents much larger in sizes than the log lines and the log groups in our experiments. 
\eat{Similarly, when we index with only 4 or 8 bigrams, the performance of the bit-vector index and inverted index are almost identical for most of the 5 granularity levels.}

\subsection{Case Study: Impact of Query Set Size} \label{subsec:expr:db_full}
{
To analyze the scalability of \sys{} as the query workload size increases, we conducted an additional experiment comparing the sampled \wlsqlfull{} workload (132 queries) with the complete \wlsqlfull{} workload (8,941 queries). As the query set expanded, we did not see a significant change in index construction time. This is because the query set size impacts only the bigram selection process, which is very lightweight: around 0.05 seconds in the entire $\sim 500$ seconds of index construction time. For query performance, the sampled workload runs 35.7s  using \sys{} indexing with 64 bigrams, achieving a runtime speedup of more than $9\times$, whereas the larger workload runs 49.9s with the same configuration, achieving a significant speedup of $379\times$ over the baseline time of 18928s.
These results confirm that \sys{} efficiently scales to larger query sets.
}

\section{Related Work} \label{sec:related}

Due to the complexity of regex patterns and the variety of datasets, efficient regular expression matching and indexing are vital in several areas such as bioinformatics~\cite{Hammel_2002, Arslan_2006, Mulder_2006, Gouret_2009, Prieto_2014}, event stream processing~\cite{Agrawal_2008, Cohen_2008, Halle_2014, SUGIURA_2020},
network intrusion detection system~\cite{Baker_2006, Roan_2006, Kumar_2012, Jamshed_2012, Vasiliadis_2011,Tripp_2007, Hong_Jip_Jung_2006, Ming_Gao_2006, Meiners_2010, Liu_2004, Lee_2015}, text processing~\cite{Tabuchi_2003,Bui_2014} and data mining~\cite{Garofalakis_1999, Garofalakis_2002, de_Amo_2007, Trasarti_2008}. 

\introparagraph{Traditional regex matching algorithms}
Without the aid of indexing, 
traditional pattern-matching approaches, such as
\eat{deterministic finite automata (DFA) and nondeterministic finite automata (NFA)}
{DFA and NFA}, have established the foundations. Adapting ideas from KMP~\cite{KMP}, Aho-Corasick algorithms, and Boyer-Moore~\cite{Boyer_1977}, there has been work optimizing automata matching~\cite{Horspool_1980, Sunday_1990, Crochemore_1994, Trivedi_2020}. DFA, despite its reputation for rapid matching, frequently encounters state explosion issues, {
and works had been done optimizing its space usage~\cite{Wang_2014, Yu_2006, Yang_Prasanna_2011, Kumar_2006, Becchi_2007, Becchi_2007_2, Ficara_2011, Huang_2013, Brodie_2006, Becchi_2013, Kong_2008, Huo_2015, Hua_2009, Yang_2016, Najam_2014, Qiu_2022, Kumar_2012}.
NFA provides a more space-efficient option, at the cost of computational blow-up during matching, and there are works on optimizing its performance~\cite{Yang_2011, Myers_1999, Navarro_2000, hyyro_2002, Roan_2006, Kaneta_2010}.
Automaton variations that employed one or more of the above techniques have been proposed, including suffix automaton~\cite{Crochemore_1999, Navarro_2000}, XFA~\cite{Smith_2008}, D$^2$FA~\cite{Kumar_2006}, and $\delta^N$FA~\cite{Ficara_2011}.
Recent work has also investigated how to obtain speed up those methods using modern hardwares~\cite{Lee_2015, Tripp_2007, Baker_2004, Hong_Jip_Jung_2006, Brodie_2006, Yang_2008, Sidhu_2001, Jamshed_2012, Vasiliadis_2011, Huang_2013, Peng_2011, Meiners_2010, Ming_Gao_2006, Wang_2023, wang2019hyperscan, Scarpazza_2009}
Those solutions targeted general regex matching task without considering the characteristics of a specific workload. }
\eat{
\introparagraph{Indexing the regexes}
{In tasks like intrusion detection, where each pattern is matched only once by a large number of rules that can go as high as 10k, it is useful to build indexes on the patterns than on the data. }
There are works touching upon multi-pattern matching, which is common in intrusion detection tasks. {In such scenario, each incoming packet is matched against a rule set of size more than 10K in its arrival. Most of the packets can be filtered by string matching in the rules~\cite{pigasus}. The query set of the workload is also unchanged or infrequently changed. Due to those characteristics, existing works tend to build index-like structures on the rules rather than on the incoming data.} One approach to tackling this problem is to selectively group regexes together to form a composite DFA~\cite{Yu_2006}. Another popular method involves building tree-like data structures for sub-strings, regular expressions or DFAs~\cite{Chan_2008, Chan_2002, Lee_2015, Chan_2003}. Intuitively, some other works also look into preprocessing and indexing the dataset to achieve higher query efficiency. 
}

\introparagraph{Theoretical analysis of index building}
When we can preprocess the dataset, building an index for faster regex queries becomes an option.
\cite{Gibney_2021} conducted a theoretical analysis to understand hardness constraints and trade-offs of the index built for regex matching. Different data preparation methods and indexing strategies are proposed considering the trade-offs.
\eat{For example, ~\cite{Gibney_2021} presented a method based on Boolean matrix multiplication that can provide fast regex matching while taking exponential space in the dataset length. }
Some works focus on a specific type of automata (e.g. Wheeler automata), analyzing the performance bounds~\cite{Cotumaccio_2021}
\eat{, and building FM-index using the idea of co-lexicographical ordering and prefix-sorting techniques~\cite{Cotumaccio_2021, cotumaccio2023colexicographically, Alanko_2020}}. 

\introparagraph{Bit-Vector Indexes} {In the related realm of information retrieval, research has been done on signature file, generating bit signatures for documents~\cite{Faloutsos_1984}, and multikey attributes~\cite{Sacks_Davis_1987} for fast lookup. The signatures are generated by hashing each word into bit-vector of the same size, and then superimposing the bit-vector~\cite{Faloutsos_1984}. Though with various attempts of optimization~\cite{Sacks_Davis_1987, Ishikawa_1993, Lin_1992, Kent_1990}, evaluation~\cite{Zobel_1998} demonstrated that there is a trade-off wrt. inverted indexes, as it requires multiple hash computation for each word and a signature size as large as $10K$ bits. The difference in bit-vector generation step makes \sys{} inherently free of the above disadvantages.
}
{Recent work on Multi-Dimensional Data Layouts (MDDL)~\cite{Ding_2024} has also explored bitmap-based indexing for prefiltering in database systems. The key idea in MDDL is to maintain a bitmap per row that encodes whether the row satisfies a fixed set of query predicates. While MDDL orders predicates based on estimated selectivity to minimize overall filtering cost, our approach orders $n$-grams by their frequency in the query workload to ensure broader applicability across diverse regex patterns. Additionally, MDDL supports dynamic updates and maintenance of bitmap indexes as workloads evolve, which our current design does not support—an exciting direction for future work. However, unlike MDDL, which benefits from sorted bitmap indexes, applying similar sorting strategies to log data for regex evaluation is significantly more expensive due to the sheer data volume. Moreover, such sorting could degrade regex evaluation performance, especially when logs are already sorted by another dimension (e.g., time). Consequently, data-partitioning-based techniques require deeper investigation, presenting an exciting future work direction.}{Feed-forward Bloom filter is also used in \cite{Moraru_2012, Cha_2011} as a lightweight indexing structure on the patterns for multi-pattern matching together with hardware optimization techniques.}

\introparagraph{Inverted indexes}
A very large collection of works uses inverted index or similar data structures to index multigrams of the dataset. {Inverted index with positional information works well for database when each entry is large.} The inverted index uses $n$-gram as the search key and returns a list of positional identifiers~\cite{Williams_2002, Navarro_2018, FREE, Kim_2010, tinylex, SHIN_2018}. 
\eat{Matching on a regex query involves taking all the $n$-gram entries that appear in the query, and then doing set operations on the values. }
There are works that aim to reduce the space consumption~\cite{tinylex} or build the index under space constraint~\cite{BEST}.

\introparagraph{Tree indexes}
General tries and suffix trees are also commonly used for indexing data for pattern matching~\cite{Baeza_Yates_1996, Burkhardt_1999, Giladi_2002, Meek_2003,Cao_2005, Kandhan_2010}. {However, they have a heavy space requirement. Without compression, these types of indexes on a natural language dataset can easily reach a size more than 10 times of the original dataset~\cite{Ruano_2021}.} A two-level index is also employed using a hash table and a trie on encoded signature of DNA data~\cite{Cao_2005}. {This method is popular in genome database where the alphabet size is small. Suffix trees are also commonly used for multi-pattern matching where indexes are constructed on the queries.} 
\eat{SigMatch~\cite{Kandhan_2010} uses fixed height trie together with bloom filters to achieve high performance with low space requirements.}

\introparagraph{$N$-grams for indexing and filtering}
Besides regex matching, the $n$-gram indexing or pre-filtering is broadly used in string matching for predicate matching and joins~\cite{qgram_dbms}.
\eat{Solutions usually start with $n$-grams candidate extracted by sliding windows, but some work with non-overlapping $n$-grams for approximate match results~\cite{qsample}.}
Trigrams are commonly used in existing indexing solutions~\cite{cox2012regular, trigram-search, postgres_trg} {where all trigrams are indexed, usually using inverted index.} Works has been done comparing the impact of $n$ of $n$-gram for inverted index~\cite{tinylex, Kim_2010, cox2012regular}. Some~\cite{Jakubicek_2014, tinylex} build an index for all $n$-grams of $n \leq k$. Solutions in~\cite{Navarro_2009, FREE, Qiu_2022, BEST} instead use \eat{unfixed }variable length multigrams. Github Code Search~\cite{github} assigns weights to each bigram, and selectively indexes on multigrams that are composed of consecutive bigrams {according to their weights}
\eat{with weights of inner bigrams strictly below weights of bigrams at the two ends}. \eat{Besides building inverted on $n$-grams of some fixed $n$, ~\cite{Qiu_2022, BEST} also generates a tree-based query plan considering the selectivity of each $n$-grams.}

\eat{
Considering the space constraint for the index, in many cases it is unrealistic to store all possible grams. There is a rich literature on the strategy of selecting the optimal subset for indexing {when facing with a fixed set of queries and fixed dataset.} The search for the optimal subset is shown to be NP-hard~\cite{BEST}; most works focus on techniques to reduce the search space and reach a good enough subset.
\textsf{BEST}~\cite{BEST} uses bit-matrices and suffix array preprocessing to approximate the best covering subset of multigrams for indexing under space constraints. \textsf{FREE}~\cite{FREE} also defined the usefulness of $n$-grams and selected the minimal prefix-free set of multigrams surpassing the usefulness threshold previously. \textsf{FAST}~\cite{FAST} combined strategies from \textsf{BEST} and \textsf{FREE} and reshaped the selection problem to a combinatorial optimization framework. Considering the characteristics of the real-world dataset, SigMatch~\cite{Kandhan_2010} takes a greedy approach and adds grams with the highest frequency in incoming packets.
Though indexing all $n$-grams,~\cite{Kim_2010} uses estimate I/O cost, selectivity estimation, and intersection size estimation techniques to choose an optimal set of $n$-grams for each query. 
\TODO{Some works are for fixed-pattern matching I would say. If need space we should remove them or mention and condense them.}
}

\eat{\introparagraph{Data exploration in log analysis frameworks}
Beyond methods and data structure that accelerate regular expression matching in log analysis, some log analysis frameworks move a step further and use clustering methods to automatically detect patterns~\cite{Debnath_2018,Tang_2011} and events~\cite{Makanju_2009, Vaarandi, Fu_2009} within log data.
\eat{There are other works that bypass the structured information extraction phase and more specifically target one task like anomaly/failure detection~\cite{He_2016, Bodik_2010, Chen_2004} using statistical or ML techniques.
General data analysis frameworks are also used to extract insightful information and expose hidden patterns for the exploration phase of log analysis tasks. }
Text search engines are used for general queries among log data~\cite{lucene, Wei_2017}. Data exploration frameworks like Lux~\cite{Lee_2021} and DataPrep.EDA~\cite{Peng_2021} can automatically generate statistical modeling and visualization recommendations, providing insights to data scientists for deeper analysis. }

\section{Conclusion and Future Work} \label{sec:conclusion_future}
We have revisited the challenging problem of regular expression indexing in this paper. We conducted a systematic analysis of traditional methods like inverted indexes with $n$-grams. We examined how well they performed when combined with modern hardware and state-of-the-art regex evaluation engines.
We introduced a lightweight indexing framework, \sys{}, that is configurable. This allows for a balance between the cost of building indexes and the improvement in query performance. Our technique stands out from existing approaches in two key aspects. First, we employ a bit filtering-based index for a carefully chosen subset of $n$-grams. Second, we store the bit index directly alongside each log line in our input, rather than in a separate data structure. This integration not only simplifies the indexing process with a much lower cost than existing methods but also results in a more streamlined and efficient framework. {It is also important to emphasize that \sys{} remains independent of the regex engine, which guarantees flexibility and wide applicability on different platforms and use cases. }

We envision our lightweight indexing framework developing further in the future to leverage distributed computing and parallel processing to manage ever-increasing volumes of data with ease. To further increase the adaptability and robustness of our framework, another area worth noting is the development of a dynamic index updating strategy. The strategy is expected to include the ability to add, remove, or update bigrams in the index. {Future extensions could leverage incremental statistics from query execution to adaptively refine the index, enabling dynamic reconfiguration for settings where the data and queries are continuously changing. This setting is an important consideration to remedy the limitation of our work's assumption that the distribution of literals is close to that of the English language. The dynamic setting is also common in streaming and security-focused scenarios such as intrusion detection systems (for example, Snort~\cite{roesch1999snort} and Suricata~\cite{suricata2009home}).}
{We also plan to explore bitmap compression techniques like RLE, Roaring Bitmaps, and WAH, which can reduce storage overhead without affecting precision. These methods are especially promising for integration with columnar storage formats, further improving space efficiency.} Furthermore, a thorough investigation of storage strategies on efficient data formats like Parquet can enhance index storage and retrieval procedures. 
\eat{{We also plan to develop heuristics for auto-configuring parameters such as $k$ (number of bigrams) and index granularity based on user-defined size constraints, using techniques involving bigram selectivity analysis. }}

\begin{acks}
This research was supported in part by a grant from the Microsoft Jim Gray Systems Lab (GSL) and by the National Science Foundation (NSF) under grant CCF-2407690.
\end{acks}
\clearpage

\bibliographystyle{ACM-Reference-Format}
\bibliography{references}


\begin{thebibliography}{113}


\ifx \showCODEN    \undefined \def \showCODEN     #1{\unskip}     \fi
\ifx \showISBNx    \undefined \def \showISBNx     #1{\unskip}     \fi
\ifx \showISBNxiii \undefined \def \showISBNxiii  #1{\unskip}     \fi
\ifx \showISSN     \undefined \def \showISSN      #1{\unskip}     \fi
\ifx \showLCCN     \undefined \def \showLCCN      #1{\unskip}     \fi
\ifx \shownote     \undefined \def \shownote      #1{#1}          \fi
\ifx \showarticletitle \undefined \def \showarticletitle #1{#1}   \fi
\ifx \showURL      \undefined \def \showURL       {\relax}        \fi
\providecommand\bibfield[2]{#2}
\providecommand\bibinfo[2]{#2}
\providecommand\natexlab[1]{#1}
\providecommand\showeprint[2][]{arXiv:#2}

\bibitem[{aappleby}({[n.\,d.]})]%
        {murmurhash}
\bibfield{author}{\bibinfo{person}{{aappleby}}.} \bibinfo{year}{[n.\,d.]}\natexlab{}.
\newblock \bibinfo{booktitle}{\emph{SMHasher}}.
\newblock
\urldef\tempurl%
\url{https://github.com/aappleby/smhasher}
\showURL{%
\tempurl}


\bibitem[Adams and Meltzer(1993)]%
        {trigram-search}
\bibfield{author}{\bibinfo{person}{Elizabeth~S. Adams} {and} \bibinfo{person}{Arnold~C. Meltzer}.} \bibinfo{year}{1993}\natexlab{}.
\newblock \showarticletitle{Trigrams as Index Element in Full Text Retrieval: Observations and Experimental Results}. In \bibinfo{booktitle}{\emph{Proceedings of the 1993 ACM Conference on Computer Science}} (Indianapolis, Indiana, USA) \emph{(\bibinfo{series}{CSC '93})}. \bibinfo{publisher}{Association for Computing Machinery}, \bibinfo{address}{New York, NY, USA}, \bibinfo{pages}{433–439}.
\newblock
\showISBNx{0897915585}
\href{https://doi.org/10.1145/170791.170891}{doi:\nolinkurl{10.1145/170791.170891}}


\bibitem[Agrawal et~al\mbox{.}(2008)]%
        {Agrawal_2008}
\bibfield{author}{\bibinfo{person}{Jagrati Agrawal}, \bibinfo{person}{Yanlei Diao}, \bibinfo{person}{Daniel Gyllstrom}, {and} \bibinfo{person}{Neil Immerman}.} \bibinfo{year}{2008}\natexlab{}.
\newblock \showarticletitle{Efficient pattern matching over event streams}. In \bibinfo{booktitle}{\emph{Proceedings of the 2008 {ACM} {SIGMOD} international conference on Management of data}}. \bibinfo{publisher}{{ACM}}.
\newblock
\href{https://doi.org/10.1145/1376616.1376634}{doi:\nolinkurl{10.1145/1376616.1376634}}


\bibitem[Arslan and He(2006)]%
        {Arslan_2006}
\bibfield{author}{\bibinfo{person}{Abdullah~N. Arslan} {and} \bibinfo{person}{Dan He}.} \bibinfo{year}{2006}\natexlab{}.
\newblock \showarticletitle{An improved algorithm for the regular expression constrained multiple sequence alignment problem}. In \bibinfo{booktitle}{\emph{Sixth {IEEE} Symposium on {BioInformatics} and {BioEngineering} ({BIBE}{\textquotesingle}06)}}. \bibinfo{publisher}{{IEEE}}.
\newblock
\href{https://doi.org/10.1109/bibe.2006.253324}{doi:\nolinkurl{10.1109/bibe.2006.253324}}


\bibitem[Baeza-Yates and Gonnet(1996)]%
        {Baeza_Yates_1996}
\bibfield{author}{\bibinfo{person}{Ricardo~A. Baeza-Yates} {and} \bibinfo{person}{Gaston~H. Gonnet}.} \bibinfo{year}{1996}\natexlab{}.
\newblock \showarticletitle{Fast text searching for regular expressions or automaton searching on tries}.
\newblock \bibinfo{journal}{\emph{J. ACM}} \bibinfo{volume}{43}, \bibinfo{number}{6} (\bibinfo{date}{Nov.} \bibinfo{year}{1996}), \bibinfo{pages}{915–936}.
\newblock
\showISSN{1557-735X}
\href{https://doi.org/10.1145/235809.235810}{doi:\nolinkurl{10.1145/235809.235810}}


\bibitem[Baker et~al\mbox{.}(2006)]%
        {Baker_2006}
\bibfield{author}{\bibinfo{person}{Zachary Baker}, \bibinfo{person}{Hong jip Jung}, {and} \bibinfo{person}{Viktor Prasanna}.} \bibinfo{year}{2006}\natexlab{}.
\newblock \showarticletitle{Regular Expression Software Deceleration for Intrusion Detection Systems}. In \bibinfo{booktitle}{\emph{2006 International Conference on Field Programmable Logic and Applications}}. \bibinfo{publisher}{{IEEE}}.
\newblock
\href{https://doi.org/10.1109/fpl.2006.311246}{doi:\nolinkurl{10.1109/fpl.2006.311246}}


\bibitem[Baker and Prasanna(2004)]%
        {Baker_2004}
\bibfield{author}{\bibinfo{person}{Zachary~K. Baker} {and} \bibinfo{person}{Viktor~K. Prasanna}.} \bibinfo{year}{2004}\natexlab{}.
\newblock \showarticletitle{Time and area efficient pattern matching on {FPGAs}}. In \bibinfo{booktitle}{\emph{Proceedings of the 2004 {ACM}/{SIGDA} 12th international symposium on Field programmable gate arrays}}. \bibinfo{publisher}{{ACM}}.
\newblock
\href{https://doi.org/10.1145/968280.968312}{doi:\nolinkurl{10.1145/968280.968312}}


\bibitem[Becchi and Cadambi(2007)]%
        {Becchi_2007_2}
\bibfield{author}{\bibinfo{person}{M. Becchi} {and} \bibinfo{person}{S. Cadambi}.} \bibinfo{year}{2007}\natexlab{}.
\newblock \showarticletitle{Memory-Efficient Regular Expression Search Using State Merging}. In \bibinfo{booktitle}{\emph{{IEEE} {INFOCOM} 2007 - 26th {IEEE} International Conference on Computer Communications}}. \bibinfo{publisher}{{IEEE}}.
\newblock
\href{https://doi.org/10.1109/infcom.2007.128}{doi:\nolinkurl{10.1109/infcom.2007.128}}


\bibitem[Becchi and Crowley(2007)]%
        {Becchi_2007}
\bibfield{author}{\bibinfo{person}{Michela Becchi} {and} \bibinfo{person}{Patrick Crowley}.} \bibinfo{year}{2007}\natexlab{}.
\newblock \showarticletitle{An improved algorithm to accelerate regular expression evaluation}. In \bibinfo{booktitle}{\emph{Proceedings of the 3rd {ACM}/{IEEE} Symposium on Architecture for networking and communications systems}}. \bibinfo{publisher}{{ACM}}.
\newblock
\href{https://doi.org/10.1145/1323548.1323573}{doi:\nolinkurl{10.1145/1323548.1323573}}


\bibitem[Becchi and Crowley(2013)]%
        {Becchi_2013}
\bibfield{author}{\bibinfo{person}{Michela Becchi} {and} \bibinfo{person}{Patrick Crowley}.} \bibinfo{year}{2013}\natexlab{}.
\newblock \showarticletitle{A-{DFA}: A Time- and Space-Efficient DFA Compression Algorithm for Fast Regular Expression Evaluation}.
\newblock \bibinfo{journal}{\emph{{ACM} Transactions on Architecture and Code Optimization}} \bibinfo{volume}{10}, \bibinfo{number}{1} (\bibinfo{date}{apr} \bibinfo{year}{2013}), \bibinfo{pages}{1--26}.
\newblock
\href{https://doi.org/10.1145/2445572.2445576}{doi:\nolinkurl{10.1145/2445572.2445576}}


\bibitem[Boyer and Moore(1977)]%
        {Boyer_1977}
\bibfield{author}{\bibinfo{person}{Robert~S. Boyer} {and} \bibinfo{person}{J.~Strother Moore}.} \bibinfo{year}{1977}\natexlab{}.
\newblock \showarticletitle{A Fast String Searching Algorithm}.
\newblock \bibinfo{journal}{\emph{Commun. ACM}} \bibinfo{volume}{20}, \bibinfo{number}{10} (\bibinfo{date}{oct} \bibinfo{year}{1977}), \bibinfo{pages}{762–772}.
\newblock
\showISSN{0001-0782}
\href{https://doi.org/10.1145/359842.359859}{doi:\nolinkurl{10.1145/359842.359859}}


\bibitem[Brodie et~al\mbox{.}(2006)]%
        {Brodie_2006}
\bibfield{author}{\bibinfo{person}{B.C. Brodie}, \bibinfo{person}{D.E. Taylor}, {and} \bibinfo{person}{R.K. Cytron}.} \bibinfo{year}{2006}\natexlab{}.
\newblock \showarticletitle{A Scalable Architecture For High-Throughput Regular-Expression Pattern Matching}. In \bibinfo{booktitle}{\emph{33rd International Symposium on Computer Architecture (ISCA'06)}}. \bibinfo{pages}{191--202}.
\newblock
\href{https://doi.org/10.1109/ISCA.2006.7}{doi:\nolinkurl{10.1109/ISCA.2006.7}}


\bibitem[Bui and Zeng-Treitler(2014)]%
        {Bui_2014}
\bibfield{author}{\bibinfo{person}{D.~D.~A. Bui} {and} \bibinfo{person}{Q. Zeng-Treitler}.} \bibinfo{year}{2014}\natexlab{}.
\newblock \showarticletitle{Learning regular expressions for clinical text classification}.
\newblock \bibinfo{journal}{\emph{Journal of the American Medical Informatics Association}} \bibinfo{volume}{21}, \bibinfo{number}{5} (\bibinfo{date}{sep} \bibinfo{year}{2014}), \bibinfo{pages}{850--857}.
\newblock
\href{https://doi.org/10.1136/amiajnl-2013-002411}{doi:\nolinkurl{10.1136/amiajnl-2013-002411}}


\bibitem[Burkhardt et~al\mbox{.}(1999)]%
        {Burkhardt_1999}
\bibfield{author}{\bibinfo{person}{Stefan Burkhardt}, \bibinfo{person}{Andreas Crauser}, \bibinfo{person}{Paolo Ferragina}, \bibinfo{person}{Hans-Peter Lenhof}, \bibinfo{person}{Eric Rivals}, {and} \bibinfo{person}{Martin Vingron}.} \bibinfo{year}{1999}\natexlab{}.
\newblock \showarticletitle{q-gram based database searching using a suffix array ({QUASAR})}. In \bibinfo{booktitle}{\emph{Proceedings of the third annual international conference on Computational molecular biology}}. \bibinfo{publisher}{{ACM}}.
\newblock
\href{https://doi.org/10.1145/299432.299460}{doi:\nolinkurl{10.1145/299432.299460}}


\bibitem[Cao et~al\mbox{.}(2005)]%
        {Cao_2005}
\bibfield{author}{\bibinfo{person}{Xia Cao}, \bibinfo{person}{Shuai~Cheng Li}, {and} \bibinfo{person}{Anthony K.~H. Tung}.} \bibinfo{year}{2005}\natexlab{}.
\newblock \showarticletitle{Indexing {DNA} Sequences Using q-Grams}.
\newblock In \bibinfo{booktitle}{\emph{Database Systems for Advanced Applications}}. \bibinfo{publisher}{Springer Berlin Heidelberg}, \bibinfo{pages}{4--16}.
\newblock
\href{https://doi.org/10.1007/11408079_4}{doi:\nolinkurl{10.1007/11408079_4}}


\bibitem[Cha et~al\mbox{.}(2011)]%
        {Cha_2011}
\bibfield{author}{\bibinfo{person}{Sang~Kil Cha}, \bibinfo{person}{Iulian Moraru}, \bibinfo{person}{Jiyong Jang}, \bibinfo{person}{John Truelove}, \bibinfo{person}{David Brumley}, {and} \bibinfo{person}{David~G. Andersen}.} \bibinfo{year}{2011}\natexlab{}.
\newblock \showarticletitle{{SplitScreen}: Enabling Efficient, Distributed Malware Detection}.
\newblock \bibinfo{journal}{\emph{Journal of Communications and Networks}} \bibinfo{volume}{13}, \bibinfo{number}{2} (\bibinfo{date}{apr} \bibinfo{year}{2011}), \bibinfo{pages}{187--200}.
\newblock
\href{https://doi.org/10.1109/jcn.2011.6157418}{doi:\nolinkurl{10.1109/jcn.2011.6157418}}


\bibitem[Chan and Ioannidis(1998)]%
        {Chan_1998}
\bibfield{author}{\bibinfo{person}{Chee-Yong Chan} {and} \bibinfo{person}{Yannis~E. Ioannidis}.} \bibinfo{year}{1998}\natexlab{}.
\newblock \showarticletitle{Bitmap index design and evaluation}.
\newblock \bibinfo{journal}{\emph{ACM SIGMOD Record}} \bibinfo{volume}{27}, \bibinfo{number}{2} (\bibinfo{date}{June} \bibinfo{year}{1998}), \bibinfo{pages}{355–366}.
\newblock
\showISSN{0163-5808}
\href{https://doi.org/10.1145/276305.276336}{doi:\nolinkurl{10.1145/276305.276336}}


\bibitem[Chan and Ioannidis(1999)]%
        {Chan_1999}
\bibfield{author}{\bibinfo{person}{Chee-Yong Chan} {and} \bibinfo{person}{Yannis~E. Ioannidis}.} \bibinfo{year}{1999}\natexlab{}.
\newblock \showarticletitle{An efficient bitmap encoding scheme for selection queries}.
\newblock \bibinfo{journal}{\emph{ACM SIGMOD Record}} \bibinfo{volume}{28}, \bibinfo{number}{2} (\bibinfo{date}{June} \bibinfo{year}{1999}), \bibinfo{pages}{215–226}.
\newblock
\showISSN{0163-5808}
\href{https://doi.org/10.1145/304181.304201}{doi:\nolinkurl{10.1145/304181.304201}}


\bibitem[Cho and Rajagopalan(2002)]%
        {FREE}
\bibfield{author}{\bibinfo{person}{Junghoo Cho} {and} \bibinfo{person}{S. Rajagopalan}.} \bibinfo{year}{2002}\natexlab{}.
\newblock \showarticletitle{A fast regular expression indexing engine}. In \bibinfo{booktitle}{\emph{Proceedings 18th International Conference on Data Engineering}}. \bibinfo{pages}{419--430}.
\newblock
\href{https://doi.org/10.1109/ICDE.2002.994755}{doi:\nolinkurl{10.1109/ICDE.2002.994755}}


\bibitem[Christopher Kings-Lynne and Korotkov({[n.\,d.]})]%
        {postgres_trg}
\bibfield{author}{\bibinfo{person}{Teodor~Sigaev Christopher Kings-Lynne, Oleg~Bartunov} {and} \bibinfo{person}{Alexander Korotkov}.} \bibinfo{year}{[n.\,d.]}\natexlab{}.
\newblock \bibinfo{booktitle}{\emph{F.35. pg\_trgm — support for similarity of text using trigram matching}}.
\newblock
\urldef\tempurl%
\url{https://www.postgresql.org/docs/current/pgtrgm.html}
\showURL{%
Retrieved October 12, 2023 from \tempurl}


\bibitem[Clem(2023)]%
        {github}
\bibfield{author}{\bibinfo{person}{Timothy Clem}.} \bibinfo{year}{2023}\natexlab{}.
\newblock \bibinfo{title}{The technology behind GitHub’s new code search}.
\newblock
\urldef\tempurl%
\url{https://github.blog/2023-02-06-the-technology-behind-githubs-new-code-search/}
\showURL{%
\tempurl}


\bibitem[Coetzee(2008)]%
        {tinylex}
\bibfield{author}{\bibinfo{person}{Derrick Coetzee}.} \bibinfo{year}{2008}\natexlab{}.
\newblock \showarticletitle{TinyLex: Static n-Gram Index Pruning with Perfect Recall}. In \bibinfo{booktitle}{\emph{Proceedings of the 17th ACM Conference on Information and Knowledge Management}} (Napa Valley, California, USA) \emph{(\bibinfo{series}{CIKM '08})}. \bibinfo{publisher}{Association for Computing Machinery}, \bibinfo{address}{New York, NY, USA}, \bibinfo{pages}{409–418}.
\newblock
\showISBNx{9781595939913}
\href{https://doi.org/10.1145/1458082.1458138}{doi:\nolinkurl{10.1145/1458082.1458138}}


\bibitem[Cohen and Kalleberg(2008)]%
        {Cohen_2008}
\bibfield{author}{\bibinfo{person}{Norman~H. Cohen} {and} \bibinfo{person}{Karl~Trygve Kalleberg}.} \bibinfo{year}{2008}\natexlab{}.
\newblock \showarticletitle{{EventScript}}. In \bibinfo{booktitle}{\emph{Proceedings of the 2008 {ACM} {SIGPLAN}-{SIGBED} conference on Languages, compilers, and tools for embedded systems}}. \bibinfo{publisher}{{ACM}}.
\newblock
\href{https://doi.org/10.1145/1375657.1375673}{doi:\nolinkurl{10.1145/1375657.1375673}}


\bibitem[Cotumaccio and Prezza(2021)]%
        {Cotumaccio_2021}
\bibfield{author}{\bibinfo{person}{Nicola Cotumaccio} {and} \bibinfo{person}{Nicola Prezza}.} \bibinfo{year}{2021}\natexlab{}.
\newblock \showarticletitle{On Indexing and Compressing Finite Automata}.
\newblock In \bibinfo{booktitle}{\emph{Proceedings of the 2021 {ACM}-{SIAM} Symposium on Discrete Algorithms ({SODA})}}. \bibinfo{publisher}{Society for Industrial and Applied Mathematics}, \bibinfo{pages}{2585--2599}.
\newblock
\href{https://doi.org/10.1137/1.9781611976465.153}{doi:\nolinkurl{10.1137/1.9781611976465.153}}


\bibitem[Cox(2012)]%
        {cox2012regular}
\bibfield{author}{\bibinfo{person}{Russ Cox}.} \bibinfo{year}{2012}\natexlab{}.
\newblock \bibinfo{title}{Regular Expression Matching with a Trigram Index or How Google Code Search Worked}.
\newblock
\urldef\tempurl%
\url{https://swtch.com/%7Ersc/regexp/regexp4.html}
\showURL{%
\tempurl}


\bibitem[Crochemore et~al\mbox{.}(1994)]%
        {Crochemore_1994}
\bibfield{author}{\bibinfo{person}{M. Crochemore}, \bibinfo{person}{A. Czumaj}, \bibinfo{person}{L. Gasieniec}, \bibinfo{person}{S. Jarominek}, \bibinfo{person}{T. Lecroq}, \bibinfo{person}{W. Plandowski}, {and} \bibinfo{person}{W. Rytter}.} \bibinfo{year}{1994}\natexlab{}.
\newblock \showarticletitle{Speeding up two string-matching algorithms}.
\newblock \bibinfo{journal}{\emph{Algorithmica}} \bibinfo{volume}{12}, \bibinfo{number}{4-5} (\bibinfo{date}{nov} \bibinfo{year}{1994}), \bibinfo{pages}{247--267}.
\newblock
\href{https://doi.org/10.1007/bf01185427}{doi:\nolinkurl{10.1007/bf01185427}}


\bibitem[Crochemore et~al\mbox{.}(1999)]%
        {Crochemore_1999}
\bibfield{author}{\bibinfo{person}{Maxime Crochemore}, \bibinfo{person}{A. Czumaj}, \bibinfo{person}{L. Gasieniec}, \bibinfo{person}{T. Lecroq}, \bibinfo{person}{W. Plandowski}, {and} \bibinfo{person}{W. Rytter}.} \bibinfo{year}{1999}\natexlab{}.
\newblock \showarticletitle{Fast practical multi-pattern matching}.
\newblock \bibinfo{journal}{\emph{Inform. Process. Lett.}} \bibinfo{volume}{71}, \bibinfo{number}{3-4} (\bibinfo{date}{aug} \bibinfo{year}{1999}), \bibinfo{pages}{107--113}.
\newblock
\href{https://doi.org/10.1016/s0020-0190(99)00092-7}{doi:\nolinkurl{10.1016/s0020-0190(99)00092-7}}


\bibitem[de~Amo and Furtado(2007)]%
        {de_Amo_2007}
\bibfield{author}{\bibinfo{person}{Sandra de Amo} {and} \bibinfo{person}{Daniel~A. Furtado}.} \bibinfo{year}{2007}\natexlab{}.
\newblock \showarticletitle{First-order temporal pattern mining with regular expression constraints}.
\newblock \bibinfo{journal}{\emph{Data and Knowledge Engineering}} \bibinfo{volume}{62}, \bibinfo{number}{3} (\bibinfo{date}{sep} \bibinfo{year}{2007}), \bibinfo{pages}{401--420}.
\newblock
\href{https://doi.org/10.1016/j.datak.2006.08.009}{doi:\nolinkurl{10.1016/j.datak.2006.08.009}}


\bibitem[Ding et~al\mbox{.}(2024)]%
        {Ding_2024}
\bibfield{author}{\bibinfo{person}{Jialin Ding}, \bibinfo{person}{Matt Abrams}, \bibinfo{person}{Sanghita Bandyopadhyay}, \bibinfo{person}{Luciano Di~Palma}, \bibinfo{person}{Yanzhu Ji}, \bibinfo{person}{Davide Pagano}, \bibinfo{person}{Gopal Paliwal}, \bibinfo{person}{Panos Parchas}, \bibinfo{person}{Pascal Pfeil}, \bibinfo{person}{Orestis Polychroniou}, \bibinfo{person}{Gaurav Saxena}, \bibinfo{person}{Aamer Shah}, \bibinfo{person}{Amina Voloder}, \bibinfo{person}{Sherry Xiao}, \bibinfo{person}{Davis Zhang}, {and} \bibinfo{person}{Tim Kraska}.} \bibinfo{year}{2024}\natexlab{}.
\newblock \showarticletitle{Automated Multidimensional Data Layouts in Amazon Redshift}. In \bibinfo{booktitle}{\emph{Companion of the 2024 International Conference on Management of Data}} \emph{(\bibinfo{series}{SIGMOD/PODS ’24})}. \bibinfo{publisher}{ACM}, \bibinfo{pages}{55–67}.
\newblock
\href{https://doi.org/10.1145/3626246.3653379}{doi:\nolinkurl{10.1145/3626246.3653379}}


\bibitem[Ellis et~al\mbox{.}(2015)]%
        {Ellis_2015}
\bibfield{author}{\bibinfo{person}{Jason Ellis}, \bibinfo{person}{Achille Fokoue}, \bibinfo{person}{Oktie Hassanzadeh}, \bibinfo{person}{Anastasios Kementsietsidis}, \bibinfo{person}{Kavitha Srinivas}, {and} \bibinfo{person}{Michael~J. Ward}.} \bibinfo{year}{2015}\natexlab{}.
\newblock \showarticletitle{Exploring Big Data with Helix}.
\newblock \bibinfo{journal}{\emph{{ACM} {SIGMOD} Record}} \bibinfo{volume}{43}, \bibinfo{number}{4} (\bibinfo{date}{feb} \bibinfo{year}{2015}), \bibinfo{pages}{43--54}.
\newblock
\href{https://doi.org/10.1145/2737817.2737829}{doi:\nolinkurl{10.1145/2737817.2737829}}


\bibitem[Faloutsos and Christodoulakis(1984)]%
        {Faloutsos_1984}
\bibfield{author}{\bibinfo{person}{Chris Faloutsos} {and} \bibinfo{person}{Stavros Christodoulakis}.} \bibinfo{year}{1984}\natexlab{}.
\newblock \showarticletitle{Signature files: an access method for documents and its analytical performance evaluation}.
\newblock \bibinfo{journal}{\emph{ACM Transactions on Information Systems}} \bibinfo{volume}{2}, \bibinfo{number}{4} (\bibinfo{date}{Oct.} \bibinfo{year}{1984}), \bibinfo{pages}{267–288}.
\newblock
\showISSN{1558-2868}
\href{https://doi.org/10.1145/2275.357411}{doi:\nolinkurl{10.1145/2275.357411}}


\bibitem[Ficara et~al\mbox{.}(2011)]%
        {Ficara_2011}
\bibfield{author}{\bibinfo{person}{Domenico Ficara}, \bibinfo{person}{Andrea~Di Pietro}, \bibinfo{person}{Stefano Giordano}, \bibinfo{person}{Gregorio Procissi}, \bibinfo{person}{Fabio Vitucci}, {and} \bibinfo{person}{Gianni Antichi}.} \bibinfo{year}{2011}\natexlab{}.
\newblock \showarticletitle{Differential Encoding of {DFAs} for Fast Regular Expression Matching}.
\newblock \bibinfo{journal}{\emph{{IEEE}/{ACM} Transactions on Networking}} \bibinfo{volume}{19}, \bibinfo{number}{3} (\bibinfo{date}{jun} \bibinfo{year}{2011}), \bibinfo{pages}{683--694}.
\newblock
\href{https://doi.org/10.1109/tnet.2010.2089639}{doi:\nolinkurl{10.1109/tnet.2010.2089639}}


\bibitem[Gao et~al\mbox{.}(2006)]%
        {Ming_Gao_2006}
\bibfield{author}{\bibinfo{person}{Ming Gao}, \bibinfo{person}{Kenong Zhang}, {and} \bibinfo{person}{Jiahua Lu}.} \bibinfo{year}{2006}\natexlab{}.
\newblock \showarticletitle{Efficient packet matching for gigabit network intrusion detection using {TCAMs}}. In \bibinfo{booktitle}{\emph{20th International Conference on Advanced Information Networking and Applications - Volume 1 ({AINA}{\textquotesingle}06)}}. \bibinfo{publisher}{{IEEE}}.
\newblock
\href{https://doi.org/10.1109/aina.2006.164}{doi:\nolinkurl{10.1109/aina.2006.164}}


\bibitem[Garofalakis et~al\mbox{.}(2002)]%
        {Garofalakis_2002}
\bibfield{author}{\bibinfo{person}{M. Garofalakis}, \bibinfo{person}{R. Rastogi}, {and} \bibinfo{person}{K. Shim}.} \bibinfo{year}{2002}\natexlab{}.
\newblock \showarticletitle{Mining sequential patterns with regular expression constraints}.
\newblock \bibinfo{journal}{\emph{{IEEE} Transactions on Knowledge and Data Engineering}} \bibinfo{volume}{14}, \bibinfo{number}{3} (\bibinfo{date}{may} \bibinfo{year}{2002}), \bibinfo{pages}{530--552}.
\newblock
\href{https://doi.org/10.1109/tkde.2002.1000341}{doi:\nolinkurl{10.1109/tkde.2002.1000341}}


\bibitem[Garofalakis et~al\mbox{.}(1999)]%
        {Garofalakis_1999}
\bibfield{author}{\bibinfo{person}{Minos~N. Garofalakis}, \bibinfo{person}{Rajeev Rastogi}, {and} \bibinfo{person}{Kyuseok Shim}.} \bibinfo{year}{1999}\natexlab{}.
\newblock \showarticletitle{SPIRIT: Sequential Pattern Mining with Regular Expression Constraints}. In \bibinfo{booktitle}{\emph{Proceedings of the 25th International Conference on Very Large Data Bases}} \emph{(\bibinfo{series}{VLDB '99})}. \bibinfo{publisher}{Morgan Kaufmann Publishers Inc.}, \bibinfo{address}{San Francisco, CA, USA}, \bibinfo{pages}{223–234}.
\newblock
\showISBNx{1558606157}


\bibitem[Gibney and Thankachan(2021)]%
        {Gibney_2021}
\bibfield{author}{\bibinfo{person}{Daniel Gibney} {and} \bibinfo{person}{Sharma~V. Thankachan}.} \bibinfo{year}{2021}\natexlab{}.
\newblock \showarticletitle{Text Indexing for Regular Expression Matching}.
\newblock \bibinfo{journal}{\emph{Algorithms}} \bibinfo{volume}{14}, \bibinfo{number}{5} (\bibinfo{date}{apr} \bibinfo{year}{2021}), \bibinfo{pages}{133}.
\newblock
\href{https://doi.org/10.3390/a14050133}{doi:\nolinkurl{10.3390/a14050133}}


\bibitem[Giladi et~al\mbox{.}(2002)]%
        {Giladi_2002}
\bibfield{author}{\bibinfo{person}{Eldar Giladi}, \bibinfo{person}{Michael~G. Walker}, \bibinfo{person}{James~Z. Wang}, {and} \bibinfo{person}{Wayne Volkmuth}.} \bibinfo{year}{2002}\natexlab{}.
\newblock \showarticletitle{{SST}: An Algorithm for Finding Near-Exact Sequence Matches in Time Proportional to The Logarithm of The Database Size}.
\newblock \bibinfo{journal}{\emph{Bioinformatics}} \bibinfo{volume}{18}, \bibinfo{number}{6} (\bibinfo{date}{jun} \bibinfo{year}{2002}), \bibinfo{pages}{873--877}.
\newblock
\href{https://doi.org/10.1093/bioinformatics/18.6.873}{doi:\nolinkurl{10.1093/bioinformatics/18.6.873}}


\bibitem[Goodwin et~al\mbox{.}(2017)]%
        {Goodwin_2017}
\bibfield{author}{\bibinfo{person}{Bob Goodwin}, \bibinfo{person}{Michael Hopcroft}, \bibinfo{person}{Dan Luu}, \bibinfo{person}{Alex Clemmer}, \bibinfo{person}{Mihaela Curmei}, \bibinfo{person}{Sameh Elnikety}, {and} \bibinfo{person}{Yuxiong He}.} \bibinfo{year}{2017}\natexlab{}.
\newblock \showarticletitle{BitFunnel: Revisiting Signatures for Search}. In \bibinfo{booktitle}{\emph{Proceedings of the 40th International ACM SIGIR Conference on Research and Development in Information Retrieval}} \emph{(\bibinfo{series}{SIGIR ’17})}. \bibinfo{publisher}{ACM}, \bibinfo{pages}{605–614}.
\newblock
\href{https://doi.org/10.1145/3077136.3080789}{doi:\nolinkurl{10.1145/3077136.3080789}}


\bibitem[{Google}({[n.\,d.]})]%
        {re2}
\bibfield{author}{\bibinfo{person}{{Google}}.} \bibinfo{year}{[n.\,d.]}\natexlab{}.
\newblock \bibinfo{booktitle}{\emph{Google-RE2}}.
\newblock
\urldef\tempurl%
\url{https://github.com/google/re2}
\showURL{%
\tempurl}


\bibitem[Gouret et~al\mbox{.}(2009)]%
        {Gouret_2009}
\bibfield{author}{\bibinfo{person}{Philippe Gouret}, \bibinfo{person}{Julie~D Thompson}, {and} \bibinfo{person}{Pierre Pontarotti}.} \bibinfo{year}{2009}\natexlab{}.
\newblock \showarticletitle{{PhyloPattern}: regular expressions to identify complex patterns in phylogenetic trees}.
\newblock \bibinfo{journal}{\emph{{BMC} Bioinformatics}} \bibinfo{volume}{10}, \bibinfo{number}{1} (\bibinfo{date}{sep} \bibinfo{year}{2009}).
\newblock
\href{https://doi.org/10.1186/1471-2105-10-298}{doi:\nolinkurl{10.1186/1471-2105-10-298}}


\bibitem[Grabowski et~al\mbox{.}(2016)]%
        {Grabowski_2016}
\bibfield{author}{\bibinfo{person}{Szymon Grabowski}, \bibinfo{person}{Robert Susik}, {and} \bibinfo{person}{Marcin Raniszewski}.} \bibinfo{year}{2016}\natexlab{}.
\newblock \showarticletitle{A Bloom filter based semi‐index on q‐grams}.
\newblock \bibinfo{journal}{\emph{Software: Practice and Experience}} \bibinfo{volume}{47}, \bibinfo{number}{6} (\bibinfo{date}{Aug.} \bibinfo{year}{2016}), \bibinfo{pages}{799–811}.
\newblock
\showISSN{1097-024X}
\href{https://doi.org/10.1002/spe.2431}{doi:\nolinkurl{10.1002/spe.2431}}


\bibitem[Gravano et~al\mbox{.}(2001)]%
        {qgram_dbms}
\bibfield{author}{\bibinfo{person}{Luis Gravano}, \bibinfo{person}{Panos Ipeirotis}, \bibinfo{person}{H. Jagadish}, \bibinfo{person}{Nick Koudas}, \bibinfo{person}{Senthilmurugan Muthukrishnan}, \bibinfo{person}{Lauri Pietarinen}, {and} \bibinfo{person}{Divesh Srivastava}.} \bibinfo{year}{2001}\natexlab{}.
\newblock \showarticletitle{Using q-grams in a DBMS for Approximate String Processing.}
\newblock \bibinfo{journal}{\emph{IEEE Data Eng. Bull.}}  \bibinfo{volume}{24} (\bibinfo{year}{2001}), \bibinfo{pages}{28--34}.
\newblock


\bibitem[Halle and Varvaressos(2014)]%
        {Halle_2014}
\bibfield{author}{\bibinfo{person}{Sylvain Halle} {and} \bibinfo{person}{Simon Varvaressos}.} \bibinfo{year}{2014}\natexlab{}.
\newblock \showarticletitle{A Formalization of Complex Event Stream Processing}. In \bibinfo{booktitle}{\emph{2014 {IEEE} 18th International Enterprise Distributed Object Computing Conference}}. \bibinfo{publisher}{{IEEE}}.
\newblock
\href{https://doi.org/10.1109/edoc.2014.12}{doi:\nolinkurl{10.1109/edoc.2014.12}}


\bibitem[Hammel and Patel(2002)]%
        {Hammel_2002}
\bibfield{author}{\bibinfo{person}{Laurie Hammel} {and} \bibinfo{person}{Jignesh~M. Patel}.} \bibinfo{year}{2002}\natexlab{}.
\newblock \showarticletitle{Searching on the Secondary Structure of Protein Sequences}.
\newblock In \bibinfo{booktitle}{\emph{{VLDB} {\textquotesingle}02: Proceedings of the 28th International Conference on Very Large Databases}}. \bibinfo{publisher}{Elsevier}, \bibinfo{pages}{634--645}.
\newblock
\href{https://doi.org/10.1016/b978-155860869-6/50062-7}{doi:\nolinkurl{10.1016/b978-155860869-6/50062-7}}


\bibitem[Hore et~al\mbox{.}(2004)]%
        {BEST}
\bibfield{author}{\bibinfo{person}{Bijit Hore}, \bibinfo{person}{Hakan Hacigumus}, \bibinfo{person}{Bala Iyer}, {and} \bibinfo{person}{Sharad Mehrotra}.} \bibinfo{year}{2004}\natexlab{}.
\newblock \showarticletitle{Indexing text data under space constraints}. In \bibinfo{booktitle}{\emph{Proceedings of the thirteenth ACM international conference on Information and knowledge management}}. \bibinfo{pages}{198--207}.
\newblock


\bibitem[Horspool(1980)]%
        {Horspool_1980}
\bibfield{author}{\bibinfo{person}{R.~Nigel Horspool}.} \bibinfo{year}{1980}\natexlab{}.
\newblock \showarticletitle{Practical fast searching in strings}.
\newblock \bibinfo{journal}{\emph{Software: Practice and Experience}} \bibinfo{volume}{10}, \bibinfo{number}{6} (\bibinfo{date}{jun} \bibinfo{year}{1980}), \bibinfo{pages}{501--506}.
\newblock
\href{https://doi.org/10.1002/spe.4380100608}{doi:\nolinkurl{10.1002/spe.4380100608}}


\bibitem[Hua et~al\mbox{.}(2009)]%
        {Hua_2009}
\bibfield{author}{\bibinfo{person}{N. Hua}, \bibinfo{person}{H. Song}, {and} \bibinfo{person}{T.~V. Lakshman}.} \bibinfo{year}{2009}\natexlab{}.
\newblock \showarticletitle{Variable-Stride Multi-Pattern Matching For Scalable Deep Packet Inspection}. In \bibinfo{booktitle}{\emph{IEEE INFOCOM 2009}}. \bibinfo{pages}{415--423}.
\newblock
\href{https://doi.org/10.1109/INFCOM.2009.5061946}{doi:\nolinkurl{10.1109/INFCOM.2009.5061946}}


\bibitem[Huang et~al\mbox{.}(2013)]%
        {Huang_2013}
\bibfield{author}{\bibinfo{person}{Kun Huang}, \bibinfo{person}{Linxuan Ding}, \bibinfo{person}{Gaogang Xie}, \bibinfo{person}{Dafang Zhang}, \bibinfo{person}{Alex~X. Liu}, {and} \bibinfo{person}{Kave Salamatian}.} \bibinfo{year}{2013}\natexlab{}.
\newblock \showarticletitle{Scalable {TCAM}-based regular expression matching with compressed finite automata}. In \bibinfo{booktitle}{\emph{Architectures for Networking and Communications Systems}}. \bibinfo{publisher}{{IEEE}}.
\newblock
\href{https://doi.org/10.1109/ancs.2013.6665178}{doi:\nolinkurl{10.1109/ancs.2013.6665178}}


\bibitem[Huo et~al\mbox{.}(2015)]%
        {Huo_2015}
\bibfield{author}{\bibinfo{person}{Sheng Huo}, \bibinfo{person}{Dafang Zhang}, {and} \bibinfo{person}{Yanbiao Li}.} \bibinfo{year}{2015}\natexlab{}.
\newblock \showarticletitle{Fast and Scalable Regular Expressions Matching with Multi-Stride Index {NFA}}.
\newblock In \bibinfo{booktitle}{\emph{Algorithms and Architectures for Parallel Processing}}. \bibinfo{publisher}{Springer International Publishing}, \bibinfo{pages}{597--610}.
\newblock
\href{https://doi.org/10.1007/978-3-319-27137-8_43}{doi:\nolinkurl{10.1007/978-3-319-27137-8_43}}


\bibitem[Hyyr{\"o} and Navarro(2002)]%
        {hyyro_2002}
\bibfield{author}{\bibinfo{person}{Heikki Hyyr{\"o}} {and} \bibinfo{person}{Gonzalo Navarro}.} \bibinfo{year}{2002}\natexlab{}.
\newblock \showarticletitle{Faster Bit-Parallel Approximate String Matching}. In \bibinfo{booktitle}{\emph{Combinatorial Pattern Matching}}, \bibfield{editor}{\bibinfo{person}{Alberto Apostolico} {and} \bibinfo{person}{Masayuki Takeda}} (Eds.). \bibinfo{publisher}{Springer Berlin Heidelberg}, \bibinfo{address}{Berlin, Heidelberg}, \bibinfo{pages}{203--224}.
\newblock
\showISBNx{978-3-540-45452-6}
\href{https://doi.org/10.1007/3-540-45452-7_18}{doi:\nolinkurl{10.1007/3-540-45452-7_18}}


\bibitem[Ishikawa et~al\mbox{.}(1993)]%
        {Ishikawa_1993}
\bibfield{author}{\bibinfo{person}{Yoshiharu Ishikawa}, \bibinfo{person}{Hiroyuki Kitagawa}, {and} \bibinfo{person}{Nobuo Ohbo}.} \bibinfo{year}{1993}\natexlab{}.
\newblock \showarticletitle{Evaluation of signature files as set access facilities in OODBs}.
\newblock \bibinfo{journal}{\emph{ACM SIGMOD Record}} \bibinfo{volume}{22}, \bibinfo{number}{2} (\bibinfo{date}{June} \bibinfo{year}{1993}), \bibinfo{pages}{247–256}.
\newblock
\showISSN{0163-5808}
\href{https://doi.org/10.1145/170036.170076}{doi:\nolinkurl{10.1145/170036.170076}}


\bibitem[Jakub{\'{\i}}cek and Rychl{\'{y}}(2014)]%
        {Jakubicek_2014}
\bibfield{author}{\bibinfo{person}{Milos Jakub{\'{\i}}cek} {and} \bibinfo{person}{Pavel Rychl{\'{y}}}.} \bibinfo{year}{2014}\natexlab{}.
\newblock \showarticletitle{Optimization of Regular Expression Evaluation within the Manatee Corpus Management System}. In \bibinfo{booktitle}{\emph{The 8th Workshop on Recent Advances in Slavonic Natural Languages Processing, {RASLAN} 2014, Karlova Studanka, Czech Republic, December 5-7, 2014}}, \bibfield{editor}{\bibinfo{person}{Ales Hor{\'{a}}k} {and} \bibinfo{person}{Pavel Rychl{\'{y}}}} (Eds.). \bibinfo{publisher}{Tribun {EU}}, \bibinfo{pages}{37--48}.
\newblock
\urldef\tempurl%
\url{http://nlp.fi.muni.cz/raslan/2014/1.pdf}
\showURL{%
\tempurl}


\bibitem[Jamshed et~al\mbox{.}(2012)]%
        {Jamshed_2012}
\bibfield{author}{\bibinfo{person}{Muhammad~Asim Jamshed}, \bibinfo{person}{Jihyung Lee}, \bibinfo{person}{Sangwoo Moon}, \bibinfo{person}{Insu Yun}, \bibinfo{person}{Deokjin Kim}, \bibinfo{person}{Sungryoul Lee}, \bibinfo{person}{Yung Yi}, {and} \bibinfo{person}{KyoungSoo Park}.} \bibinfo{year}{2012}\natexlab{}.
\newblock \showarticletitle{Karguss: a highly-scalable software-based intrusion detection system}. In \bibinfo{booktitle}{\emph{Proceedings of the 2012 {ACM} conference on Computer and communications security}}. \bibinfo{publisher}{{ACM}}.
\newblock
\href{https://doi.org/10.1145/2382196.2382232}{doi:\nolinkurl{10.1145/2382196.2382232}}


\bibitem[Jones and Mewhort(2004)]%
        {english_bigram}
\bibfield{author}{\bibinfo{person}{Michael~N. Jones} {and} \bibinfo{person}{D.~J.~K. Mewhort}.} \bibinfo{year}{2004}\natexlab{}.
\newblock \showarticletitle{Case-sensitive letter and bigram frequency counts from large-scale English corpora}.
\newblock \bibinfo{journal}{\emph{Behavior Research Methods, Instruments, Computers}} \bibinfo{volume}{36}, \bibinfo{number}{3} (\bibinfo{date}{aug} \bibinfo{year}{2004}), \bibinfo{pages}{388--396}.
\newblock
\href{https://doi.org/10.3758/bf03195586}{doi:\nolinkurl{10.3758/bf03195586}}


\bibitem[Jung et~al\mbox{.}(2006)]%
        {Hong_Jip_Jung_2006}
\bibfield{author}{\bibinfo{person}{Hong-Jip Jung}, \bibinfo{person}{Z.K. Baker}, {and} \bibinfo{person}{V.K. Prasanna}.} \bibinfo{year}{2006}\natexlab{}.
\newblock \showarticletitle{Performance of {FPGA} implementation of bit-split architecture for intrusion detection systems}. In \bibinfo{booktitle}{\emph{Proceedings 20th {IEEE} International Parallel and Distributed Processing Symposium}}. \bibinfo{publisher}{{IEEE}}.
\newblock
\href{https://doi.org/10.1109/ipdps.2006.1639434}{doi:\nolinkurl{10.1109/ipdps.2006.1639434}}


\bibitem[Kandhan et~al\mbox{.}(2010)]%
        {Kandhan_2010}
\bibfield{author}{\bibinfo{person}{Ramakrishnan Kandhan}, \bibinfo{person}{Nikhil Teletia}, {and} \bibinfo{person}{Jignesh~M. Patel}.} \bibinfo{year}{2010}\natexlab{}.
\newblock \showarticletitle{{SigMatch}: Fast And Scalable Multi-Pattern Matching}.
\newblock \bibinfo{journal}{\emph{Proceedings of the {VLDB} Endowment}} \bibinfo{volume}{3}, \bibinfo{number}{1-2} (\bibinfo{date}{sep} \bibinfo{year}{2010}), \bibinfo{pages}{1173--1184}.
\newblock
\href{https://doi.org/10.14778/1920841.1920987}{doi:\nolinkurl{10.14778/1920841.1920987}}


\bibitem[Kaneta et~al\mbox{.}(2010)]%
        {Kaneta_2010}
\bibfield{author}{\bibinfo{person}{Yusaku Kaneta}, \bibinfo{person}{Shin-Ichi Minato}, {and} \bibinfo{person}{Hiroki Arimura}.} \bibinfo{year}{2010}\natexlab{}.
\newblock \showarticletitle{Fast Bit-Parallel Matching for Network and Regular Expressions}. In \bibinfo{booktitle}{\emph{Proceedings of the 17th International Conference on String Processing and Information Retrieval}} (Los Cabos, Mexico) \emph{(\bibinfo{series}{SPIRE'10})}. \bibinfo{publisher}{Springer-Verlag}, \bibinfo{address}{Berlin, Heidelberg}, \bibinfo{pages}{372–384}.
\newblock
\showISBNx{3642163203}
\href{https://doi.org/10.1007/978-3-642-16321-0_39}{doi:\nolinkurl{10.1007/978-3-642-16321-0_39}}


\bibitem[Kent et~al\mbox{.}(1990)]%
        {Kent_1990}
\bibfield{author}{\bibinfo{person}{A. Kent}, \bibinfo{person}{R. Sacks-Davis}, {and} \bibinfo{person}{K. Ramamohanarao}.} \bibinfo{year}{1990}\natexlab{}.
\newblock \showarticletitle{A signature file scheme based on multiple organizations for indexing very large text databases}.
\newblock \bibinfo{journal}{\emph{Journal of the American Society for Information Science}} \bibinfo{volume}{41}, \bibinfo{number}{7} (\bibinfo{date}{Oct.} \bibinfo{year}{1990}), \bibinfo{pages}{508–534}.
\newblock
\showISSN{1097-4571}
\href{https://doi.org/10.1002/(sici)1097-4571(199010)41:7<508::aid-asi5>3.0.co;2-j}{doi:\nolinkurl{10.1002/(sici)1097-4571(199010)41:7<508::aid-asi5>3.0.co;2-j}}


\bibitem[Kim et~al\mbox{.}(2010)]%
        {Kim_2010}
\bibfield{author}{\bibinfo{person}{Younghoon Kim}, \bibinfo{person}{Kyoung-Gu Woo}, \bibinfo{person}{Hyoungmin Park}, {and} \bibinfo{person}{Kyuseok Shim}.} \bibinfo{year}{2010}\natexlab{}.
\newblock \showarticletitle{Efficient processing of substring match queries with inverted q-gram indexes}. In \bibinfo{booktitle}{\emph{2010 IEEE 26th International Conference on Data Engineering (ICDE 2010)}}. \bibinfo{pages}{721--732}.
\newblock
\href{https://doi.org/10.1109/ICDE.2010.5447866}{doi:\nolinkurl{10.1109/ICDE.2010.5447866}}


\bibitem[Knuth et~al\mbox{.}(1977)]%
        {KMP}
\bibfield{author}{\bibinfo{person}{Donald~E. Knuth}, \bibinfo{person}{James~H. Morris, Jr.}, {and} \bibinfo{person}{Vaughan~R. Pratt}.} \bibinfo{year}{1977}\natexlab{}.
\newblock \showarticletitle{Fast Pattern Matching in Strings}.
\newblock \bibinfo{journal}{\emph{SIAM J. Comput.}} \bibinfo{volume}{6}, \bibinfo{number}{2} (\bibinfo{year}{1977}), \bibinfo{pages}{323--350}.
\newblock
\showeprint{https://doi.org/10.1137/0206024}
\href{https://doi.org/10.1137/0206024}{doi:\nolinkurl{10.1137/0206024}}


\bibitem[Kong et~al\mbox{.}(2008)]%
        {Kong_2008}
\bibfield{author}{\bibinfo{person}{Shijin Kong}, \bibinfo{person}{Randy Smith}, {and} \bibinfo{person}{Cristian Estan}.} \bibinfo{year}{2008}\natexlab{}.
\newblock \showarticletitle{Efficient signature matching with multiple alphabet compression tables}. In \bibinfo{booktitle}{\emph{Proceedings of the 4th international conference on Security and privacy in communication netowrks}}. \bibinfo{publisher}{{ACM}}.
\newblock
\href{https://doi.org/10.1145/1460877.1460879}{doi:\nolinkurl{10.1145/1460877.1460879}}


\bibitem[Kumar and Singh(2012)]%
        {Kumar_2012}
\bibfield{author}{\bibinfo{person}{Pawan Kumar} {and} \bibinfo{person}{Virendra Singh}.} \bibinfo{year}{2012}\natexlab{}.
\newblock \showarticletitle{Efficient regular expression pattern matching for network intrusion detection systems using modified word-based automata}. In \bibinfo{booktitle}{\emph{Proceedings of the Fifth International Conference on Security of Information and Networks}}. \bibinfo{publisher}{{ACM}}.
\newblock
\href{https://doi.org/10.1145/2388576.2388590}{doi:\nolinkurl{10.1145/2388576.2388590}}


\bibitem[Kumar et~al\mbox{.}(2006)]%
        {Kumar_2006}
\bibfield{author}{\bibinfo{person}{Sailesh Kumar}, \bibinfo{person}{Sarang Dharmapurikar}, \bibinfo{person}{Fang Yu}, \bibinfo{person}{Patrick Crowley}, {and} \bibinfo{person}{Jonathan Turner}.} \bibinfo{year}{2006}\natexlab{}.
\newblock \showarticletitle{Algorithms to Accelerate Multiple Regular Expressions Matching for Deep Packet Inspection}. In \bibinfo{booktitle}{\emph{Proceedings of the 2006 Conference on Applications, Technologies, Architectures, and Protocols for Computer Communications}} (Pisa, Italy) \emph{(\bibinfo{series}{SIGCOMM '06})}. \bibinfo{publisher}{Association for Computing Machinery}, \bibinfo{address}{New York, NY, USA}, \bibinfo{pages}{339–350}.
\newblock
\showISBNx{1595933085}
\href{https://doi.org/10.1145/1159913.1159952}{doi:\nolinkurl{10.1145/1159913.1159952}}


\bibitem[Lee et~al\mbox{.}(2015)]%
        {Lee_2015}
\bibfield{author}{\bibinfo{person}{Chun-Liang Lee}, \bibinfo{person}{Chung-Yuan Huang}, \bibinfo{person}{Kai-Ping Lu}, {and} \bibinfo{person}{Jhao-Han Chen}.} \bibinfo{year}{2015}\natexlab{}.
\newblock \showarticletitle{A Fast and Scalable Multi-Pattern Matching Algorithm for Intrusion Detection Systems}. In \bibinfo{booktitle}{\emph{The Proceedings of the 2nd International Conference on Industrial Application Engineering 2015}}. \bibinfo{publisher}{The Institute of Industrial Applications Engineers}.
\newblock
\href{https://doi.org/10.12792/iciae2015.057}{doi:\nolinkurl{10.12792/iciae2015.057}}


\bibitem[Lin and Faloutsos(1992)]%
        {Lin_1992}
\bibfield{author}{\bibinfo{person}{Z. Lin} {and} \bibinfo{person}{C. Faloutsos}.} \bibinfo{year}{1992}\natexlab{}.
\newblock \showarticletitle{Frame-sliced signature files}.
\newblock \bibinfo{journal}{\emph{IEEE Transactions on Knowledge and Data Engineering}} \bibinfo{volume}{4}, \bibinfo{number}{3} (\bibinfo{date}{June} \bibinfo{year}{1992}), \bibinfo{pages}{281–289}.
\newblock
\showISSN{1041-4347}
\href{https://doi.org/10.1109/69.142018}{doi:\nolinkurl{10.1109/69.142018}}


\bibitem[Liu et~al\mbox{.}(2004)]%
        {Liu_2004}
\bibfield{author}{\bibinfo{person}{Rong-Tai Liu}, \bibinfo{person}{Nen-Fu Huang}, \bibinfo{person}{Chih-Hao Chen}, {and} \bibinfo{person}{Chia-Nan Kao}.} \bibinfo{year}{2004}\natexlab{}.
\newblock \showarticletitle{A fast string-matching algorithm for network processor-based intrusion detection system}.
\newblock \bibinfo{journal}{\emph{{ACM} Transactions on Embedded Computing Systems}} \bibinfo{volume}{3}, \bibinfo{number}{3} (\bibinfo{date}{aug} \bibinfo{year}{2004}), \bibinfo{pages}{614--633}.
\newblock
\href{https://doi.org/10.1145/1015047.1015055}{doi:\nolinkurl{10.1145/1015047.1015055}}


\bibitem[Meek et~al\mbox{.}(2003)]%
        {Meek_2003}
\bibfield{author}{\bibinfo{person}{Colin Meek}, \bibinfo{person}{Jignesh~M. Patel}, {and} \bibinfo{person}{Shruti Kasetty}.} \bibinfo{year}{2003}\natexlab{}.
\newblock \showarticletitle{{OASIS}: An Online and Accurate Technique for Local-alignment Searches on Biological Sequences}.
\newblock In \bibinfo{booktitle}{\emph{Proceedings 2003 {VLDB} Conference}}. \bibinfo{publisher}{Elsevier}, \bibinfo{pages}{910--921}.
\newblock
\href{https://doi.org/10.1016/b978-012722442-8/50085-9}{doi:\nolinkurl{10.1016/b978-012722442-8/50085-9}}


\bibitem[Meiners et~al\mbox{.}(2010)]%
        {Meiners_2010}
\bibfield{author}{\bibinfo{person}{Chad~R. Meiners}, \bibinfo{person}{Jignesh Patel}, \bibinfo{person}{Eric Norige}, \bibinfo{person}{Eric Torng}, {and} \bibinfo{person}{Alex~X. Liu}.} \bibinfo{year}{2010}\natexlab{}.
\newblock \showarticletitle{Fast Regular Expression Matching Using Small TCAMs for Network Intrusion Detection and Prevention Systems}. In \bibinfo{booktitle}{\emph{Proceedings of the 19th USENIX Conference on Security}} (Washington, DC) \emph{(\bibinfo{series}{USENIX Security'10})}. \bibinfo{publisher}{USENIX Association}, \bibinfo{address}{USA}, \bibinfo{pages}{8}.
\newblock
\showISBNx{8887666655554}


\bibitem[Moosavi et~al\mbox{.}(2019)]%
        {Moosavi19}
\bibfield{author}{\bibinfo{person}{Sobhan Moosavi}, \bibinfo{person}{Mohammad~Hossein Samavatian}, \bibinfo{person}{Srinivasan Parthasarathy}, \bibinfo{person}{Radu Teodorescu}, {and} \bibinfo{person}{Rajiv Ramnath}.} \bibinfo{year}{2019}\natexlab{}.
\newblock \showarticletitle{Accident Risk Prediction Based on Heterogeneous Sparse Data: New Dataset and Insights}. In \bibinfo{booktitle}{\emph{Proceedings of the 27th ACM SIGSPATIAL International Conference on Advances in Geographic Information Systems}} (Chicago, IL, USA) \emph{(\bibinfo{series}{SIGSPATIAL '19})}. \bibinfo{publisher}{Association for Computing Machinery}, \bibinfo{address}{New York, NY, USA}, \bibinfo{pages}{33–42}.
\newblock
\showISBNx{9781450369091}
\href{https://doi.org/10.1145/3347146.3359078}{doi:\nolinkurl{10.1145/3347146.3359078}}


\bibitem[Moraru and Andersen(2012)]%
        {Moraru_2012}
\bibfield{author}{\bibinfo{person}{Iulian Moraru} {and} \bibinfo{person}{David~G. Andersen}.} \bibinfo{year}{2012}\natexlab{}.
\newblock \showarticletitle{Exact Pattern Matching with Feed-Forward Bloom Filters}.
\newblock \bibinfo{journal}{\emph{ACM J. Exp. Algorithmics}}  \bibinfo{volume}{17}, Article \bibinfo{articleno}{3.4} (\bibinfo{date}{sep} \bibinfo{year}{2012}), \bibinfo{numpages}{18}~pages.
\newblock
\showISSN{1084-6654}
\href{https://doi.org/10.1145/2133803.2330085}{doi:\nolinkurl{10.1145/2133803.2330085}}


\bibitem[Mulder and Nezlek(2006)]%
        {Mulder_2006}
\bibfield{author}{\bibinfo{person}{M. Mulder} {and} \bibinfo{person}{G.S. Nezlek}.} \bibinfo{year}{2006}\natexlab{}.
\newblock \showarticletitle{Creating protein sequence patterns using efficient regular expressions in bioinformatics research}. In \bibinfo{booktitle}{\emph{28th International Conference on Information Technology Interfaces, 2006.}} \bibinfo{publisher}{{IEEE}}.
\newblock
\href{https://doi.org/10.1109/iti.2006.1708479}{doi:\nolinkurl{10.1109/iti.2006.1708479}}


\bibitem[Myers(1999)]%
        {Myers_1999}
\bibfield{author}{\bibinfo{person}{Gene Myers}.} \bibinfo{year}{1999}\natexlab{}.
\newblock \showarticletitle{A fast bit-vector algorithm for approximate string matching based on dynamic programming}.
\newblock \bibinfo{journal}{\emph{Journal of the {ACM}}} \bibinfo{volume}{46}, \bibinfo{number}{3} (\bibinfo{date}{may} \bibinfo{year}{1999}), \bibinfo{pages}{395--415}.
\newblock
\href{https://doi.org/10.1145/316542.316550}{doi:\nolinkurl{10.1145/316542.316550}}


\bibitem[Najam et~al\mbox{.}(2014)]%
        {Najam_2014}
\bibfield{author}{\bibinfo{person}{Maleeha Najam}, \bibinfo{person}{Usman Younis}, {and} \bibinfo{person}{Raihan~Ur Rasool}.} \bibinfo{year}{2014}\natexlab{}.
\newblock \showarticletitle{Multi-byte Pattern Matching Using Stride-K {DFA} for High Speed Deep Packet Inspection}. In \bibinfo{booktitle}{\emph{2014 {IEEE} 17th International Conference on Computational Science and Engineering}}. \bibinfo{publisher}{{IEEE}}.
\newblock
\href{https://doi.org/10.1109/cse.2014.125}{doi:\nolinkurl{10.1109/cse.2014.125}}


\bibitem[Navarro and Baeza-Yates(2018)]%
        {Navarro_2018}
\bibfield{author}{\bibinfo{person}{Gonzalo Navarro} {and} \bibinfo{person}{Ricardo Baeza-Yates}.} \bibinfo{year}{2018}\natexlab{}.
\newblock \showarticletitle{A Practical q -Gram Index for Text Retrieval Allowing Errors}.
\newblock \bibinfo{journal}{\emph{{CLEI} Electronic Journal}} \bibinfo{volume}{1}, \bibinfo{number}{2} (\bibinfo{date}{sep} \bibinfo{year}{2018}).
\newblock
\href{https://doi.org/10.19153/cleiej.1.2.3}{doi:\nolinkurl{10.19153/cleiej.1.2.3}}


\bibitem[Navarro and Raffinot(2000)]%
        {Navarro_2000}
\bibfield{author}{\bibinfo{person}{Gonzalo Navarro} {and} \bibinfo{person}{Mathieu Raffinot}.} \bibinfo{year}{2000}\natexlab{}.
\newblock \showarticletitle{Fast and flexible string matching by combining bit-parallelism and suffix automata}.
\newblock \bibinfo{journal}{\emph{{ACM} Journal of Experimental Algorithmics}}  \bibinfo{volume}{5} (\bibinfo{date}{dec} \bibinfo{year}{2000}), \bibinfo{pages}{4}.
\newblock
\href{https://doi.org/10.1145/351827.384246}{doi:\nolinkurl{10.1145/351827.384246}}


\bibitem[Navarro and Salmela(2009)]%
        {Navarro_2009}
\bibfield{author}{\bibinfo{person}{Gonzalo Navarro} {and} \bibinfo{person}{Leena Salmela}.} \bibinfo{year}{2009}\natexlab{}.
\newblock \showarticletitle{Indexing Variable Length Substrings for Exact and Approximate Matching}.
\newblock In \bibinfo{booktitle}{\emph{String Processing and Information Retrieval}}. \bibinfo{publisher}{Springer Berlin Heidelberg}, \bibinfo{pages}{214--221}.
\newblock
\href{https://doi.org/10.1007/978-3-642-03784-9_21}{doi:\nolinkurl{10.1007/978-3-642-03784-9_21}}


\bibitem[Peng et~al\mbox{.}(2011)]%
        {Peng_2011}
\bibfield{author}{\bibinfo{person}{Kunyang Peng}, \bibinfo{person}{Siyuan Tang}, \bibinfo{person}{Min Chen}, {and} \bibinfo{person}{Qunfeng Dong}.} \bibinfo{year}{2011}\natexlab{}.
\newblock \showarticletitle{Chain-Based {DFA} Deflation for Fast and Scalable Regular Expression Matching Using {TCAM}}. In \bibinfo{booktitle}{\emph{2011 {ACM}/{IEEE} Seventh Symposium on Architectures for Networking and Communications Systems}}. \bibinfo{publisher}{{IEEE}}.
\newblock
\href{https://doi.org/10.1109/ancs.2011.13}{doi:\nolinkurl{10.1109/ancs.2011.13}}


\bibitem[Prieto et~al\mbox{.}(2014)]%
        {Prieto_2014}
\bibfield{author}{\bibinfo{person}{Gorka Prieto}, \bibinfo{person}{Asier Fullaondo}, {and} \bibinfo{person}{Jose~A. Rodriguez}.} \bibinfo{year}{2014}\natexlab{}.
\newblock \showarticletitle{Prediction of nuclear export signals using weighted regular expressions (Wregex)}.
\newblock \bibinfo{journal}{\emph{Bioinformatics}} \bibinfo{volume}{30}, \bibinfo{number}{9} (\bibinfo{date}{jan} \bibinfo{year}{2014}), \bibinfo{pages}{1220--1227}.
\newblock
\href{https://doi.org/10.1093/bioinformatics/btu016}{doi:\nolinkurl{10.1093/bioinformatics/btu016}}


\bibitem[Qiu et~al\mbox{.}(2022)]%
        {Qiu_2022}
\bibfield{author}{\bibinfo{person}{Tao Qiu}, \bibinfo{person}{Xiaochun Yang}, \bibinfo{person}{Bin Wang}, {and} \bibinfo{person}{Wei Wang}.} \bibinfo{year}{2022}\natexlab{}.
\newblock \showarticletitle{Efficient Regular Expression Matching Based on Positional Inverted Index}.
\newblock \bibinfo{journal}{\emph{{IEEE} Transactions on Knowledge and Data Engineering}} \bibinfo{volume}{34}, \bibinfo{number}{3} (\bibinfo{date}{mar} \bibinfo{year}{2022}), \bibinfo{pages}{1133--1148}.
\newblock
\href{https://doi.org/10.1109/tkde.2020.2992295}{doi:\nolinkurl{10.1109/tkde.2020.2992295}}


\bibitem[Roan et~al\mbox{.}(2006)]%
        {Roan_2006}
\bibfield{author}{\bibinfo{person}{Huang-chun Roan}, \bibinfo{person}{Wen-jyi Hwang}, {and} \bibinfo{person}{Chia-tien~Dan Lo}.} \bibinfo{year}{2006}\natexlab{}.
\newblock \showarticletitle{Shift-Or Circuit for Efficient Network Intrusion Detection Pattern Matching}. In \bibinfo{booktitle}{\emph{2006 International Conference on Field Programmable Logic and Applications}}. \bibinfo{pages}{1--6}.
\newblock
\href{https://doi.org/10.1109/FPL.2006.311314}{doi:\nolinkurl{10.1109/FPL.2006.311314}}


\bibitem[Roesch et~al\mbox{.}(1999)]%
        {roesch1999snort}
\bibfield{author}{\bibinfo{person}{Martin Roesch} {et~al\mbox{.}}} \bibinfo{year}{1999}\natexlab{}.
\newblock \showarticletitle{Snort: Lightweight intrusion detection for networks.}. In \bibinfo{booktitle}{\emph{Lisa}}, Vol.~\bibinfo{volume}{99}. \bibinfo{pages}{229--238}.
\newblock


\bibitem[Ruano et~al\mbox{.}(2021)]%
        {Ruano_2021}
\bibfield{author}{\bibinfo{person}{Darío Ruano}, \bibinfo{person}{Norma Herrera}, \bibinfo{person}{Jésica Cornejo}, {and} \bibinfo{person}{Paola Azar}.} \bibinfo{year}{2021}\natexlab{}.
\newblock \bibinfo{booktitle}{\emph{Sequential Representation of Suffix Trie: An Empirical Evaluation}}.
\newblock \bibinfo{publisher}{Springer International Publishing}, \bibinfo{pages}{182–196}.
\newblock
\showISBNx{9783030758363}
\showISSN{1865-0937}
\href{https://doi.org/10.1007/978-3-030-75836-3_13}{doi:\nolinkurl{10.1007/978-3-030-75836-3_13}}


\bibitem[Sacks-Davis et~al\mbox{.}(1987)]%
        {Sacks_Davis_1987}
\bibfield{author}{\bibinfo{person}{R. Sacks-Davis}, \bibinfo{person}{A. Kent}, {and} \bibinfo{person}{K. Ramamohanarao}.} \bibinfo{year}{1987}\natexlab{}.
\newblock \showarticletitle{Multikey access methods based on superimposed coding techniques}.
\newblock \bibinfo{journal}{\emph{ACM Transactions on Database Systems}} \bibinfo{volume}{12}, \bibinfo{number}{4} (\bibinfo{date}{Nov.} \bibinfo{year}{1987}), \bibinfo{pages}{655–696}.
\newblock
\showISSN{1557-4644}
\href{https://doi.org/10.1145/32204.32222}{doi:\nolinkurl{10.1145/32204.32222}}


\bibitem[Scarpazza and Russell(2009)]%
        {Scarpazza_2009}
\bibfield{author}{\bibinfo{person}{Daniele~Paolo Scarpazza} {and} \bibinfo{person}{Gregory~F. Russell}.} \bibinfo{year}{2009}\natexlab{}.
\newblock \showarticletitle{High-performance regular expression scanning on the Cell/B.E. processor}. In \bibinfo{booktitle}{\emph{Proceedings of the 23rd international conference on Supercomputing}}. \bibinfo{publisher}{{ACM}}.
\newblock
\href{https://doi.org/10.1145/1542275.1542284}{doi:\nolinkurl{10.1145/1542275.1542284}}


\bibitem[SHIN et~al\mbox{.}(2018)]%
        {SHIN_2018}
\bibfield{author}{\bibinfo{person}{Seon-Ho SHIN}, \bibinfo{person}{HyunBong KIM}, {and} \bibinfo{person}{MyungKeun YOON}.} \bibinfo{year}{2018}\natexlab{}.
\newblock \showarticletitle{Regular Expression Filtering on Multiple q-Grams}.
\newblock \bibinfo{journal}{\emph{{IEICE} Transactions on Information and Systems}} \bibinfo{volume}{E101.D}, \bibinfo{number}{1} (\bibinfo{year}{2018}), \bibinfo{pages}{253--256}.
\newblock
\href{https://doi.org/10.1587/transinf.2017edl8180}{doi:\nolinkurl{10.1587/transinf.2017edl8180}}


\bibitem[Sidhu and Prasanna(2001)]%
        {Sidhu_2001}
\bibfield{author}{\bibinfo{person}{Reetinder Sidhu} {and} \bibinfo{person}{Viktor~K. Prasanna}.} \bibinfo{year}{2001}\natexlab{}.
\newblock \showarticletitle{Fast Regular Expression Matching Using FPGAs}. In \bibinfo{booktitle}{\emph{Proceedings of the the 9th Annual IEEE Symposium on Field-Programmable Custom Computing Machines}} \emph{(\bibinfo{series}{FCCM '01})}. \bibinfo{publisher}{IEEE Computer Society}, \bibinfo{address}{USA}, \bibinfo{pages}{227–238}.
\newblock
\showISBNx{0769526675}


\bibitem[Smith et~al\mbox{.}(2008)]%
        {Smith_2008}
\bibfield{author}{\bibinfo{person}{Randy Smith}, \bibinfo{person}{Cristian Estan}, {and} \bibinfo{person}{Somesh Jha}.} \bibinfo{year}{2008}\natexlab{}.
\newblock \showarticletitle{{XFA}: Faster Signature Matching with Extended Automata}. In \bibinfo{booktitle}{\emph{2008 {IEEE} Symposium on Security and Privacy (sp 2008)}}. \bibinfo{publisher}{{IEEE}}.
\newblock
\href{https://doi.org/10.1109/sp.2008.14}{doi:\nolinkurl{10.1109/sp.2008.14}}


\bibitem[SUGIURA and ISHIKAWA(2020)]%
        {SUGIURA_2020}
\bibfield{author}{\bibinfo{person}{Kento SUGIURA} {and} \bibinfo{person}{Yoshiharu ISHIKAWA}.} \bibinfo{year}{2020}\natexlab{}.
\newblock \showarticletitle{Multiple Regular Expression Pattern Monitoring over Probabilistic Event Streams}.
\newblock \bibinfo{journal}{\emph{{IEICE} Transactions on Information and Systems}} \bibinfo{volume}{E103.D}, \bibinfo{number}{5} (\bibinfo{date}{may} \bibinfo{year}{2020}), \bibinfo{pages}{982--991}.
\newblock
\href{https://doi.org/10.1587/transinf.2019dap0009}{doi:\nolinkurl{10.1587/transinf.2019dap0009}}


\bibitem[Sunday(1990)]%
        {Sunday_1990}
\bibfield{author}{\bibinfo{person}{Daniel~M. Sunday}.} \bibinfo{year}{1990}\natexlab{}.
\newblock \showarticletitle{A very fast substring search algorithm}.
\newblock \bibinfo{journal}{\emph{Commun. ACM}} \bibinfo{volume}{33}, \bibinfo{number}{8} (\bibinfo{date}{aug} \bibinfo{year}{1990}), \bibinfo{pages}{132--142}.
\newblock
\href{https://doi.org/10.1145/79173.79184}{doi:\nolinkurl{10.1145/79173.79184}}


\bibitem[Suricata(2009)]%
        {suricata2009home}
\bibfield{author}{\bibinfo{person}{Home Suricata}.} \bibinfo{year}{2009}\natexlab{}.
\newblock \showarticletitle{Home-Suricata}.
\newblock \bibinfo{journal}{\emph{The Open Information Security Foundation}} (\bibinfo{year}{2009}).
\newblock


\bibitem[Tabuchi et~al\mbox{.}(2003)]%
        {Tabuchi_2003}
\bibfield{author}{\bibinfo{person}{Naoshi Tabuchi}, \bibinfo{person}{Eijiro Sumii}, {and} \bibinfo{person}{Akinori Yonezawa}.} \bibinfo{year}{2003}\natexlab{}.
\newblock \showarticletitle{Regular Expression Types for Strings in a Text Processing Language}.
\newblock \bibinfo{journal}{\emph{Electronic Notes in Theoretical Computer Science}}  \bibinfo{volume}{75} (\bibinfo{date}{feb} \bibinfo{year}{2003}), \bibinfo{pages}{95--113}.
\newblock
\href{https://doi.org/10.1016/s1571-0661(04)80781-3}{doi:\nolinkurl{10.1016/s1571-0661(04)80781-3}}


\bibitem[{The Apache Software Foundation}(2022)]%
        {lucene}
\bibfield{author}{\bibinfo{person}{{The Apache Software Foundation}}.} \bibinfo{year}{2022}\natexlab{}.
\newblock \bibinfo{booktitle}{\emph{Apache Lucene 9.4.2 Documentation}}.
\newblock
\urldef\tempurl%
\url{https://lucene.apache.org/core/9_4_2/index.html}
\showURL{%
Retrieved January 17, 2023 from \tempurl}


\bibitem[Trasarti et~al\mbox{.}(2008)]%
        {Trasarti_2008}
\bibfield{author}{\bibinfo{person}{Roberto Trasarti}, \bibinfo{person}{Francesco Bonchi}, {and} \bibinfo{person}{Bart Goethals}.} \bibinfo{year}{2008}\natexlab{}.
\newblock \showarticletitle{Sequence Mining Automata: A New Technique for Mining Frequent Sequences under Regular Expressions}. In \bibinfo{booktitle}{\emph{2008 Eighth {IEEE} International Conference on Data Mining}}. \bibinfo{publisher}{{IEEE}}.
\newblock
\href{https://doi.org/10.1109/icdm.2008.111}{doi:\nolinkurl{10.1109/icdm.2008.111}}


\bibitem[Tripp(2007)]%
        {Tripp_2007}
\bibfield{author}{\bibinfo{person}{Gerald Tripp}.} \bibinfo{year}{2007}\natexlab{}.
\newblock \showarticletitle{Regular expression matching with input compression: a hardware design for use within network intrusion detection systems}.
\newblock \bibinfo{journal}{\emph{Journal in Computer Virology}} \bibinfo{volume}{3}, \bibinfo{number}{2} (\bibinfo{date}{apr} \bibinfo{year}{2007}), \bibinfo{pages}{125--134}.
\newblock
\href{https://doi.org/10.1007/s11416-007-0047-z}{doi:\nolinkurl{10.1007/s11416-007-0047-z}}


\bibitem[Trivedi(2020)]%
        {Trivedi_2020}
\bibfield{author}{\bibinfo{person}{Uday Trivedi}.} \bibinfo{year}{2020}\natexlab{}.
\newblock \showarticletitle{An Optimized Aho-Corasick Multi-Pattern Matching Algorithm for Fast Pattern Matching}. In \bibinfo{booktitle}{\emph{2020 IEEE 17th India Council International Conference (INDICON)}}. \bibinfo{pages}{1--5}.
\newblock
\href{https://doi.org/10.1109/INDICON49873.2020.9342041}{doi:\nolinkurl{10.1109/INDICON49873.2020.9342041}}


\bibitem[Tsang and Chawla(2011)]%
        {FAST}
\bibfield{author}{\bibinfo{person}{Dominic Tsang} {and} \bibinfo{person}{Sanjay Chawla}.} \bibinfo{year}{2011}\natexlab{}.
\newblock \showarticletitle{A robust index for regular expression queries}. In \bibinfo{booktitle}{\emph{Proceedings of the 20th {ACM} international conference on Information and knowledge management}}. \bibinfo{publisher}{{ACM}}.
\newblock
\href{https://doi.org/10.1145/2063576.2063968}{doi:\nolinkurl{10.1145/2063576.2063968}}


\bibitem[Vasiliadis et~al\mbox{.}(2011)]%
        {Vasiliadis_2011}
\bibfield{author}{\bibinfo{person}{Giorgos Vasiliadis}, \bibinfo{person}{Michalis Polychronakis}, {and} \bibinfo{person}{Sotiris Ioannidis}.} \bibinfo{year}{2011}\natexlab{}.
\newblock \showarticletitle{{MIDeA}: A Multi-Parallel Intrusion Detection Architecture}. In \bibinfo{booktitle}{\emph{Proceedings of the 18th {ACM} conference on Computer and communications security}}. \bibinfo{publisher}{{ACM}}.
\newblock
\href{https://doi.org/10.1145/2046707.2046741}{doi:\nolinkurl{10.1145/2046707.2046741}}


\bibitem[Wang et~al\mbox{.}(2014)]%
        {Wang_2014}
\bibfield{author}{\bibinfo{person}{Kai Wang}, \bibinfo{person}{Zhe Fu}, \bibinfo{person}{Xiaohe Hu}, {and} \bibinfo{person}{Jun Li}.} \bibinfo{year}{2014}\natexlab{}.
\newblock \showarticletitle{Practical regular expression matching free of scalability and performance barriers}.
\newblock \bibinfo{journal}{\emph{Computer Communications}}  \bibinfo{volume}{54} (\bibinfo{date}{dec} \bibinfo{year}{2014}), \bibinfo{pages}{97--119}.
\newblock
\href{https://doi.org/10.1016/j.comcom.2014.08.005}{doi:\nolinkurl{10.1016/j.comcom.2014.08.005}}


\bibitem[Wang et~al\mbox{.}(2023)]%
        {Wang_2023}
\bibfield{author}{\bibinfo{person}{Shicheng Wang}, \bibinfo{person}{Menghao Zhang}, \bibinfo{person}{Guanyu Li}, \bibinfo{person}{Chang Liu}, \bibinfo{person}{Zhiliang Wang}, \bibinfo{person}{Ying Liu}, {and} \bibinfo{person}{Mingwei Xu}.} \bibinfo{year}{2023}\natexlab{}.
\newblock \showarticletitle{Bolt: Scalable and Cost-Efficient Multistring Pattern Matching With Programmable Switches}.
\newblock \bibinfo{journal}{\emph{{IEEE}/{ACM} Transactions on Networking}} \bibinfo{volume}{31}, \bibinfo{number}{2} (\bibinfo{date}{apr} \bibinfo{year}{2023}), \bibinfo{pages}{846--861}.
\newblock
\href{https://doi.org/10.1109/tnet.2022.3202523}{doi:\nolinkurl{10.1109/tnet.2022.3202523}}


\bibitem[Wang et~al\mbox{.}(2019)]%
        {wang2019hyperscan}
\bibfield{author}{\bibinfo{person}{Xiang Wang}, \bibinfo{person}{Yang Hong}, \bibinfo{person}{Harry Chang}, \bibinfo{person}{KyoungSoo Park}, \bibinfo{person}{Geoff Langdale}, \bibinfo{person}{Jiayu Hu}, {and} \bibinfo{person}{Heqing Zhu}.} \bibinfo{year}{2019}\natexlab{}.
\newblock \showarticletitle{Hyperscan: A Fast Multi-pattern Regex Matcher for Modern {CPUs}}. In \bibinfo{booktitle}{\emph{16th USENIX Symposium on Networked Systems Design and Implementation (NSDI 19)}}. \bibinfo{publisher}{USENIX Association}, \bibinfo{address}{Boston, MA}, \bibinfo{pages}{631--648}.
\newblock
\showISBNx{978-1-931971-49-2}
\urldef\tempurl%
\url{https://www.usenix.org/conference/nsdi19/presentation/wang-xiang}
\showURL{%
\tempurl}


\bibitem[Weigert et~al\mbox{.}(2014)]%
        {Weigert_2014}
\bibfield{author}{\bibinfo{person}{Stefan Weigert}, \bibinfo{person}{Matti~A. Hiltunen}, {and} \bibinfo{person}{Christof Fetzer}.} \bibinfo{year}{2014}\natexlab{}.
\newblock \showarticletitle{Finding the Needle in the Haystack: Identifying Business Communities in Internet Traffic}. In \bibinfo{booktitle}{\emph{2014 {IEEE}/{WIC}/{ACM} International Joint Conferences on Web Intelligence ({WI}) and Intelligent Agent Technologies ({IAT})}}. \bibinfo{publisher}{{IEEE}}.
\newblock
\href{https://doi.org/10.1109/wi-iat.2014.31}{doi:\nolinkurl{10.1109/wi-iat.2014.31}}


\bibitem[Whitaker et~al\mbox{.}(2004)]%
        {Whitaker_2004}
\bibfield{author}{\bibinfo{person}{Andrew Whitaker}, \bibinfo{person}{Richard~S. Cox}, {and} \bibinfo{person}{Steven~D. Gribble}.} \bibinfo{year}{2004}\natexlab{}.
\newblock \showarticletitle{Configuration Debugging as Search: Finding the Needle in the Haystack}. In \bibinfo{booktitle}{\emph{Proceedings of the 6th Conference on Symposium on Operating Systems Design \& Implementation - Volume 6}} (San Francisco, CA) \emph{(\bibinfo{series}{OSDI'04})}. \bibinfo{publisher}{USENIX Association}, \bibinfo{address}{USA}, \bibinfo{pages}{6}.
\newblock


\bibitem[Williams and Zobel(2002)]%
        {Williams_2002}
\bibfield{author}{\bibinfo{person}{H.E. Williams} {and} \bibinfo{person}{J. Zobel}.} \bibinfo{year}{2002}\natexlab{}.
\newblock \showarticletitle{Indexing and retrieval for genomic databases}.
\newblock \bibinfo{journal}{\emph{{IEEE} Transactions on Knowledge and Data Engineering}} \bibinfo{volume}{14}, \bibinfo{number}{1} (\bibinfo{year}{2002}), \bibinfo{pages}{63--78}.
\newblock
\href{https://doi.org/10.1109/69.979973}{doi:\nolinkurl{10.1109/69.979973}}


\bibitem[Wong et~al\mbox{.}(1985)]%
        {Wong_1985}
\bibfield{author}{\bibinfo{person}{Harry K.~T. Wong}, \bibinfo{person}{Hsiu-Fen Liu}, \bibinfo{person}{Frank Olken}, \bibinfo{person}{Doron Rotem}, {and} \bibinfo{person}{Linda Wong}.} \bibinfo{year}{1985}\natexlab{}.
\newblock \showarticletitle{Bit Transposed Files}. In \bibinfo{booktitle}{\emph{Proceedings of the 11th International Conference on Very Large Data Bases - Volume 11}} (Stockholm, Sweden) \emph{(\bibinfo{series}{VLDB '85})}. \bibinfo{publisher}{VLDB Endowment}, \bibinfo{pages}{448–457}.
\newblock


\bibitem[Yang et~al\mbox{.}(2016)]%
        {Yang_2016}
\bibfield{author}{\bibinfo{person}{Jiajia Yang}, \bibinfo{person}{Lei Jiang}, \bibinfo{person}{Qiu Tang}, \bibinfo{person}{Qiong Dai}, {and} \bibinfo{person}{Jianlong Tan}.} \bibinfo{year}{2016}\natexlab{}.
\newblock \showarticletitle{{PiDFA}: A Practical Multi-Stride Regular Expression Matching Engine Based on {FPGA}}. In \bibinfo{booktitle}{\emph{2016 {IEEE} International Conference on Communications ({ICC})}}. \bibinfo{publisher}{{IEEE}}.
\newblock
\href{https://doi.org/10.1109/icc.2016.7511199}{doi:\nolinkurl{10.1109/icc.2016.7511199}}


\bibitem[Yang et~al\mbox{.}(2011)]%
        {Yang_2011}
\bibfield{author}{\bibinfo{person}{Liu Yang}, \bibinfo{person}{Rezwana Karim}, \bibinfo{person}{Vinod Ganapathy}, {and} \bibinfo{person}{Randy Smith}.} \bibinfo{year}{2011}\natexlab{}.
\newblock \showarticletitle{Fast, memory-efficient regular expression matching with {NFA}-{OBDDs}}.
\newblock \bibinfo{journal}{\emph{Computer Networks}} \bibinfo{volume}{55}, \bibinfo{number}{15} (\bibinfo{date}{oct} \bibinfo{year}{2011}), \bibinfo{pages}{3376--3393}.
\newblock
\href{https://doi.org/10.1016/j.comnet.2011.07.002}{doi:\nolinkurl{10.1016/j.comnet.2011.07.002}}


\bibitem[Yang et~al\mbox{.}(2008)]%
        {Yang_2008}
\bibfield{author}{\bibinfo{person}{Yi-Hua~E. Yang}, \bibinfo{person}{Weirong Jiang}, {and} \bibinfo{person}{Viktor~K. Prasanna}.} \bibinfo{year}{2008}\natexlab{}.
\newblock \showarticletitle{Compact architecture for high-throughput regular expression matching on {FPGA}}. In \bibinfo{booktitle}{\emph{Proceedings of the 4th {ACM}/{IEEE} Symposium on Architectures for Networking and Communications Systems}}. \bibinfo{publisher}{{ACM}}.
\newblock
\href{https://doi.org/10.1145/1477942.1477948}{doi:\nolinkurl{10.1145/1477942.1477948}}


\bibitem[Yang and Prasanna(2011)]%
        {Yang_Prasanna_2011}
\bibfield{author}{\bibinfo{person}{Yi-Hua~E. Yang} {and} \bibinfo{person}{Viktor~K. Prasanna}.} \bibinfo{year}{2011}\natexlab{}.
\newblock \showarticletitle{Space-time tradeoff in regular expression matching with semi-deterministic finite automata}. In \bibinfo{booktitle}{\emph{2011 Proceedings {IEEE} {INFOCOM}}}. \bibinfo{publisher}{{IEEE}}.
\newblock
\href{https://doi.org/10.1109/infcom.2011.5934986}{doi:\nolinkurl{10.1109/infcom.2011.5934986}}


\bibitem[Yu et~al\mbox{.}(2006)]%
        {Yu_2006}
\bibfield{author}{\bibinfo{person}{Fang Yu}, \bibinfo{person}{Zhifeng Chen}, \bibinfo{person}{Yanlei Diao}, \bibinfo{person}{T.~V. Lakshman}, {and} \bibinfo{person}{Randy~H. Katz}.} \bibinfo{year}{2006}\natexlab{}.
\newblock \showarticletitle{Fast and memory-efficient regular expression matching for deep packet inspection}. In \bibinfo{booktitle}{\emph{Proceedings of the 2006 {ACM}/{IEEE} symposium on Architecture for networking and communications systems}}. \bibinfo{publisher}{{ACM}}.
\newblock
\href{https://doi.org/10.1145/1185347.1185360}{doi:\nolinkurl{10.1145/1185347.1185360}}


\bibitem[Yu et~al\mbox{.}(2019)]%
        {Yu_2019}
\bibfield{author}{\bibinfo{person}{Han Yu}, \bibinfo{person}{Aiping Li}, {and} \bibinfo{person}{Rong Jiang}.} \bibinfo{year}{2019}\natexlab{}.
\newblock \showarticletitle{Needle in a Haystack: Attack Detection from Large-Scale System Audit}. In \bibinfo{booktitle}{\emph{2019 {IEEE} 19th International Conference on Communication Technology ({ICCT})}}. \bibinfo{publisher}{{IEEE}}.
\newblock
\href{https://doi.org/10.1109/icct46805.2019.8947201}{doi:\nolinkurl{10.1109/icct46805.2019.8947201}}


\bibitem[Zhang et~al\mbox{.}(2023)]%
        {blare}
\bibfield{author}{\bibinfo{person}{Ling Zhang}, \bibinfo{person}{Shaleen Deep}, \bibinfo{person}{Avrilia Floratou}, \bibinfo{person}{Anja Gruenheid}, \bibinfo{person}{Jignesh~M. Patel}, {and} \bibinfo{person}{Yiwen Zhu}.} \bibinfo{year}{2023}\natexlab{}.
\newblock \showarticletitle{Exploiting Structure in Regular Expression Queries}.
\newblock \bibinfo{journal}{\emph{Proceedings of the {ACM} on Management of Data}} \bibinfo{volume}{1}, \bibinfo{number}{2} (\bibinfo{date}{jun} \bibinfo{year}{2023}), \bibinfo{pages}{1--28}.
\newblock
\href{https://doi.org/10.1145/3589297}{doi:\nolinkurl{10.1145/3589297}}


\bibitem[Zhang et~al\mbox{.}(2025)]%
        {ngram_select_comp}
\bibfield{author}{\bibinfo{person}{Ling Zhang}, \bibinfo{person}{Shaleen Deep}, \bibinfo{person}{Jignesh~M. Patel}, {and} \bibinfo{person}{Karthikeyan Sankaralingam}.} \bibinfo{year}{2025}\natexlab{}.
\newblock \bibinfo{title}{An Evaluation of N-Gram Selection Strategies for Regular Expression Indexing in Contemporary Text Analysis Tasks}.
\newblock
\showeprint[arxiv]{2504.12251}~[cs.DB]
\urldef\tempurl%
\url{https://arxiv.org/abs/2504.12251}
\showURL{%
\tempurl}


\bibitem[Zobel et~al\mbox{.}(1998)]%
        {Zobel_1998}
\bibfield{author}{\bibinfo{person}{Justin Zobel}, \bibinfo{person}{Alistair Moffat}, {and} \bibinfo{person}{Kotagiri Ramamohanarao}.} \bibinfo{year}{1998}\natexlab{}.
\newblock \showarticletitle{Inverted files versus signature files for text indexing}.
\newblock \bibinfo{journal}{\emph{ACM Transactions on Database Systems}} \bibinfo{volume}{23}, \bibinfo{number}{4} (\bibinfo{date}{Dec.} \bibinfo{year}{1998}), \bibinfo{pages}{453–490}.
\newblock
\showISSN{1557-4644}
\href{https://doi.org/10.1145/296854.277632}{doi:\nolinkurl{10.1145/296854.277632}}


\end{thebibliography}


\appendix

\onecolumn

\section{Correctness Proof} \label{app:correct}

\begin{lemma}
    Algorithm~\ref{algo:index_query} correctly generates the output of a regex for the given log $L$.
\end{lemma}
\begin{proof}
    Consider some log line $\ell \in L$. First, note that the bit index for $\ell$ stores a set bit for all bigrams $b \in G$ if $b$ is a substring of length $2$ in $\ell$. To show the correctness of Algorithm~\ref{algo:index_query}, it suffices to argue that the bit filtering on line~\ref{line:filter} does not remove any log line from consideration that may contain a potential regex match. We consider two cases:

    \noindent \introparagraph{Regex has no $n$-grams} In this case, the \textsf{Set-BitVector} on line~\ref{line:setbit} returns a bit-vector with all bits as $0$ and line~\ref{line:flip} will set $m$ to bit-vector with all bits as $1$. Thus, the check on line~\ref{line:filter} will be evaluated to true and the regex will be evaluated on the log line under consideration.

    \noindent \introparagraph{Regex has $n$-grams} Suppose $g_1, \dots, g_j$ are the $j$ $n$-grams present in the regex and also as keys in the {inverted index\eat{hashmap}} $G$.  Let us consider gram $g_1$. \textsf{Set-BitVector} will set the bit for $g_1$ as $1$ and line~\ref{line:flip} will set it to $0$. Next, for the check on line~\ref{line:filter}, the output bit corresponding to $g_1$ will be $1$ if and only if the bit in $I[j]$ for $g_1$ is set to $1$ (i.e. $g_1$ is a $n$-gram that is present in the log line $L_j$). In other words, if $I[j]$ for $g_1$ is $0$, we correctly do not process the log line. Otherwise, the output bit is correctly set to $1$. The same reasoning applies to all other $n$-grams as well. Thus, the check on line~\ref{line:filter} is passed if and only if each of the $n$-grams of the regex are also present in the index $I[j]$.
\end{proof}

\section{Advantages of Bit-Vector Index for other settings}

With a variety of storage and operational benefits, the bit-vector index appears as a particularly effective tool in the context of streaming log management systems. Local indexing of logs during ingestion eliminates the need for rigorous concurrent index updates, speeding up the ingestion process. When queries are run on a portion of the log collection, the local indexing, especially at a line granularity, is found to be helpful because it effectively reduces querying overhead by proactively eliminating out-of-range logs from the index. Furthermore, in the context of a time-series database of logs, where entries adhere to a Time-to-Live (TTL) rule and older logs are cleaned up regularly, the bit-vector index ensures that the purging of outdated log lines is executed with minimal additional effort. This local, granular indexing strategy also exhibits its advantages in distributed storage settings for log data storage spanning multiple instances, nodes, or regions. Local indices reduce the cost and latency issues of retrieving remote indexes. Although it is possible to build an inverted index for each local instance, the bit-vector index stands out for its flexibility, requiring no specialized logic or additional considerations during its construction and usage. The bit-vector index is flexible during data partitioning or re-partitioning. By simply slicing the index along with the data, as opposed to other indexing methods, it avoids the significant overhead of rebuilding smaller local indices, as seen with inverted index partitions.

\begin{figure}[!tp]
\begin{subfigure}{0.45\textwidth}
  \centering
  \includegraphics[width=0.9\columnwidth]{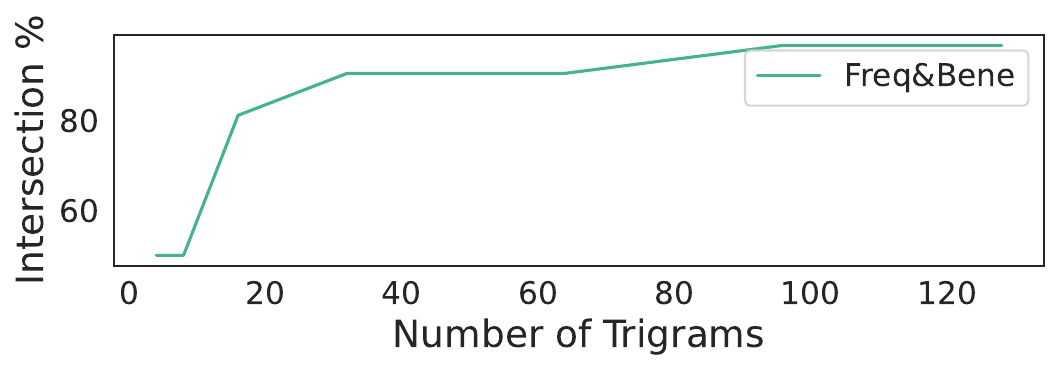}  
  \vspace{-0.5em}
  \caption{Trigram intersection as a percentage from strategy pairs.}
  \label{fig:microbench_intersect_tri}
\end{subfigure}
\begin{subfigure}{0.45\textwidth}
  \centering
  \includegraphics[width=0.9\columnwidth]{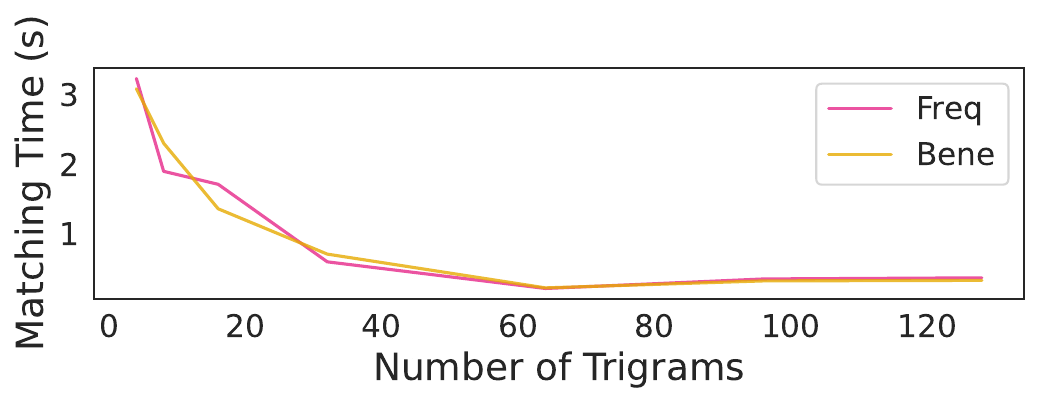} 
  \vspace{-0.5em}
  \caption{Regex Matching Time.}
  \label{fig:microbench_time_trigram}
\end{subfigure}
\caption{Varying the number of trigrams, compare the set of trigrams selected by the Frequency and Benefit methods and the matching time applying their resulting indices. The experimental setting is the same as described in Section~\ref{subsec:gram_select}}
\end{figure}

\begin{table}[t]
\captionof{table}{For a fixed $k$ and selection method, the percentage of top-$k$ bigrams that are substring of at least one trigram in the top-$k$ trigrams.}
\vspace{-1em}
\centering
\begin{tabular}{ c| | r | r | r | r | r | r | r }\hline
method$(\downarrow) ; k (\rightarrow)$ & 4 & 8 & 16 & 32 & 64 & 96 & 128  \\ \hline
 Freq & 0.25 & 0.62 & 0.75 & 0.78 & 0.83 & 0.79 & 0.77  \\ 
Bene & 0.50 & 0.38 & 0.56 & 0.81 & 0.86 & 0.79 & 0.74 \\ \hline
\end{tabular}
\label{tab:bigram_in_trigram}
\end{table}

\section{Discussion on Bigrams vs. Trigrams}\label{app:bigram_trigram}

In this section, we present more details on why bigrams are found to be as effective compared to using trigrams which are the multigram of choice for indexing related applications. For a given a multigram $g$, we define the selectivity of $g$ as the fraction of log lines in log $L$ that \emph{do not} contain $g$ as a substring. We begin by stating a simple fact.

\begin{observation} \label{prop:one}
    Suppose a bigram $g$ has selectivity $\sigma$. Then, any string literal $s$ containing $g$ has selectivity at most $\sigma$.
\end{observation}
\begin{proof}
    Consider the set of bigrams generated by $s$ (which includes $g$). The number of log lines that contains $g$ is $\sigma \cdot |L|$. Clearly, by adding the additional requirement that all other bigrams of $s$ be also present in the log lines that contain $g$, the number of such log lines can only decrease from $\sigma \cdot |L|$ which implies that the selectivity of $s$ is at most $\sigma$. 
\end{proof}

Observation~\ref{prop:one} tells us that if a bigram has high selectivity, then there is likely not much benefit to choosing a bigger multigram. Obviously, a bigger multigram may have lower selectivity than the bigram but as we saw in Section~\ref{subsec:gram_select}, the cost of finding such multigrams can be prohibitive when using benefit or incremental benefit based approach. 
{Similar to bigrams as we discussed in~\Cref{subsec:gram_select}, we demonstrate that the top trigrams selected by frequency and by individual benefit overlaps to a large extent as shown in~\Cref{fig:microbench_intersect_tri}. Also, the selection results produce indices with similar performance using the same $k$ for both methods, as presented in~\Cref{fig:microbench_time_trigram}. }
Further, a bigger multigram is also likely to appear in fewer regexes in the workload, making its overall utility for the entire workload lower. 
{Therefore, we draw the relationship between the set of bigrams and the set of trigrams selected using the same method (Freq or Bene) and $k$.}
\Cref{tab:bigram_in_trigram} shows the the percentage of number of bigrams, which are substring of one or more trigrams in the top-$k$ trigrams, among the top-$k$ bigrams.
For both methods, when $k \geq 32$, {the overlap percentage is large, steadily within the range of 74\% and 86\%.}

On the other hand, suppose that we fix $n=3$ and there is a selective trigram $s$ that is picked by the frequent strategy from Section~\ref{subsec:gram_select}. Then, it is immediate to make the following observation.

\begin{observation} \label{prop:two}
    Consider a trigram $s = \alpha \beta \gamma$ with selectivity $\sigma$ that is picked based on its frequency in the query workload. Then, the bigrams $\alpha \beta$ and $\beta \gamma$ together have selectivity $\sigma$.
\end{observation}

A trigram that appears (say) $c$ times in the workload implies that its bigrams also appear at least $c$ times (apriori property). Thus, there is a good chance that both of the bigrams are picked, which would ensure that the high selectivity property of $s$ is utilized. In the worst case, it could be that both $\alpha \beta$ and $\beta \gamma$ have low selectivity but $\alpha \gamma$ has high selectivity. Assuming the characters of the trigram are picked in an iid fashion and the distribution is zipfian with parameter $a \geq 0$ (i.e. the character with rank $k$ occurs with probability $\frac{1}{H_{a,d}} \cdot \frac{1}{k^a}$ where $H_{a,d} = \sum_{i = 1}^{d} \frac{1}{i^a}$). The probability that $\alpha \gamma$ has a higher selectivity than both $\alpha \beta$ and $\beta \gamma$ can be expressed as $\sum_{i=1}^{d-1} \frac{1}{H_{a,d}} \cdot \frac{1}{i^a} \cdot \big(\sum_{j = i+1}^{d} \frac{1}{H_{a,d}} \cdot \frac{1}{j^a} \big)^2 $. Picking $d=26$ and $a=1.1$ for english language distribution, we get the probability as $0.24$. Therefore, even considering as few as $5$ trigrams is enough to ensure that with probability greater than $0.99$, at least one selective bigram will be present in the bigram decomposition of the trigram string literal.

\section{Case Study: Index Configuration under Size Constraints} \label{app:rei_tuner}

{

In practical log indexing scenarios, an essential challenge is determining an appropriate number of bigrams $k$ and corresponding index granularity $m$ within user-specified space constraints. As the optimal selection of these parameters involves a combinational search over a large set of candidates and actually running the workloads, comparing exact performance is computationally prohibitive for typical workloads. We propose extending \sys{} with an efficient heuristic approach for auto-configuring both $k$ and $m$, guided by an approximate estimate of filtering effectiveness.

\subsection{Heuristic Configuration Scoring}

We approximate the overall filtering effectiveness provided by indexing the top $k$ bigrams with index granularity $m$ by considering the distribution of bigrams within both the query set $Q$ and the log dataset $L$. Given the granularity $m$, we partition the dataset $L$ into $N_{\text{blocks}} = \lceil |L|/m \rceil$ blocks, each consisting of $m$ consecutive log lines (with the last block potentially smaller). For each bigram $g$, we define the fraction of log blocks that contain at least one occurrence of $g$:
$$
p_{L,m}(g) = \frac{|\{B \in \mathcal{B}_m(L) \mid g \text{ appears in at least one line of } B\}|}{N_{\text{blocks}}}
$$
where $\mathcal{B}_m(L)$ is the set of all blocks derived from partitioning $L$ at granularity $m$. Using $p_Q(g)$ to denote the fraction of queries containing bigram $g$, the probability of filtering a randomly selected query–log-block pair with bigram $g$ is approximated as:
$$
p_Q(g) = \frac{\text{freq}(g)}{|Q|} \qquad p_{Q}(g) \cdot (1 - p_{L,m}(g))
$$

Consequently, the approximate cumulative filtering effectiveness of indexing the top $k$ bigrams ($I_k$) at granularity $m$ is expressed by:
$$
\widehat{\text{score}}(I_{k,m}) = |Q| \cdot N_{\text{blocks}} \cdot \left(1 - \prod_{g \in I_k}\left[1 - p_Q(g)\cdot(1 - p_{L,m}(g))\right]\right)
$$
Intuitively, the quantity $p_Q(g)\cdot(1 - p_{L, m}(g))$ represents the probability that a randomly selected query–log block pair would be filtered by bigram $g$. By multiplying these individual probabilities for all top-$k$ bigrams, we approximate the cumulative filtering probability of indexing at least one of these $k$ bigrams. This approach inherently accounts for potential overlaps in bigram coverage and yields a fast but informative heuristic estimate of overall filtering effectiveness.

Using this approximation, we perform an efficient search over candidate values of $k$ and index granularities $m$. Specifically, for each candidate granularity $m$, the total index size in bits is explicitly computed as:
The total index size in bits for a candidate pair $(k,m)$ is computed as:
$$
S_{\text{index}}(k,m) = k \times N_{\text{blocks}}
$$
Under a fixed user-defined storage constraint $S$, we efficiently perform a binary search to identify the largest $k$ satisfying:
$$
S_{\text{index}}(k,m) \leq S
$$
We then select the optimal configuration $(k,m)$ that maximizes the filtering effectiveness per bit, given by:
$$
\eta(k,m) = \frac{\widehat{\text{score}}(I_k,m)}{S_{\text{index}}(k,m)}
$$

\subsection{Experiment Demonstration}
We did a simple demonstration of the configuration suggestion with the above heuristic on the \wlsqlfull{} workload.
The experiment searches over candidate $k$ values for each granularity candidate $m \in \{1,8,32,64,128,256,512\}$ and identify the parameter combination that maximizes estimated filtering effectiveness per bit under the $256$ MB storage constraint.
We implement the above method in C++, and the entire heuristic evaluation completed in approximately 80 seconds of computation on a single-threaded setup.

The heuristic recommended two primary configurations as most effective under the given constraints, according to the estimate scoring functions. The top-ranked configuration recommended indexing $k=83$ bigrams at granularity $m=1$. A closely competitive second-best configuration recommended indexing 84 bigrams at granularity $m=8$.
These heuristic recommendations closely align with the manually identified best configuration among the configurations tested in~\Cref{subsec:qgram_num} and ~\Cref{subsec:idx_granu}, which are pairs $k=64$ with $m=1$ and $k=128$ with $m=8$. 

We use the results to demonstrate the practical effectiveness of our heuristic approach for automatic index parameter selection. It is computational efficient, while presenting estimates closely comparable to the manually optimized configurations identified through substantially more resource-intensive exploration through actual workload executions.

It is important to emphasize that the presented heuristic implementation serves primarily as an experimental prototype aimed at demonstrating the practicality of our approach. In the future, we can explore detailed optimization of this heuristic method using advanced techniques such as machine learning or adaptive sampling to further improve both its runtime performance and accuracy.
}

\eat{
\TODO{Maybe keep only the trimmed part in main paper.}
The choice of $n$-grams for the index is crucial to the effectiveness of the indexing framework when matching regular expressions over large log datasets. It is influenced by four components: workload of queries $Q$, dataset $L$, the choice(s) of $n$, and the number of chosen $n$-grams ($k$). Prior works have made mixed choices of $n$ or have chosen a fixed sized $n$. We made the choice of only using bigrams (i.e. $n=2$).
\eat{. Some works choose to only use bigrams (i.e. $n=2$). }This choice is motivated by the observation that $n=2$ finds the right balance between the space of possible $n$-grams that need to be considered and the cost of querying the index. We find this choice to be apt for our setting as well and demonstrate it experimentally.
\eat{{Some uses trigrams like Postgres. } }

\eat{{Note: I removed percentage benefit and percentage incremental benefit as we are not using them. What we have is more or less enough for our purpose. The percentage are calculated primarily for comparing benefit of a bigram g1 and the incremental benefit of g1 conditional on some other g2, and is not defined in BEST anyway.} \sd{Here, we should highlight that the purpose of doing all the benefit analysis is to show the reader that BEST and FREE method of choosing n-grams is not a good idea (R2 was confused on why we talk about all this).}}
Next, we compare against the strategy of choosing multigrams as introduced by the \textsf{BEST} and \textsf{FREE} algorithms. We recall the definition of $n$-gram \textit{benefit} (\textit{bene}) defined in BEST~\cite{BEST}. Let $L$ be the log dataset, $g$ be some $n$-gram, $W$ be the workload with $L$ and $Q$ as dataset and query set respectively. Let $\mathnormal{grams}(q)$ be the set of unique bigrams in $q$. Then $\mathnormal{bene}_L(g)$ is the number of log lines that can be filtered out by an $n$-gram $g$. Its \textit{benefit} on the entire workload are therefore
$$\mathnormal{bene}_W(g) = \sum_{q\in Q}\mathnormal{bene}_L(g) \cdot \mathbbm{1}(g \in \mathnormal{grams}(q))$$

\eat{$$\%\mathnormal{bene}_W(g) = \frac{\mathnormal{bene}_W(g)}{|L|\cdot|Q|}$$}
The existence of certain sets of $n$-grams may have a high correlation. In our example from the previous section, although \emph{vm} and \emph{mN} have high individual benefit on the workload, string \emph{vmName} would make them occur together frequently, and therefore indexing on both \emph{mN} and \emph{vm} is unlikely to be significantly more helpful compared to indexing on \emph{vm} alone. We use \textit{incremental benefit} (\textit{incr\_bene}) to represent the pairwise effectiveness of bigrams.
Let $g_1$ and $g_2$ be two grams and $L_{g_1}$ be the subset of log containing lines with the bigram $g_1$, then the \textit{incremental benefit} of $g_2$ given $g_1$, $\mathnormal{incr\_bene}_{L_{g_1}}(g_2 | g_1)$ is the number of logs without $g_2$ in $L_{g_1}$. The \textit{incremental benefit}
\eat{ and \textit{percentage incremental benefit} }
on the entire workload is:
\begin{align*}
    & \mathnormal{incr\_bene}_W(g_2|g_1) \\
    = & \sum_{q\in Q} \Big(\mathnormal{incr\_bene}_{L_{g_1}}(g_2|g_1) \cdot \mathbbm{1}(g_1, g_2\in \mathnormal{grams}(q))\\
     & + \mathnormal{bene}_{L}(g_2) \cdot \mathbbm{1}(g_1  \notin \mathnormal{grams}(q))\land g_2 \in \mathnormal{grams}(q))\Big)
\end{align*}
\eat{
$$\%\mathnormal{incr\_bene}_W(g_2|g_1) = \frac{\mathnormal{incr\_bene}_W(g_2|g_1)}{\mathnormal{bene}_W(g_1) + |L|\cdot \# \text{queries without }g_1}$$
}
\begin{figure}[!tp]
\begin{subfigure}{0.95\columnwidth}
  \centering
  \includegraphics[width=0.8\columnwidth]{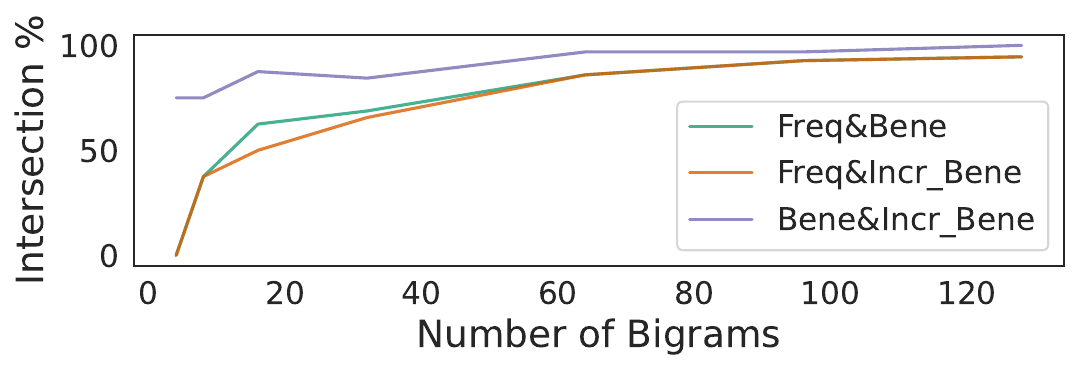}  
  \vspace{-0.5em}
  \caption{Bigram intersection  as a percentage from strategy pairs.}
  \label{fig:microbench_intersect}
\end{subfigure}
\\
\hspace{1.1em}
\begin{subfigure}{0.95\columnwidth}
  \centering
  \includegraphics[width=0.78\columnwidth]{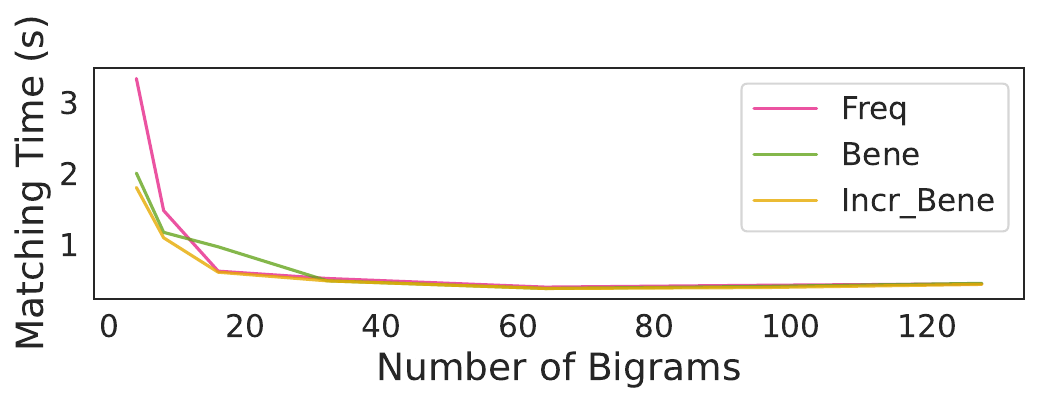}  
  \vspace{-0.5em}
  \caption{Regex Matching Time.}
  \label{fig:microbench_time}
\end{subfigure}
  \vspace{-1.5em}
\caption{Varying the number of bigrams, compare the set of bigrams selected by the three methods and the matching time applying their resulting indices. \eat{\sd{Legend should read Incr\_Bene?}}}
  \vspace{-1.5em}
\end{figure}

\textsf{BEST}~\cite{BEST}\eat{presented an algorithm that} finds $k$ $n$-grams that maximizes the effectiveness of the chosen $n$-grams in time $O(k \cdot |Q| \cdot |L|)$ and uses $O(|Q| \cdot |L|)$ space. Although their algorithm picks a set of $n$-grams with provable guarantees of their effectiveness (using \eat{the benefit as defined above}{\textit{incremental benefit} conditional on all other selected $n$-grams}), the running time complexity is prohibitive. The excessive time requirement could be mitigated by using parallel processing (in fact,~\cite{BEST} provides a parallel algorithm to reduce the execution time), but it does not reduce the overall computation required. 

We propose an $n$-gram selection strategy that runs in time $O(|Q|)$ {by selecting based on $n$-gram frequency in the query workloads}. 
In order to show the cost and effectiveness of $n$-gram selection strategies, we run a micro-benchmark that builds and uses indexes selected with 3 different strategies {with low index overhead}: (1) the most frequently occurring bigrams (Freq), (2) bigrams with the highest benefit (Bene), and (3) bigrams selected based on incremental benefit (Incr\_Bene). We run the three methods on a subset of a real-world workload containing one million log lines and 132 queries.
When selecting according to the incremental benefit, we first select the bigram with the highest individual benefit as $g_1$, then select $g_2$ with the highest incremental benefit conditional on $g_1$; we iteratively select $g_3$ conditional on $g_2$ until we reach the desired number of bigrams\footnote{We intentionally omit conditioning on $g_1, \dots, g_{i-1}$ when selecting $g_i$ since the cost to record the benefit for all subsets of already chosen bigrams is prohibitively large.}. When no bigrams generate a positive incremental benefit for some $g_i$, we select the bigram with the highest benefit from the remaining bigrams.

We look into how these three approaches overlap in their selected bigrams when selecting the top 4, 8, 16, 32, 64, 96, and 128 bigrams and present the percentage overlapping in~\Cref{fig:microbench_intersect}. 
There is a notable similarity in bigrams selected by Bene and by Incr\_Bene. The percentage of common bigrams is high at 75\% with only 4 and 8 bigrams selected and steadily increases to 96.9\% for 96 bigrams and 100\% for 128 bigrams. We observed that conditional on the top-100 selected bigrams, no other remaining bigrams can generate a positive incremental benefit. As the number of chosen bigrams increases from 4 to 64, the intersection percentage of bigrams selected by Freq with the other two methods gradually rises to 85.9\%. Their rate of intersection plateaus as the number of bigrams selected continues to increase. 

Further, we evaluate the performance of matching using indexes derived from the three methods and present the results in~\Cref{fig:microbench_time}. 
The matching time is quite different when the number of bigrams selected is small, where using the index built with bigrams selected by Freq gives the slowest matching time. The times gradually converge for indices built with the three methods as the number of bigrams increases. The running time reaches a minimum when indexing with 64 bigrams, where they are 0.395, 0.376, and 0.376 for Freq, Bene, and Incr\_Bene respectively.

We also evaluate the time needed to choose bigrams for each approach.  
The ranking calculations for frequency and benefit are finished in 0.0014 seconds and 21.5 seconds, respectively for selecting 64 bigrams. In sharp contrast, the incremental benefit technique requires hours (31,638.2 seconds) when only taking into account the bigrams with top-200 individual benefit. The computational time increases even more when the incremental benefit approach is applied to the entire dataset, taking well over $24$ hours to complete\footnote{We terminate the experiments at $24$ hours.}. When choosing bigram selection strategies for indexing, there are trade-offs between filtering power of selection result and computational efficiency.

Based on this microbenchmark, we find that choosing bigrams based on frequency is the best strategy for log analysis workloads. Despite having a minor performance advantage, the incremental benefit method and methods that find the optimal bigrams conditional on all other selected bigrams would require excessive processing resources, making marginal gains insignificant compared to frequency-based selection.

\eat{
\begin{table}[t]
\centering
\begin{tabular}{ c |c| | r | r | r | r | r | r | r  }\hline
& \# Bigram & 4 & 8 & 16 & 32 & 64 & 96 & 128 \\ \hline
\multirow{3}{*}{Intersection}&Freq\&Bene. & e\textvisiblespace & n\textvisiblespace{} & in & at & ed & d\textvisiblespace & on  \\\hline
&Freq\&Bene. & e\textvisiblespace & n\textvisiblespace{} & in & at & ed & d\textvisiblespace & on  \\\hline
&Freq\&Bene. & e\textvisiblespace & n\textvisiblespace{} & in & at & ed & d\textvisiblespace & on  \\\hline
\end{tabular}
\captionof{table}{Top 10 most frequently occurring bigrams in the log data and the queries of the same real-world workload.}
\label{tab:microbench}
\vspace{-2em}
\end{table}

\begin{table}[t]
\centering
\begin{tabular}{ c |c| | r | r | r | r | r | r | r  }\hline
 \multirow{2}{*}{\# Bigram} & \multicolumn{3}{c}{\% of Intersection}             & \multirow{2}{*}{Match Time(s)}\\\cline{2-4}
                             & Freq\&Bene. & Freq\&Incr.-Bene. & Bene.\&Incr.-Bene. & \\ \hline
 4 &                         & 50           & 0                & 25                 & \\ \hline
 8 &                         & 50           & 25                & 25                 & \\ \hline
 16 &                        & 75          & 31.3            & 31.3                 & \\ \hline
 32 &                        & 87.5        & 37.5             & 37.5                 & \\ \hline
 64 &                         & 95.3           & 31.3                & 31.3                 & \\ \hline
 96 &                         & 97.9           & 29.2                & 25                 & \\ \hline
 128 &
\end{tabular}
\captionof{table}{Top 10 most frequently occurring bigrams in the log data and the queries of the same real-world workload.}
\label{tab:strategy_microbench}
\vspace{-2em}
\end{table}
}

The choice of $k$ is more complicated. Let us first assume that we magically know the right value of $k$. In particular, as demonstrated by the microbenchmark above, we find that it is sufficient to find the $k$ most frequently occurring $n$-grams and use them for index construction. The intuition here is that regex query workloads are mostly finding \emph{needle-in-a-haystack} style queries, a property that has been noted by several prior works~\cite{Whitaker_2004, Weigert_2014, Ellis_2015, Yu_2019}. In other words, the selectivity of the queries is very low. This low selectivity is further attributed to the string literals that occur in the query. Therefore, picking $n$-grams based on only the string literals in the workload is sufficient. The value of $k$ also involves a trade-off. While increasing the number of indexed grams might seem beneficial, there is a point of diminishing returns. {We discuss general guidelines to select $k$ in practice in~\Cref{subsubsec:num_ngram}.}

{We also run the microbenchmark on trigrams, and the results are included \biVsTri{} of the full paper. For log analysis tasks, trigrams usually have slightly higher filtering power than bigrams but are present less frequently in the queries. Therefore, there is no incentive to pick trigrams over bigrams since bigrams have enough selectivity and provides performance advantages over $n$-grams with $n \geq 3$.}

\eat{These observation not only highlight the need to maintain a careful balance so that the number of indexed grams does not endanger performance by being too sparse or cause unneeded space overhead by being too abundant, but it also highlights the value of our framework's adaptability. The system can be tuned to an ideal number of $n$-grams indexed that balances storage expense with performance improvement by users, allowing users to make wise decisions. 
}

\eat{Our method for choosing $n$-grams is based on a thorough comprehension of the characteristics of both the queries and the log dataset. 
Our primary focus for gram selection is the examination of queries. While it might seem intuitive to incorporate the log dataset in this selection process, such an approach is less efficient due to the size of the dataset. 

Using the most frequent $n$-grams have been shown to be effective without the need for information from the dataset itself. This strategy guarantees that the chosen $n$-grams can serve as filters for majority of the queries, reducing the number of log lines needed for the more expensive regular expression matching, and therefore improving the performance of the regex matching on the whole dataset.

Moreover, the most frequently occurring $n$-grams the log dataset often those that appears mostly in the content we aim to extract using regex queries, for example, sequences of numbers that are common in machine ID's.
\Cref{tab:top_bigram} shows the bar plots of the frequencies of the top-10 most frequently occurring bigrams in the log dataset and the regex queries on the dataset respectively. We observed  that the most frequent bigrams in the log dataset may consist of numbers and symbols, such as "00", "0" and "02". These frequently occurring bigrams, rarely appear in regexes. Therefore, indexing significant $n$-grams in the log dataset is counterproductive.

There has been studies on the types of grams used for indexing, options being using $n$-grams of sizes within a range, using arbitrary sized grams, and using one fixed length of grams. We opt for a fixed length of $n$-gram. The rationale behind this choice is the reduction of overhead in gram extraction from the log dataset when building the index. 
Given the substantial size of our dataset, any cost associated with gram extraction becomes significant. {Also, only a small subset of the log data may gets queried for individual tasks, and log lines may not be queried in its lifetime. Therefore, we want the indexing to be some cheap operation that serves as a "just-in-case" solution to potential incoming queries.} 

We choose to use bigrams for indexing for its ability to generalize. Even if derived from an imperfect subset of queries, the selected set of bigrams has a high chance that it can help most queries pre-filtering un-matched log lines. As the $n$ becomes larger, the $n$-grams becomes more specific to the specific queries, and less general to other queries. Remember, if none of the indexed $n$-gram appears in a regex query, we will need to run expensive regex matching for all log lines in the dataset. The emphasis on generality, which makes the chosen gram types and strategies suitable for a wider range of regexes and data, help us achieves substantial low index building overhead while achieving noticeable query performance improvement,

There is a trade-off when deciding how many $n$-grams to index. While increasing the number of indexed grams might seem beneficial, there's a point of diminishing returns. \TODO{Add motivating small experiment to demonstrate? Though it is already in the evaluation section}
This observation not only highlights the need to maintain a careful balance so that the number of indexed grams does not endanger performance by being too sparse or cause unneeded space overhead by being too abundant, but it also highlights the value of our framework's adaptability. The system can be tuned to an ideal number of $n$-grams indexed that balances storage expense with performance improvement by users wisely.}

}

\end{document}